\documentclass[sigconf]{acmart}
\AtBeginDocument{%
  \providecommand\BibTeX{{%
    \normalfont B\kern-0.5em{\scshape i\kern-0.25em b}\kern-0.8em\TeX}}}
    
\setcopyright{acmcopyright}
\copyrightyear{2024}
\acmYear{2024}
\acmDOI{XXXXXXX.XXXXXXX}

\acmConference[SIGMOD'24]{}{}{}
\acmPrice{}
\acmISBN{}

\usepackage{flushend}
\usepackage[a-1b]{pdfx}   
\usepackage{balance}
\usepackage{pifont}

\usepackage{color}
\usepackage[labelfont=bf]{caption}
\usepackage{subcaption}
\usepackage{hyperref}
\usepackage{makecell}
\usepackage{fixltx2e}
\usepackage{rotating,multirow}
\usepackage{etoolbox,xspace}
\usepackage{algorithmicx}
\usepackage{algorithm}
\usepackage{fontawesome}
\usepackage{algpseudocode}
\usepackage{wrapfig}
\algtext*{EndWhile}
\algtext*{EndIf}
\algtext*{EndFunction}
\algtext*{EndFor}
\algrenewcommand\algorithmicrequire{\textbf{Input:}}
\algrenewcommand\algorithmicensure{\textbf{Output:}}

\usepackage{amsthm}

\usepackage{tabularx}
\usepackage{ragged2e}  
\usepackage{booktabs}  
\usepackage{textcomp}
\usepackage{url}
\usepackage{comment}
\usepackage{enumitem}
\usepackage[simplified]{pgf-umlcd}
\usepackage{tikz}

\newcolumntype{Y}{>{\RaggedRight\arraybackslash}X}
\newtheorem{proposition}{Proposition}
\newtheorem{definition}{Definition}
\newtheorem{theorem}{Theorem}
\newtheorem{lemma}{Lemma}
\newtheorem{corollary}{Corollary}

\newcommand{\eop}{\hspace*{\fill}\mbox{$\Box$}}     
\newenvironment{proof}{\paragraph{Proof:}}{\hfill$\square$}
\newcounter{example}
\renewcommand{\theexample}{\arabic{example}}
\newenvironment{example}{
        \vspace{1ex}
        \refstepcounter{example}
        {\noindent\bf Example \theexample:}}
	{\eop\vspace{1ex}}

\newcommand{\squishlist}{
 \begin{list}{$\bullet$}
  { \setlength{\itemsep}{0pt}
     \setlength{\parsep}{1pt}
     \setlength{\topsep}{1pt}
     \setlength{\partopsep}{0pt}
     \setlength{\leftmargin}{1em}
     \setlength{\labelwidth}{1em}
     \setlength{\labelsep}{0.5em} } }

\newcommand{\squishend}{
  \end{list}
}

\definecolor{americanrose}{rgb}{1.0, 0.01, 0.24}
\definecolor{airforceblue}{rgb}{0.36, 0.54, 0.66}
\definecolor{ao(english)}{rgb}{0.0, 0.5, 0.0}
\definecolor{ao}{rgb}{0.0, 0.0, 1.0}

\newcommand{\nima}[1]{\textcolor{red}{Nima: #1}}

\newcommand{\eat}[1]{}

\usepackage{xspace}
\newcommand{\stitle}[1]{\vspace{2mm}\noindent{\bf #1:}}

\newcommand{\InputSet}{P}
\newcommand{\point}{p}
\newcommand{\GroupPoints}{\gee}
\newcommand{\hashmap}{\mathcal{H}}
\newcommand{\hashfunction}{h}
\newcommand{\bucket}{b}
\renewcommand{\vector}{w}
\newcommand{\Intervals}{\mathcal{I}}
\newcommand{\interval}{I}
\newcommand{\Arrang}{\mathcal{A}}

\newcommand{\Gee}{\mathcal{G}}
\newcommand{\gee}{\mathbf{g}}

\newcommand{\eps}{\varepsilon}

\newcommand{\at}[1]{{\tt \small #1}\xspace}

\newcommand{\fairhash}{\textsc{FairHash}\xspace}
\newcommand{\ranker}{\textsc{Ranking}\xspace}
\newcommand{\rtwoD}{\textsc{Ranking$_{2D}$}\xspace}
\newcommand{\rankerProb}{\textsc{Ranking}$^+$\xspace}
\newcommand{\pd}{\textsc{Sweep\&Cut}\xspace}
\newcommand{\necklaceb}{\textsc{Necklace$_{2g}$}\xspace}
\newcommand{\necklacek}{\textsc{Necklace$_{kg}$}\xspace}

\newcommand{\cdf}{\textsc{CDF}-based\xspace}
\newcommand{\fag}{\textsc{Fairness-agnostic}\xspace}

\newcommand{\Groups}{\Gee}
\newcommand{\group}{\gee}
\renewcommand{\Re}{\mathbb{R}}%
\def\mparagraph#1{\par\noindent\textbf{#1.}\quad}
\newcommand{\floor}[1]{\left \lfloor #1 \right \rfloor}
\newcommand{\poly}{\operatorname{poly}}
\newcommand{\polylog}{\operatorname{polylog}}
\newcommand{\optRank}{\eps_{R}}
\newcommand{\optDisc}{\eps_{D}}

\newcommand{\adult}{\textsc{Adult}\xspace}
\newcommand{\compas}{\textsc{Compas}\xspace}
\newcommand{\diabetes}{\textsc{Diabetes}\xspace}
\newcommand{\popsim}{\textsc{ChicagoPop}\xspace}

\DeclareMathOperator*{\argmin}{arg\,min}

\begin{document}

\title{\fairhash: A Fair and Memory/Time-efficient Hashmap}

\author{Nima Shahbazi}
\email{nshahb3@uic.edu}
\orcid{0000-0001-7016-3807}
\affiliation{%
  \institution{University of Illinois Chicago}
  \country{USA}
}

\author{Stavros Sintos}
\email{stavros@uic.edu}
\orcid{0000-0002-2114-8886}
\affiliation{%
  \institution{University of Illinois Chicago}
  \country{USA}
}

\author{Abolfazl Asudeh}
\email{asudeh@uic.edu}
\orcid{0000-0002-5251-6186}
\affiliation{%
  \institution{University of Illinois Chicago}
  \country{USA}
}

%
\renewcommand{\shortauthors}{Shahbazi, et al.}

\begin{abstract}
Hashmap is a fundamental data structure in computer science.
There has been extensive research on constructing hashmaps that minimize the number of collisions leading to efficient lookup query time.
Recently, the data-dependant approaches, construct hashmaps tailored for a target data distribution that guarantee to uniformly distribute data across different buckets and hence minimize the collisions.
Still, to the best of our knowledge, none of the existing technique guarantees group fairness among different groups of items stored in the hashmap. 

Therefore, in this paper, we introduce \fairhash, a data-dependant hashmap that guarantees uniform distribution at the group-level across hash buckets, and hence, satisfies the statistical parity notion of group fairness. We formally define, three notions of fairness and, unlike existing work, \fairhash satisfies all three of them simultaneously. 
We propose three families of algorithms to design fair hashmaps, suitable for different settings. Our ranking-based algorithms reduce the unfairness of data-dependant hashmaps without any memory-overhead. 
The cut-based algorithms guarantee zero-unfairness in all cases, irrespective of how the data is distributed, but those introduce an extra memory-overhead.  Last but not least, the discrepancy-based algorithms enable trading off between various fairness notions.
In addition to the theoretical analysis, we perform extensive experiments to evaluate the efficiency and efficacy of our algorithms on real datasets. Our results verify the superiority of \fairhash compared to the other baselines on fairness at almost no performance cost.
\end{abstract}




\maketitle

\section{Introduction}\label{sec: intro}

\subsection{Motivation}
As data-driven technologies become ingrained in our lives, their drawbacks and potential harms become increasingly evident~\cite{angwin2022machine, ntoutsi2020bias,stoyanovich2022responsible}.
Subsequently, algorithmic fairness has become central in computer science research to minimize machine bias~\cite{barocas2023fairness,mehrabi2021survey,balayn2021managing,nargesian2022responsible}.
Unfortunately, despite substantial focus on data preparation, machine learning, and algorithm design, {\em data structures and their potential to induce unfairness in downstream tasks have received limited attention}~\cite{zhang2022technical}.

Towards filling the research gap to understand potential harms and designing fair data structures, 
this paper revisits the {\em hashmap} data structure through the lens of fairness.
To the best of our knowledge, {\em this is the {\bf first} paper to study group fairness in a data structure design}.
Hashmaps are a founding block in many applications such as bloom filters for set membership~\cite{bruck2006weighted, pagh2005optimal, bloom1970space}, hash sketches for cardinality estimation~\cite{flajolet2007hyperloglog, durand2003loglog}, count sketches for frequency estimation~\cite{cormode2005improved}, min-hashes in similarity estimation~\cite{broder1997resemblance, chum2009geometric}, hashing techniques for security applications~\cite{al2011cryptographic, silva2003overview}, and many more.

Collision in a hashmap happens when the hash of two different entities is the same.
Collisions are harmful as those cause {\em false positives}. For example, in the case of bloom filters when the hash of a query point collides with a point in a queried set, the query point is falsely classified as a set member. In such cases, in the least further computations are needed to resolve the false positives. However, this would require to explicitly storing all set members in the memory, which may not be possible in all cases.
Note that false negative is impossible in case of hashmaps. That simply is because when $x=x'$, the hash of $x$ and $x'$ is always the same.
To further motivate the problem, let us consider Example~\ref{ex-1}.

\begin{example}\label{ex-1}
Consider an airline security application, which aims to identify passengers who may pose a threat, and subject them for further screening and potential prevention from boarding flights.
A set of criminal records is used to create a no-fly list. Due to privacy reasons, the criminals' identities are hashed and the list is a pair of \{hash, gender\} of individuals.
The passenger hashes are matched against the no-fly list for this purpose. 
False positives in airline security can lead to significantly inconveniencing passengers.
\end{example}

Traditional $k$-wise independent hashing~\cite{ostlin2003uniform, siegel1989universal} aims to randomly map a key (an entity) to a {\em random} value (bucket) in a specific output range.
However, given a set of points, it is unlikely that independent random value assignment to the points uniformly distribute the points to the buckets. For example, Fig.~\ref{fig:randNotUniform} shows the distribution of 100 independent and identically distributed (iid) random numbers, we generated in range [0,9]. While in a uniform distribution of the points, each bucket would have exactly 10 points, the random assignment did not satisfy it.
On the other hand, the number of collision is minimized, when the uniform distribution is satisfied.
In order to resolve this issue, 
{\em data-informed approaches} are designed, where given a set of data, the goal is to ``learn'' a proper hash function that uniformly distributes the data across different buckets~\cite{sabek2022can, kraska2018case, mitzenmacher2018model}. 
Particularly, given a data set of $n$ entities, the cumulative density function (CDF) of (the distribution represented by) the data set is constructed.
Then the hashmap is created by partitioning the range of values into $m$ buckets such that each buckets contains $\frac{n}{m}$ entities.
We refer to this approach~\cite{kraska2018case} as \cdf hashmap.
It has been shown that such index structures~\cite{sabek2022can, kraska2018case} can outperform traditional hashmaps on practical workloads.

To the best of our knowledge, none of the existing hashmap schemes consider fairness in terms of equal performance for different demographic groups. Given the wide range of hashmap applications, this can cause discrimination against minority groups, at least for social applications.
Therefore, in this paper, we study {\em group fairness} defined as {\em equal collision probability (false positive rate) for different demographic groups}, in hashmaps.
Specifically, targeting to prevent disparate treatment~\cite{green2011future,asudeh2019designing}, we propose \fairhash, a hashmap that satisfies group fairness.
We consider the \cdf hashmap for designing our fair data structure.

While there are many definitions of fairness, at a high level group fairness notions fall under three categories of independence, seperation, and sufficiency~\cite{barocas2023fairness}. 
Our proposed notion of group fairness falls under the \emph{independence} category, which is satisfied when the output of an algorithm is independent of the demographic groups (protected attributes).
Specifically, we adopt \emph{statistical parity}, A well-known definition under the independence category.

\subsection{Applications}\label{sec:intro:app}
\fairhash is a {\bf data-informed} hashmap, extended upon \cdf hashmap.
Data-informed hashmaps are proper for applications where a large-enough workload is available for learning to uniformly distribute the data across various buckets.
Particularly, data-informed hashmaps are preferred when the underlying data distribution is not uniform.
Besides, the choice between data-informed hashmaps and traditional hashmaps hinges on other factors such as conflict resolution policy and memory constraints.

Nevertheless, data-informed hashmaps effectively address a broad range of practical challenges where traditional approaches fall short. For instance, when dealing with larger payloads or distributed hash maps, it is advantageous to minimize conflicts, making data-informed hashmaps more beneficial. On the other hand, in scenarios involving small keys, small values, or data following a uniform distribution, traditional hash functions are likely to perform effectively~\cite{kraska2018case}.
In addition to Example~\ref{ex-1}, in the following we outline a few examples of the applications of data-informed hashmaps, which demonstrate the need for \fairhash in real-world scenarios.

\stitle{Table Joins in Data Lakes}
Consider an enterprise seeking to share its data with third-party companies. The data is stored in a data lake and includes sensitive information such as email, phone number, or even social security number of the clients serving as the primary keys for select tables. Disclosing such sensitive information constitutes a breach of client privacy, hence it needs to be masked through hashing.
In such cases, where the join operation is on the hashed columns, hash-value collisions would add invalid rows to the result table.
As a result, higher occurrence of false positives 
correlated with particular demographic groups can make the result table biased against that group. 
This underscores the need to employ a fair hashing scheme when anonymizing sensitive columns.

\stitle{Machine Learning Datasets}
Consider a scenario for constructing an ensemble model, where a large dataset is partitioned into multiple smaller random subsets, each used for training for a base model.
In this setting, ensuring that the sampling process does not introduce any bias is crucial; hence, each subset of the data should maintain the same ratio of each subpopulation as present in the original distribution. 
While the conventional random sampling method fails to ensure such distribution alignment, employing \fairhash to bucketize the data, guarantees this goal.

\stitle{Distributed Hashmaps and Load Balancing}
Collisions incur a substantial cost in distributed hashmaps, as each collision necessitates an extra lookup request on the remote machine through RDMA, taking on the order of microseconds. Therefore, higher collision rates linked to keys from a specific demographic group may cause a notable performance disadvantage against that group. 
Using \fairhash, for example, 
web servers can distribute incoming requests across multiple server instances. This guarantees an equitable distribution of the load, mitigating potential harm to a specific client if a server becomes unavailable.


\color{black}

\begin{figure*}[!tb] 
\begin{minipage}[t]{0.35\linewidth}
    \centering
    \includegraphics[width=\textwidth]{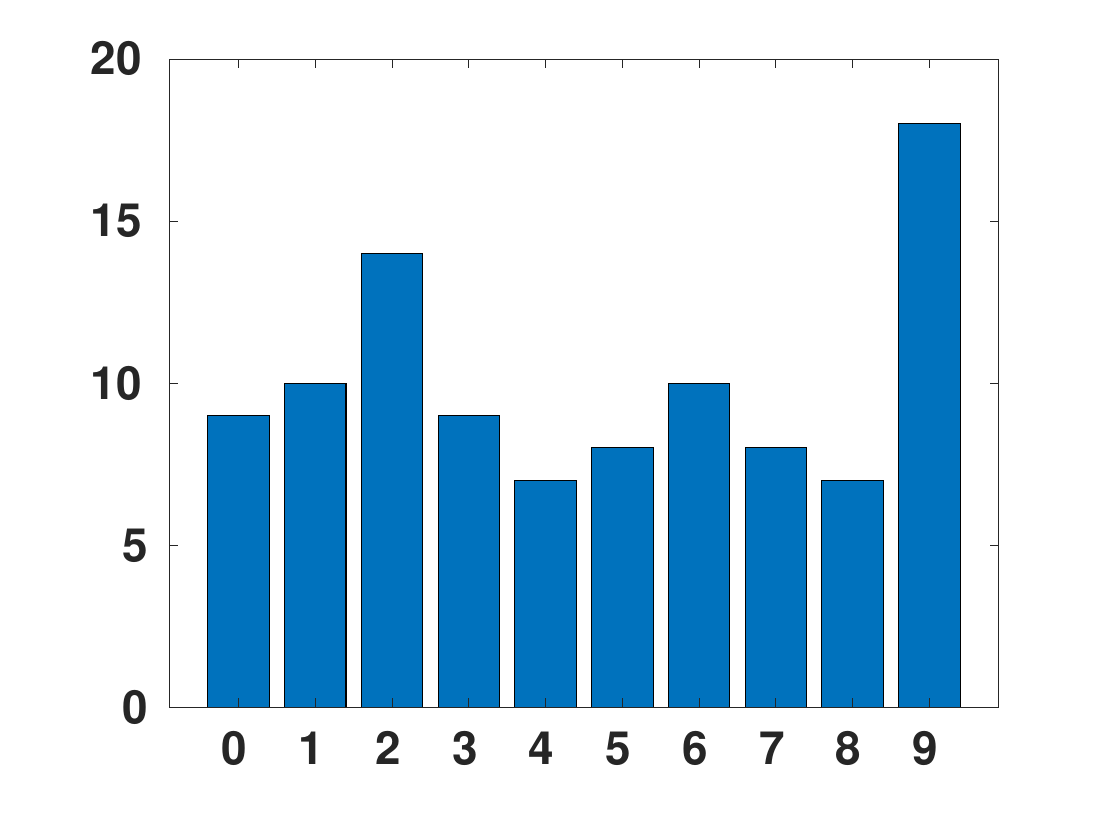}
    \caption{Distribution of 100 random integers in [0,9].}
    \label{fig:randNotUniform}
\end{minipage}
\begin{minipage}[t]{0.64\linewidth}
    \begin{subfigure}[t]{0.49\textwidth}
        \centering
        \includegraphics[width=\textwidth]{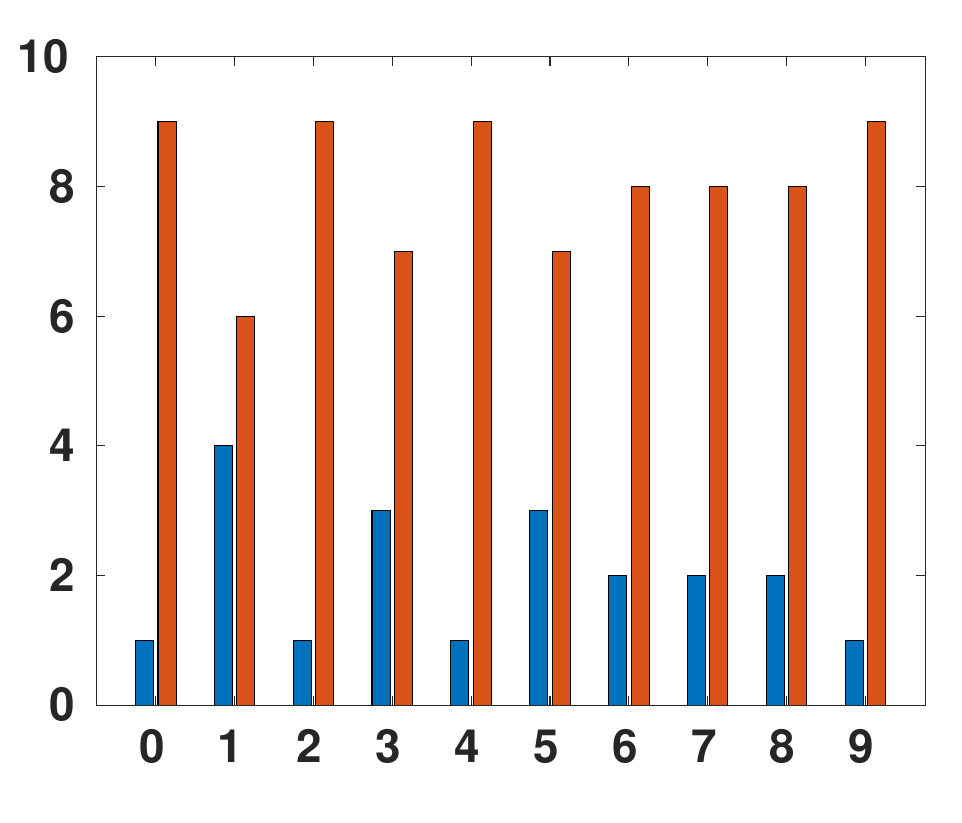}
        \vspace{-8mm}
        \caption{\cdf hashmap}
        \label{fig:cdfNotUniformGroup}
    \end{subfigure}
    \hfill
    \begin{subfigure}[t]{0.49\textwidth}
        \centering
        \includegraphics[width=\textwidth]{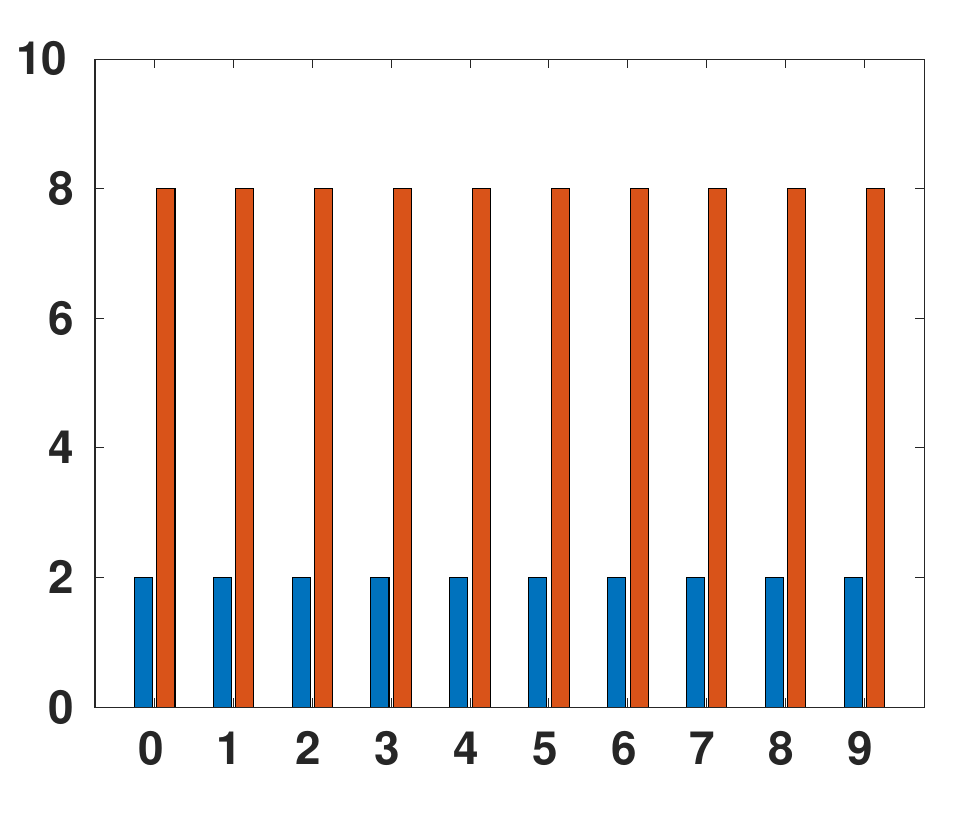}
        \vspace{-8mm}
        \caption{\fairhash}
        \label{fig:fairUniform}
    \end{subfigure}
    \vspace{-2mm}
\caption{Distribution of 100 points belonging to two groups {\tt blue} and {\tt red} in 10 buckets.}
\end{minipage}
\end{figure*}

\subsection{Technical Contributions} 

\noindent{\bf (I) Proposing \fairhash.}
We begin our technical contributions by proposing two notions of group fairness (single and pairwise) based on collision probability disparity between demographic groups.
Intuitively, single fairness is satisfied when the data is uniformly distributed across different buckets, i.e., the buckets are equi-size. 
Consequently, as reflected in Figure~\ref{fig:randNotUniform} traditional hashmaps do not satisfy single fairness. On the other hand, \cdf hashmap satisfies single fairness as all of its buckets have the same size.
Pairwise fairness is a stronger notion of fairness that, not only requires equi-size buckets but it also demands equal ratio of demographic groups across all buckets.
To better clarify this, let us consider the distribution of 100 random (synthetic) points from two groups red and blue into 10 buckets, using \cdf hashmap and \fairhash in Figure~\ref{fig:cdfNotUniformGroup} and Figure~\ref{fig:fairUniform}.
Although assigning equal number of points to each bucket, \cdf hashmap fails to satisfy equal ratio of groups across all buckets, and hence fails on the pairwise fairness. On the other hand, using \fairhash equal group ratios, hence pairwise fairness, is satisfied.
Also, from the figures it is evident that pairwise fairness is a stronger notion, not only requiring equi-size buckets but also equal group ratios. We shall prove of this in \S~\ref{sec:pre} after providing the formal terms and definitions.

\noindent{\bf (II) Ranking-based Algorithms.}
 Next, making the observation that only the {\em ranking} between the points (not their value distribution) impacts the fairness of the \cdf hashmap, we use geometric techniques to find alternative ranking of points for fair hashing. We propose multiple algorithmic results with various benefits. At a high level, our ranking-based approach, maintains the {\em same time and memory} efficiency as of \cdf hashing, while minimizing the unfairness (but not guaranteeing zero unfairness).

\noindent{\bf (III) Cut-based Algorithms.}
Our next contribution is based on the idea of adding more bins than the number of hash buckets. 
We propose \pd to prove that independent of the initial distribution of points, there {\em always exists a fair hashing} based on the cutting approach.
While guaranteeing fair hashing, \pd is not memory efficient.
Therefore, we make an interesting connection to the {\em necklace splitting} problem~\cite{alon1986borsuk,alon1987splitting,meunier2014simplotopal}, and using some of the recent advancements~\cite{alon2021efficient} on it, 
provide multiple algorithms for \fairhash. While guaranteeing the fair hashing, our algorithms achieve the same time efficiency as of the \cdf hashing with a small increase in the memory requirement. 

\noindent{\bf (IV) Discrepancy-based Algorithms.}
We also propose discrepancy-based algorithms that trade-off single fairness to achieve improve pairwise fairness and memory efficiency.
In addition to the theoretical analysis, we conduct experiments to verify the efficiency and effectiveness of our algorithms.
\section{Preliminaries}\label{sec:pre}
Let $\InputSet$ be a set of $n$ points\footnote{Throughout this paper, we assume access to $P$ is the complete set of points. For cases where instead an unbiased set of samples from $P$ is available, our results in Tables~\ref{tab:results} and \ref{tab:results2} will remain in an expected manner.} in $\Re^d$ (each point represents a tuple with $d$ attributes), where $d\geq 1$.
Let $\Groups=\{\group_1,\ldots, \group_k\}$ be a set of $k$ demographic groups\footnote{Demographic groups can be defined as the intersection of multiple sensitive attributes, such as \at{(race, gender)} as $\{$\at{black-female}$, \cdots \}$.} 
(e.g., \at{male}, \at{female}).
Each point $\point\in \InputSet$ belongs to group $\group(\point)\in \Groups$.
By slightly abusing the notation we use
$\gee_i$ to denote the set of points in group $\gee_i$.

Let $\hashmap$ be a hashmap with $m$ buckets, $\bucket_1,\ldots, \bucket_m$, and a hash function $\hashfunction: \Re^d\rightarrow [1,m]$ that maps each input point $p\in \InputSet$ to one of the $m$ buckets.
Given the pair $(P,\hashmap)$, we define three quantitative requirements that are used to define fairness in hash functions.


\vspace{2mm}
\begin{enumerate}
    \item {\em Collision Probability (Individual fairness)}: For any pair of points $p, q$ taken uniformly at random from $P$, it should hold that $Pr[\hashfunction(p)=\hashfunction(q)]= \frac{1}{m}$.
    \item {\em Single fairness}: For each $i\leq k$, for any point $p_i$ taken uniformly at random from $\gee_i$ and any point $x$ taken uniformly at random from $P$, it should hold that $Pr[\hashfunction(p_i)=\hashfunction(x)]=\ldots= Pr[\hashfunction(p_k)=\hashfunction(x)]=\frac{1}{m}$.
    \item {\em Pairwise fairness}: 
    For each $i\leq k$, for any pair of points $p_i, q_i$ taken uniformly at random from $\gee_i$, it should hold that $Pr[\hashfunction(p_i)=\hashfunction(q_i)]=\ldots= Pr[\hashfunction(p_k)=\hashfunction(q_k)]=\frac{1}{m}$.
\end{enumerate}


\stitle{Ensuring Pareto-optimality}
A major challenge when formulating fair algorithms is that those may generate Pareto-dominated solutions~\cite{nilforoshan2022causal}. 
That is, those may produce a fair solution that are worse for all groups, including minorities, compared to the unfair ones.
In particular, in a utility assignment setting, let the utility assigned to each group $\gee_i$ by the fair algorithm be $u_i$. There may exist another (unfair) assignment that assigns $u'_i>u_i$, $\forall \gee_i\in \Gee$.
This usually can happen when fairness is defined as the parity between different groups, without further specifications.

In this paper, the requirements 2 and 3 (single and pairwise fairness) have been specifically defined in a way to prevent generating Pareto-dominated fair solutions.
To better explain the rational behind our definitions, let us consider pairwise fairness (the third requirement). Only requiring equal collision probability between various groups, the fairness constraint would translate to $Pr[\hashfunction(p_i)=\hashfunction(q_i)]= Pr[\hashfunction(p_k)=\hashfunction(q_k)]$, where $p_i$ and $q_i$ belong to group $\gee_i$.
Now let us consider a toy example with two groups $\{\gee_1,\gee_2\}$,  two buckets $\{b_1,b_2\}$, and $n$ points where half belong to $\gee_1$. Let the hashmap $\hashmap$ map each point $p\in \gee_1$ to $b_1$ and each point $q\in\gee_2$ to $b_2$.
%
%
In this example, the collision probability between a random pair of points is 0.5, simply because each bucket contains half of the points.
It also satisfies the collision probability equality between pairs of points from the same groups:
$Pr[\hashfunction(p_i)=\hashfunction(q_i)]= Pr[\hashfunction(p_k)=\hashfunction(q_k)] = 1$.

This, however, is the worst assignment for both groups as their pairs {\em always collide}. In other words, it is fair, in a sense that it is equally bad for both groups. Any other hashmap would have a smaller collision probability for both groups and, even if not fair, would be more beneficial for both groups.
In other words, this is a fair but Pareto-dominated solution.

In order to ensure Pareto-optimality while developing fair hashmaps, in requirements 2 and 3 (single and pairwise fairness), not only we require the probabilities to be equal, but also we require them to be equal to the {\em best case}, where the collision probability is $\frac{1}{m}$. As a result, we guarantee that (1) our hashmap is fair and (2) no other hashmap can do better for any of the groups.

{
\stitle{Calculation of probabilities}
Next, we give the exact close forms for computing the collision probability, the single fairness, and the pairwise fairness.
Let $\alpha_{i,j}$ be the number of items from group $i$ at bucket $j$ and let $n_j=\sum_{i=1}^k\alpha_{i,j}$.
The probabilities are calculated as follows,
collision probability: $\sum_{j=1}^m\left(\frac{n_j}{n}\right)^2$, single fairness: $\sum_{j=1}^m\frac{\alpha_{i,j}}{|\gee_i|}\cdot \frac{n_j}{n}$, pairwise fairness: $\sum_{j=1}^m\left(\frac{\alpha_{i,j}}{|\gee_i|}\right)^2$.
We observe that, if and only if $|\gee_i|/m$ tuples from each group $\gee_i$ are placed in every bucket then all quantities above are equal to the optimum value $\frac{1}{m}$.
}

\stitle{The relationship between the three requirements}
In this paper we aim to construct a hashmap $\hashmap$ that satisfies (approximately) all the three requirements.
\begin{proposition}\label{prop:1}
Collision probability is satisfied if and only if all buckets contain exactly the same number of points i.e., for each bucket $b_j$, $|b_j|=\frac{n}{m}$.
\end{proposition}
\begin{proposition}\label{prop:2}
If collision probability is satisfied then single fairness is also satisfied.
\end{proposition}
\begin{proposition}\label{prop:3}
Pairwise fairness is satisfied if and only if for every group $\gee_i\in \Gee$, every bucket $b_j$ contains the same number of points from group $\gee_i$, i.e., $b_j$ contains $\frac{|\gee_i|}{m}$ items from group $\gee_i$. 
{
If pairwise fairness is satisfied then both single fairness and the collision probability are satisfied but the reverse may not necessarily hold.}
\end{proposition}
\begin{proof}
{
We give the proofs to all propositions above.

For Proposition~\ref{prop:1}, if each bucket contains $n/m$ items then the collision probability is $\sum_{j=1}^m \left(\frac{n/m}{n}\right)^2=\frac{1}{m}$, so it is satisfied. For the other direction, we assume that the collision probability holds.
Notice that $\sum_{j=1}^m \left(\frac{n_j}{n}\right)^2=\frac{1}{m}\Leftrightarrow \sum_{j=1}^m n_j^2=\frac{n^2}{m}$. For any integer values $n_j\geq 0$ with $\sum_{j=1}^m n_j=n$, it holds that $\sum_{j=1}^m n_j^2\geq \sum_{j=1}^m \left(\frac{n}{m}\right)^2=\frac{n^2}{m}$ and the minimum value is achieved only for $n_j=\frac{n}{m}$ for each $j\in[1,m]$. The result follows.

For Proposition~\ref{prop:2}, we have that if the collision probability is satisfied then from Proposition~\ref{prop:1}, it holds $n_j=\frac{n}{m}$. Then the single fairness for any group $\gee_i$ is computed as $\sum_{j=1}^m\frac{\alpha_{i,j}}{|\gee_i|}\frac{n_j}{n}=\sum_{j=1}^m\frac{\alpha_{i,j}}{|\gee_i|}\frac{n/m}{n}=\frac{1}{m\cdot |\gee_i|}\sum_{j=1}^m \alpha_{i,j}=\frac{1}{m}$, so it is satisfied.

The first  part of Proposition~\ref{prop:3} follows directly from Proposition~\ref{prop:1}, because the pairwise fairness of group $\gee_i$ is equivalent to the collision probability assuming that $P=\gee_i$. For the second part, if pairwise fairness is satisfied, then from the first part of Proposition~\ref{prop:3}, we know that $n_j=\frac{n}{m}$. Then from Proposition~\ref{prop:1} the collision probability is satisfied, so from Proposition~\ref{prop:2}, the single fairness is also satisfied. Finally, the construction in the proof of Lemma~\ref{lem:1} shows that the reverse may not necessarily hold.
}
\end{proof}

From Propositions~\ref{prop:1} to \ref{prop:3}, in order to satisfy the three requirements of collision probability, single fairness, and pairwise fairness, it is enough to develop a hashmap that satisfies pairwise fairness (which will generate equal-size buckets).
In other words, {\em pairwise fairness is the strongest property}, compared to the other two.


Our goal is to design hashmaps that, while satisfying collision probability and single fairness requirements, satisfies {\em pairwise fairness} as the stronger notion of fairness\footnote{To simplify the terms, in the rest of the paper, we use ``fairness'' to refer to pairwise fairness. We shall explicitly use ``single fairness'' when we refer to it.}.
Specifically, we want to find the hashmap $\hashmap$ with $m$ buckets to optimize the pairwise fairness.

\stitle{Measuring unfairness}
For a group $\gee\in \Gee$, let $Pr_\gee$ be the pairwise collision probability between its members. That is, $Pr_\gee=Pr[h(p)=h(q)]$, if $p$ and $q$ are selected uniformly at random from the same group $\gee(p)=\gee(q)=\gee$. 
We measure unfairness as the max-to-min ratio in pairwise collision probabilities between the groups. 
Particularly, using $\frac{1}{m}$ as the min on the collision probability, we say a hashmap is {\bf $\eps$-unfair}, if and only if 

\begin{align}
    \frac{\max_{\gee\in\Gee}(Pr_\gee)}{1/m} \leq (1 +\eps)
    ~\Rightarrow~ \max_{\gee\in\Gee}(Pr_\gee) \leq \frac{1}{m}(1 +\eps)
\end{align}

It is evident that for a hashmap that satisfies pairwise fairness, $\eps=0$.
We use $\frac{1}{m}$ as the min on the collision probability to ensure pareto-optimality.

\stitle{Memory efficiency} A hashmap with $m$ buckets needs to store at least $m-1$ \emph{boundary points} to separate the $m$ buckets. 
While our main objective is to satisfy fairness, we also would like to minimally increase the required memory to separate the buckets.
We say a hashmap with $m$ buckets satisfies {\bf $\alpha$-memory}, 
if and only if it stores at most $\alpha(m-1)$ boundary points. 
Evidently, the most memory-efficient hashmap satisfies $\alpha=1$.

\begin{definition}[ $(\eps,\alpha)$-hashmap]
    A hashmap $\hashmap$ is an {\bf $(\eps,\alpha)$-hashmap} if and only if it is $\eps$-unfair and $\alpha$-memory.
\end{definition}
\stitle{Problem formulation} Given $\InputSet$, $\Groups$, and parameters $m, \eps, \alpha$, the goal is to design an $(\eps, \alpha)$-hashmap.

In this work we mostly focus on $(\eps,1)$-hashmaps and $(0,\alpha)$-hashmaps, minimizing $\eps$ and $\alpha$, respectively\footnote{For simplicity, throughout the paper, we consider that the cardinality of each group $\gee_i$ is divisible by $m$.}.
Note that, in the best case, one would like to achieve $(0, 1)$-hashmap.
That is a hash-map that is $0$-unfair and does not require additional memory, i.e., satisfies 1-memory.

\stitle{(Review) \cdf hashmap~\cite{kraska2018case}} 
is a data-informed hashmap that ``learns'' the cumulative density function of values over a specific attribute, and use it to place the boundaries of the $m$ buckets such that an equal number of points ($\frac{n}{m}$) fall in each bucket.
Traditional hash functions and learned hash functions (CDF) usually satisfy the collision probability and the single fairness probability, however they violate the pairwise fairness property.

\begin{lemma}\label{lem:1}
While \cdf hashmap satisfies collision probability, hence single fairness, it may not satisfy pairwise fairness.
\end{lemma}
\begin{proof}
{
From~\cite{kraska2018case}, it is always the case that each bucket contains the same number of tuples, i.e., $\frac{n_j}{n}=\frac{1}{m}$ for every $j=1,\ldots, m$. Hence, 
from Proposition~\ref{prop:1}, the collision probability is satisfied. Then, from Proposition~\ref{prop:2}, the single fairness is also satisfied.

Next, we show that \cdf hashmap does not always satisfy pairwise fairness using a counter-example. Let $k=2$, $m=2$, $|\gee_1|=3$, $|\gee_2|=3|$, and $n=6$. Assume the $1$-dimensional tuples $\{1,2,3,4,5,6\}$ where the first $3$ of them belong to $\gee_1$ and the last three belong to $\gee_2$. The \cdf hashmap will place the first three tuples into the first bucket and the last three tuples into the second bucket. By definition, the pairwise fairness is $1$ (instead of $1/2$) for both groups.}
\end{proof}

{
\stitle{Data-informed hashmaps vs. traditional hashmaps}
A summary of the comparison between \cdf and traditional hashmaps is presented in Table~\ref{tab:comparison}.
The most major difference between \cdf and hashmaps is that the former is data-informed. That is, the \cdf hashmap is tailored for a specific data workload, while traditional hashmaps are data-independent, i.e., their behavior does not depend on the data those are applied on.
As a result, as mentioned in \S~\ref{sec:intro:app}, the main assumption and the requirement of the \cdf hashmap is access to a large-enough workload $\InputSet$ for learning the CDF function. On the other hand, traditional hashmaps do not require access to any data workload. 
While the traditional hashmaps compute the hash value of a query point in constant time, \cdf hashmap requires to run a binary search on the bucket boundaries, and hence has a query time logarithmic to $m$.
Having learned the data distribution, the \cdf hashmap guarantees a uniform distribution of data across different buckets, while as shown in Figure~\ref{fig:randNotUniform} traditional hashmaps cannot guarantee that. As a result, \cdf hashmap guarantees equal collision probability and single fairness, while traditional hashmaps do not.
Finally, both traditional and \cdf hashmaps fail to guarantee pairwise fairness, a requirement that \fairhash satisfies.
}
\begin{table*}[!tbh]
    \small
\centering
	\caption{Summary of the algorithmic results with exact ($1/m$) collision and single fairness probability.} 
	\label{tab:results}
	\begin{tabular}{@{}c@{}|c|c||c|c|@{}c@{}}
        &\multicolumn{2}{c||}{{\bf Assumptions}}&\multicolumn{3}{c}{{\bf Performance$^*$}}\\
		{\bf Algorithm} & {\bf No. } & {\bf No.} & {\bf $(\eps,\alpha)$-hashmap} & {\bf Query} & {\bf Pre-processing} \\
        {\bf } & {\bf Attributes} & {\bf Groups} & & {\bf time} & {\bf time} \\
  \hline \hline
        \ranker &$d\geq2$&$k\geq 2$&$(\optRank,1)$&$O(\log m)$&$O(n^{d}\log n)$\\ \hline
        \pd &$d\geq1$&$k\geq 2$&$(0,\frac{n}{m})$&$O(\log n)$&$O(n\log n)$\\ \hline
        \necklaceb &$d\geq1$&$2$&$(0,2)$&$O(\log m)$&$O(n\log n)$\\ \hline
        \necklacek & $d\geq 1$&$k> 2$&$(0,k(4+\log n))$&$O(\log (km\log n))$&$O(mk^3\log n+knm(n+m))$\\ 
        \hline
	\end{tabular}
	\begin{flushleft}
		$*$: $n$ is the dataset size, $m$ is the number of buckets, and $k$ is the number of groups.
	\end{flushleft}
\end{table*}
\begin{table*}[!tbh]
\small
	\caption{Summary of the algorithmic results with approximate collision and single fairness probability. The output of \rankerProb holds with probability at least $1-1/n$.} 
	\label{tab:results2}
	\begin{tabular}{@{}c@{}|c|c||c|c|@{}c@{}}
        &\multicolumn{2}{c||}{{\bf Assumptions}}&\multicolumn{3}{c}{{\bf Performance$^*$}}\\
		{\bf Algorithm} & {\bf No. } & {\bf No.} & {\bf $(\eps,\alpha)$-hashmap} & {\bf Query} & {\bf Pre-processing} \\
        {\bf } & {\bf Attr.} & {\bf Groups} & & {\bf time} & {\bf time} \\
  \hline \hline
        \ranker &$d\geq 2$&$k\geq 2$&$(\optDisc,1)$&$O(\log m)$&$O(n^{d+2}m\log k)$\\ \hline
        \rankerProb &$d\geq 2$&$k\geq 2$&$((1+\delta)\optDisc+\gamma,1)$&$O(\log m)$&$O(n+\frac{k^{d+2}m^{2d+5}}{\gamma^{2d+4}}\polylog (n,\frac{1}{\delta}))$\\ \hline
        \necklacek & $d\geq 1$&$k> 2$&$(\eps,k(4+\log\frac{1}{\eps}))$&$O(\log (km\log\frac{1}{\eps}))$&$O( mk^3\log\frac{1}{\eps}+knm(n+m))$\\ \hline
	\end{tabular}
\end{table*}


\begin{table*}[!tbh]
    \small
    \centering
    \caption{Comparison between \cdf and traditional hashmaps.}
    \label{tab:comparison}
    {
    \begin{tabular}{c|c|c||c|c|c}
        && {\bf Query}& {\bf Collision} & {\bf Single} & {\bf Pairwise} \\
        {\bf Hashmap} &{\bf Architecture}& {\bf time}&{\bf probability}&{\bf fairness}&{\bf fairness}\\
        \hline
        traditional&data-independent&$O(1)$&\ding{55}&\ding{55}&\ding{55}
        \\\hline
        \cdf&data-dependent&$O(\log m)$&\ding{51}&\ding{51}&\ding{55}
        \\\hline
        \fairhash&data-dependent&$O(\log m)$&\ding{51}&\ding{51}&\ding{51}
    \end{tabular}}
\end{table*}

\color{black}



\subsection{Overview of the algorithmic results}
We propose two main approaches for defining the hashmaps, called \emph{ranking-based} approach,  and \emph{cut-based} approach. 
A summary of our algorithmic results with perfect collision probability and single fairness for all groups, is shown in Table~\ref{tab:results}.

{ 
Let $W$ be the set of all possible unit vectors in $\Re^d$. 
Given a vector $w\in W$, let $P_w$ be the ordering defined by the projection of $P$ onto $w$. 
Based on this ordering, we construct $m$ equi-size buckets. 
In the ranking-based approaches, we focus on finding the best vector $\vector$ to take the projection on that minimizes the unfairness.
Let $OPT_R(P_w)$ be the smallest parameter such that an $(OPT_R(P_w), 1)$-hashmap exists in $P_w$ with collision probability and single fairness equal to $1/m$.
We define $\optRank=\min_{w\in W}OPT_R(P_w)$.

In the cut-based approaches, we define $\beta (m-1)+1$ intervals $\Intervals=\{\interval_1,\ldots, \interval_{\beta (m-1)+1}\}$ defined by the $\beta (m-1)$ boundary points, such that for each point $\point\in \InputSet_\vector$ (the projection of $\InputSet$ on a vector $\vector$) there exists an interval $\interval\in\Intervals$ where $\point\in \interval$.
Each interval $\interval\in \Intervals$ is assigned to one of the $m$ buckets.
Our focus on cut-based approaches is on finding the best way to place the boundary points on a given ordering $\InputSet_\vector$. 
In all cases the hashmap $\hashmap$ stores the vector $\vector$ along with the $\beta (m-1)+1$ intervals $\Intervals$ and their assigned buckets.
During the query phase, given a point $q\in \Re^d$, we first apply the projection $\langle\vector, q\rangle$ and we get the value $q_w\in \Re$. 
Then we run a binary search on the boundary points to find the interval $\interval\in \Intervals$ such that $q_w \in \interval$. We return $\bucket(I)$ as the bucket $q$ belongs to.
}




{
So far, we consider that the collision probability and the single fairness should be (optimum) $\frac{1}{m}$. 
However,
this restricts the options when finding fair hashmaps. 
What if, there exists a hashmap with better pairwise fairness having slightly more or less than $n/m$ items in some buckets?
Hence, we introduce the notion of $\gamma-$\emph{discrepancy}~\cite{alon2021efficient}. 
The goal is to find a hashmap with $m$ buckets such that each bucket contains at most $(1+\gamma)\frac{|\GroupPoints_i|}{m}$ and at least $(1-\gamma)\frac{|\GroupPoints_i|}{m}$ points from each group $i\leq k$.
A summary of our algorithmic results with approximate collision probability and single fairness for all groups, is shown in Table~\ref{tab:results2}.
Let $W$ be the set of all possible unit vectors in $\Re^d$. Let $OPT_D(P_w)$ be the smallest parameter such that a hashmap with $OPT_D(P_w)$-discrepancy exists in $P_w$. We define $\optDisc=\min_{w\in W}OPT_D(P_w)$.
}

{
\mparagraph{Fairness and memory efficiency trade-off} 
Ideally, one would like to develop a hashmap that is 0-unfair and 1-memory. But
in practice,
depending on the distribution of the data, achieving both at the same time may not be possible.
For cases where fairness is a hard constraint, cut-based algorithms are preferred, as those guarantee 0-unfairness, irrespective of the data distribution. But that is achieved at a cost of increasing memory usage.
On the other hand, ranking-based algorithms minimize unfairness without requiring any extra memory but do not 0-unfairness; hence those are fit when memory is a hard constraint.
Last but not least, the discrepancy-based algorithms provide a trade-off between pairwise and single fairness. Specifically, these algorithms do not guarantee to contain exactly $\frac{n}{m}$ points in each bucket, and hence, may not satisfy the first two requirements (individual and single fairness). However, for cases where adding a small amount of single unfairness is tolerable, the discrepancy-based algorithms may further reduce the pairwise unfairness of the ranking-based and the memory bound of the cut-based algorithms.
}

\mparagraph{Remark}
In all cases, for simplicity, we assume that our algorithms have access to the entire input set in order to compute a near-optimal hashmap. However, our algorithms can work in expectation, if an unbiased sample set from the input set is provided. 
We verify this, experimentally in Section~\ref{sec:exp:learning}.
\section{Ranking-based Algorithms}\label{sec:ranking}
We start our contribution by defining a space of $(\eps,1)$-hashmaps, and designing algorithms to find the (near) optimum hashmap with the {\em smallest $\eps$}. By definition, recall that $\optRank$ is the smallest unfairness we can find with this technique (assuming $1$-memory).

Our key observation is that {\em only the ordering between the tuples matters} when it comes to pairwise fairness, not the attribute values. 
Hence, assuming that $d>1$, 
our idea is to combine the attribute values of a point $p\in P$ into a single score $f(p)$, using a function $f:\Re^d\rightarrow \Re$ called the ranking function, and consider the ordering of the points based on their scores for creating the hashmap.
Then, in a class of ranking functions, the objective is to find the one that returns the best $(\eps,1)$-hashmap with the smallest value of $\eps$. Of course, $f(p)$ needs to be computed efficiently, ideally in {\em constant time}.
Therefore, we select linear ranking functions, where the points are linearly projected on a vector $w\in \Re^d$.
That is, $f_w(p)=\langle p, w\rangle = p^\top w$.
Notice that the value $f_w(p)$ can be computed in $O(d)=O(1)$ time and it defines an ordering between the different tuples $p\in P$.
For a vector $\vector$, let $P_w$ be the ordering defined by $f_{\vector}$ and let $P_{\vector}[j]$ be the $j$-th largest tuple in the ordering $P_\vector$.
Given the ordering defined by a ranking function $f_w$, we construct $(m-1)$ boundaries to partition the data into $m$ equi-size buckets (each containing $\frac{n}{m}$ points).
Then a natural algorithm to construct an $(\eps, 1)$-hashmap is to run the subroutine for each possible ranking function and in the end return the best ranking function we found.

For simplicity, we describe our first method in $\Re^2$.  All our algorithms are extended to any constant dimension $d$.

It is known that for $P\subset \Re^2$ there exist $O(n^2)$ combinatorially different ranking functions~\cite{edelsbrunner1987algorithms}. 
We can easily compute them if we work in the dual space~\cite{edelsbrunner1987algorithms}. For a point $p=(x_p,y_p)\in P$ we define the dual line $\lambda(p): x_px_1 + y_px_2=1$. Let $\Lambda=\{\lambda(p)\mid p\in P\}$ be the set of $n$ lines. Let $\Arrang(\Lambda)$ be the arrangement~\cite{edelsbrunner1987algorithms} of $\Lambda$, which is defined as the decomposition of $\Re^2$ into connected (open) cells of dimensions $0, 1, 2$ (i.e., point, line segment, and convex polygon) induced by $P$.
It is known that $\Arrang(\Lambda)$ has $O(n^2)$ cells and it can be computed in $O(n^2\log n)$ time~\cite{edelsbrunner1987algorithms}.

Given a vector $w$, the ordering $P_w$ is the same as the ordering defined by the intersections of $\Lambda$ with the line supporting $w$. Hence, someone can identify all possible ranking functions $f_w$ by traversing all intersection points in $\Arrang(\Lambda)$. Each intersection in $\Arrang(\Lambda)$ is represented by a triplet $(p,q,w)$ denoting that the lines $\lambda(p), \lambda(q)$ are intersecting and the intersection point lies on the line supporting the vector $w$. Let $\mathcal{W}$ be the array of size $O(n^2)$ containing all intersection triplets sorted in ascending order of the vectors' angles with the $x$-axis. Let $\mathcal{W}[i]$ denote the $i$-th triplet $(p_i,q_i, w_i)$. It is known that the orderings $P_{w_i}, P_{w_{i+1}}$ differ by swapping the ranking of two consecutive items.
Without loss of generality, we assume that for $(p_i,q_i, w_i)\in \mathcal{W}$,
the ranking of $p_i$ is higher than the ranking of $q_i$ for vectors with angle greater than $w_i$. 
The array $\mathcal{W}$ can be constructed in $O(n^2\log n)$ time.

\subsection{Algorithm}
Using $\mathcal{W}$, the goal is to find the best $(\eps,1)$-hashmap over all vectors satisfying the collision probability and the single fairness.
The high level idea is to consider each different vector $w\in \mathcal{W}$, and for each ordering $P_w$, find the hashmap that satisfies the collision probability and the single fairness measuring the unfairness. In the end, return the vector $w$ along with the boundaries of the hashmap with the smallest unfairness we found. If we execute it in a straightforward manner, we would have $O(|\mathcal{W}|\cdot n\log n)=O(n^{d+1}\log n)$ time algorithm. Next, we present a more efficient implementation applying fast update operations each time that we visit a new vector.

The pseudo-code of the algorithm in $\Re^2$ is provided in Algorithm~\ref{alg:ranking2d}.
The algorithm starts with the initialization of some useful variables and data structures.
Let $w_0$ be the unit vector with angle to the $x$-axis slightly smaller than $w_1$'s angle.
We visit each point in $P$ and we find the (current) ordering, denoted by $P_{w}$, sorting the projections of $P$ onto $w_0$.
Next, we compute the best $(\eps,1)$-hashmap in $P_{w}$.
The only way to achieve optimum collision probability and single fairness in $P_{w}$ is to construct exactly $m$ buckets containing the same total number of tuples in each of them.
Identifying the buckets (and constructing the hashmap) in $P_w$ is trivial because every bucket should contain exactly $n/m$ items. Hence, the boundaries of the $j$-th hashmap bucket are defined as $P_{w}[(j-1)\frac{n}{m}+1], P_{w}[j\frac{n}{m}]$, for $j\in[1,m]$.
Next, we compute the unfairness with respect to $w_0$.
For each group $\gee_i$ that contains at least one item in $j$-th bucket we set $\alpha_{i,j}=0$.
We use the notation $\alpha_{i,j}$ to denote the number of tuples from group $i$ in $j$-th bucket.
Let $P_{w}[\ell]$ be the next item in the $j$-th bucket, and let $P_{w}[\ell]\in \gee_i$. We update $\alpha_{i,j}\leftarrow \alpha_{i,j}+1$.
After traversing all items in $P_{w}$, we compute the pairwise fairness of group $\gee_i$ as $Pr_i=\sum_{j=1}^m\left(\frac{\alpha_{i,j}}{|\GroupPoints_i|}\right)^2$.
The unfairness with respect to $w_0$ is
$\eps=m\cdot \max_{i\leq k} Pr_i-1.$
After computing all values $Pr_i$, we construct a max heap $M$ over $\{Pr_1,\ldots, Pr_k\}$.
Let $w^*=w_0$. We run the algorithm visiting each vector in $\mathcal{W}$ maintaining $\eps, M, w^*, P_w, Pr_i, \alpha_{i,j}$ for each $i$ and $j$.

As the algorithm proceeds, assume that we visit a triplet $(p_s,q_s,w_s)$ in $\mathcal{W}$. If $p_s$ and $q_s$ belong in the same bucket, we only update the positions of $p_s, q_s$ in $P_w$ and we continue with the next vector. Similarly, if both $p_s, q_s$ belong in the same group $\gee_i$ we only update the position of them in $P_w$ and we continue with the next vector. Next, we consider the more interesting case where $p_s\in \gee_i$ belongs in the $j$-th bucket and $q_s\in \gee_{\ell}$ belongs in the $(j+1)$-th bucket of $P_w$, with $i\neq \ell$, just before we visit $(p_s,q_s,w_s)$. We update $P_w$ as in the other cases. However, now we need to update the pairwise fairness.
In particular, we update,
$$Pr_i=Pr_i-\left(\frac{\alpha_{i,j}}{|\GroupPoints_i|}\right)^2-\left(\frac{\alpha_{i,j+1}}{|\GroupPoints_i|}\right)^2+\left(\frac{\alpha_{i,j}-1}{|\GroupPoints_i|}\right)^2+\left(\frac{\alpha_{i,j+1}+1}{|\GroupPoints_i|}\right)^2,$$
and similarly
$$Pr_\ell=Pr_\ell-\left(\frac{\alpha_{\ell,j+1}}{|\GroupPoints_\ell|}\right)^2-\left(\frac{\alpha_{\ell,j}}{|\GroupPoints_\ell|}\right)^2+\left(\frac{\alpha_{\ell,j+1}-1}{|\GroupPoints_\ell|}\right)^2+\left(\frac{\alpha_{\ell,j}+1}{|\GroupPoints_\ell|}\right)^2.$$

\vspace{1mm}\noindent
Based on the new values of $Pr_i, Pr_\ell$ we update the max heap $M$.
We also update $\alpha_{i,j}=\alpha_{i,j}-1$, $\alpha_{i,j+1}=\alpha_{i,j+1}+1$, $\alpha_{\ell,j+1}=\alpha_{\ell,j+1}-1$, $\alpha_{\ell,j}=\alpha_{\ell,j}+1$.
If the top value of $M$ is smaller than $\eps$, then we update $\eps$ with the top value of $M$ and we update $w^*=w_s$.
After traversing all vectors in $\mathcal{W}$ we return the best vector $w^*$ we found. The boundaries can easily be constructed by finding the ordering $P_{w^*}$ satisfying that each bucket contains exactly $n/m$ items.

\begin{algorithm}[!tb]
    \caption{\rtwoD} \label{alg:ranking2d}
    \begin{algorithmic}[1] \small
    \Require{Set of points $P\in \Re^2$}
    \Ensure{vector $w^*$ and corresponding boundaries $B$}
        \State Construct and sort the vectors in $\mathcal{W}$ with respect to their angles;
        \State $w_0\leftarrow$ vector with angle to the $x$-axis slightly smaller than $w_1$'s angle;
        \State $P_w\leftarrow$ sorted projection of $P$ onto $w_0$;
        \State $b_j\leftarrow \left(P_w[(j-1)\frac{n}{m}+1], P_w[j\frac{n}{m}]\right)\quad \forall j=1,\ldots, m$;
        \State $\alpha_{i,j}\leftarrow |\gee_i\cap b_j|, \quad \forall i=1,\ldots, k$, $j=1,\ldots, m$;
        \State $Pr_i\gets \sum_{j=1}^m\left(\frac{a_{i,j}}{|\gee_i|}\right)^2, \quad \forall i=1,\ldots, k$;
        \State $M\leftarrow$ Max-Heap on $Pr_i, \forall i=1,\ldots, m$;
        \State $\eps\leftarrow m\cdot \max_{i\leq k} Pr_i -1$;
        $w^*=w_0$;
        \For{$(p_s,q_s,w_s)\in \mathcal{W}$}
            \If{$p_s, q_s$ belong in the same bucket OR $\gee(p_s)== \gee(q_s)$}
                \State Swap $p_s$, $q_s$ and update $P_w$;
            \Else
                \State Let $p_s\in b_j$, $\gee(p_s)=\gee_i$, $q_s\in b_{j+1}$, $\gee(q_s)=\gee_\ell$;
                \State $Pr_i\gets Pr_i-\left(\frac{\alpha_{i,j}}{|\GroupPoints_i|}\right)^2-\left(\frac{\alpha_{i,j+1}}{|\GroupPoints_i|}\right)^2+\left(\frac{\alpha_{i,j}-1}{|\GroupPoints_i|}\right)^2+\left(\frac{\alpha_{i,j+1}+1}{|\GroupPoints_i|}\right)^2$;
                \State $Pr_\ell\gets Pr_\ell-\left(\frac{\alpha_{\ell,j+1}}{|\GroupPoints_\ell|}\right)^2-\left(\frac{\alpha_{\ell,j}}{|\GroupPoints_\ell|}\right)^2+\left(\frac{\alpha_{\ell,j+1}-1}{|\GroupPoints_\ell|}\right)^2+\left(\frac{\alpha_{\ell,j}+1}{|\GroupPoints_\ell|}\right)^2$;
                \State $\alpha_{i,j}\!\!=\!\!\alpha_{i,j}\!-\!1, \alpha_{i,j+1}=\alpha_{i,j+1}+1, \alpha_{\ell,j+1}=\alpha_{\ell,j+1}-1, \alpha_{\ell,j}=\alpha_{\ell,j}+1$;
                \State Update $Pr_i, Pr_\ell$ in $M$;
                \If{$m\cdot M.top()-1<\eps$}
                     $\{\eps=m\cdot M.top()-1$;
                    $w^*=w_s;\}$
                \EndIf
            \EndIf
        \EndFor
        \For{$j=1$ to $m$}
             $B_j\gets \frac{P_{w^*}[j\frac{n}{m}]+P_{w^*}[j\frac{n}{m}+1]}{2}$;  {\tt\scriptsize // right boundary of $b_j$ }
        \EndFor
        \State {\bf return} $(w^*,B)$;  
    \end{algorithmic}
\end{algorithm}

\stitle{Analysis}
The correctness of the algorithm follows from the definitions.
Next, we focus on the running time. We need $O(n^2\log n)$ to construct $\Arrang(\Lambda)$ and $O(n\log n)$ additional time to initialize $\eps, M, w^*,$ $P_w, Pr_i, \alpha_{i,j}$. For each new vector $w_s$ we visit, we update $P_w$ in $O(1)$ time by storing the position of each item $p\in P$ in $P_w$ using an auxiliary array. All variables $Pr_i, Pr_{\ell}, \alpha_{i,j}, \alpha_{i,j+1}, \alpha_{\ell,j}, \alpha_{\ell, j+1}$ are updated executing simple arithmetic operations so the update requires $O(1)$ time. The max heap $M$ is updated in $O(\log m)$ time. Hence, for each vector $w_s\in \mathcal{W}$ we spend $O(\log m)$ time.
There are $O(n^2)$ vectors in $\mathcal{W}$ so the overall time of our algorithm is $O(n^2\log n)$.

\stitle{Extension to $d\geq 2$} The algorithm can straightforwardly be extended to any constant dimension $d$. Using known results~\cite{edelsbrunner1987algorithms}, we can construct the arrangement of $O(n)$ hyperplanes in $O(n^d\log n)$ time. Then, in $O(n^d\log n)$ time in total, we can traverse all combinatorially different vectors such that the orderings $P_{w_i}, P_{w_{i+1}}$ between two consecutive vectors $w_i, w_{i+1}$ differ by swapping the ranking of two consecutive items. Our algorithm is applied to all $O(n^d)$ vectors with the same way as described above. Hence, we conclude to the next theorem.

\begin{theorem}
\label{thm:2dOpt}
Let $P$ be a set of $n$ tuples in $\Re^d$. There exists an algorithm that computes an $(\optRank,1)$-hashmap satisfying the collision probability and the single fairness in $O(n^d\log n)$ time.
\end{theorem}


\stitle{Sampled vectors} \label{sec:ranking:vecSample}
So far, we consider all possible vectors $w\in \mathcal{W}$ to return the one with the optimum pairwise fairness. In practice, instead of visiting $O(n^d)$ vectors, we sample a large enough set of vectors $\widehat{\mathcal{W}}$ from $\Re^d$. We run our algorithm using the set of vectors $\widehat{\mathcal{W}}$ instead of $\mathcal{W}$ and we return the vector that leads to the minimum unfairness $\eps$. This algorithm runs in $O(|\widehat{\mathcal{W}}|n\log n)$.
\section{Cut-based Algorithms}
The ranking-based algorithms proposed in Section~\ref{sec:ranking}, cannot guarantee 0-unfairness. In other words, by re-ranking the $n$ points using linear projections in $\Re^d$, those can only achieve an $(\eps_R, 1)$-hashmap.

In this section, we introduce a new technique with the aim to {\em guarantee 0-unfairness}.
So far, our approach has been to partition the values (after projection) into $m$ equi-size buckets. In other words, each bucket $b_i$ is a continuous range of values specified by two boundary points.
The observation we make in this section is that the {\em buckets do not necessarily need to be continuous}. Specifically, we can
partition the values into more than $m$ bins while 
in a many-to-one matching, several bins are assigned to each bucket.
Using this idea, in the following we propose two approaches for developing fair hashmaps with 0-unfairness.


\subsection{\pd}
An interesting question we explore in this section is whether a cut-based algorithm exists that always guarantee 0-unfairness.

\begin{theorem}\label{th:pd}
Consider a set of $n$ points in $\Re$, where each point $p$ belongs to a group $g(p)\in\{\gee_1,\cdots,\gee_k\}$.
Independent of how the points are distributed and their orders, there always exist a cut-based hashmap that is {\bf 0-unfair}.
\end{theorem}

We prove the theorem by providing the \pd algorithm (Algorithm~\ref{alg:pd}) that always finds a 0-unfair hashmap.
Without loss of generality, let $L=\langle p_1, p_2, \cdots, p_n\rangle$ be the sorted list of points in $P$ based on their values on an attribute $x$.
\pd sweeps through $L$ from $p_1$ to $p_n$ twice.
During the first sweep (Lines~\ref{alg:pd:l2} to \ref{alg:pd:l4}), 
the algorithm keeps track of the number of instances it has observed from each group $\gee_i$. The algorithm uses $H^{tmp}$ to mark which bucket each point should fall into, such that each bucket contains $\frac{|\gee_i|}{m}$ instances from each group $\gee_i$.
During the second pass (Lines \ref{alg:pd:l6} to \ref{alg:pd:l11}), the algorithm compares the neighboring points and as long as those should belong to the same bucket (Line~\ref{alg:pd:l7}), there is no need to introduce a new boundary.
Otherwise, the algorithm adds a new boundary (in array $B$) to introduce a new bin, while assigning the bucket numbers in $H$.
Finally, the algorithm returns the bin boundaries and the corresponding buckets.

\begin{algorithm}[!tb]
    \caption{\pd} \label{alg:pd}
    \begin{algorithmic}[1] \small
    \Require{The set of points P}
    \Ensure{ Bin boundaries $B$ and corresponding buckets $H$}
        \State $\langle p_1,p_2,\cdots,p_n\rangle\gets$ {\bf sort} $P$ based on an attribute $x$
        \For{$j=1$ to $k$}
            $c_j\gets 0;$ {\tt\scriptsize // \# of instances observed from $\gee_i$}
        \EndFor
        \For{$i=1$ to $n$} \label{alg:pd:l2}
            \State let $\gee_j = \gee(p_i)$; $c_j\gets c_j+1$
            \State $H^{tmp}_i\gets \floor{\frac{c_j\times m}{|\gee_j|}}+1$ \label{alg:pd:l4}
        \EndFor
        \State $i\gets 0; j\gets 0$
        \While{{\bf True}} \label{alg:pd:l6}
            \While{$(i<n$ and $H^{tmp}_i==H^{tmp}_{i+1})$}
                $i\gets i+1$ \label{alg:pd:l7}
            \EndWhile
            \State $H_j\gets H^{tmp}_i$; {\tt\scriptsize // the bucket assigned to the $j$-th bin}
            \If{$i==n$} {\bf break} \EndIf
            \State $B_j=\frac{p_i[x]+p_{i+1}[x]}{2}$;  {\tt\scriptsize // the right boundary of the $j$-th bin}
            \State $j\gets j+1$ \label{alg:pd:l11}
        \EndWhile
        \State {\bf return} $(B,H)$
    \end{algorithmic}
\end{algorithm}

\subsubsection{Analysis}
\pd makes two linear-time passes over $P$. Therefore, considering the time to sort $P$ based on $x$, its time complexity is $O(n\log n)$.
The algorithm assigns $\frac{n}{m}$ point to each bucket; hence, following Propositions~\ref{prop:1} and \ref{prop:2} is satisfies collision probability and single fairness.
More importantly, \pd assigns $\frac{|\gee_i|}{m}$ point from each group $\gee_i$ to each bucket. As a result, for any pair $p_i$, $q_i$ in $\gee_i$, $Pr[h(p_i)=h(q_i)]=\frac{1}{m}$. Therefore, the hashmap generated by \pd is {\bf 0-unfair}, proving Theorem~\ref{th:pd}.

Despite guaranteeing 0-unfairness, \pd is not efficient in terms of memory. Particularly, in the worst case, it can introduce as much as $O(n)$ boundaries.

In a best case, where the points are already ordered in a way that dividing them into $m$ equi-size buckets is already fair, \pd will add $m$ bins ($m-1$ boundaries).
On the other hand, in adversarial setting, a large potion of the neighboring pairs within the ordering belong to different groups with different hash buckets assigned to them. Therefore, in the worst-case \pd may add up to $O(n)$ boundaries, making it satisfy $\frac{n}{m}$-{\bf memory}.
Applying the binary search on the $O(n)$ bin boundaries, the query time of \pd hashmap is in the worst-case $O(\log n)$.

\begin{lemma}\label{th:pd:mem:avg}
    In the binary demographic groups cases, where $\Gee=\{\gee_1, \gee_2\}$ and $r=|\gee_1|$, the expected number of bins added by \pd is bounded by $2\big(\frac{r(n-r)}{n}+m\big)$. 
\end{lemma}
\begin{proof}
    {
    \pd adds at most $m-1$ boundaries between the neighboring pairs that both belong to the same group, simply because 
    a boundary between neighboring pair can only happen when moving from one bucket to the next while there are $m$ buckets.

    In order to find the upper-bound on the number of bins added, in the following we compute the expected number of neighboring pairs from different groups.
    Consider the in the sorted list of points $\langle p_1,\cdots,p_n \rangle$ based on the attribute $x$.
    The probability that a point $p_i$ belongs to group $g_1$ is $Pr_1=Pr(\gee(p_i)=\gee_1)=\frac{r}{n}$.
    Now consider two consecutive points $p_i,p_{i+1}$, in the list. The probability that these two belong to different groups is $2Pr_1(1-Pr_1)$.
    Let $\mathcal{B}$ be the random variable representing the number of pairs from different groups.
    We have,
    $
    E[\mathcal{B}] = \sum_{i=1}^{n-1} 2Pr_1(1-Pr_1) = 2(n-1)\frac{r}{n} (1-\frac{r}{n})<\frac{2r(n-r)}{n}
    $.
    Therefore, $E[\mbox{No. bins}] \leq E[\mathcal{B}] + 2m < 2\Big(\frac{r(n-r)}{n}+m\Big)$.
    }
\end{proof}



\subsection{Transforming to Necklace Splitting}
Although guaranteeing 0-unfairness, \pd is not memory efficient.
The question we seek to answer in this section is whether it is possible to guarantee 0-unfairness while introducing significantly less number of bins, close to $m$.
In particular, we make an interesting connection to the so-called {\em necklace splitting} problem~\cite{meunier2014simplotopal,alon1987splitting,alon1986borsuk}, and use the recent advancements~\cite{alon2021efficient} on this problem by the math and theory community to solve the fair hashmap problem.

\stitle{(Review) Necklace Splitting} 
Consider a necklace of $T$ beads of $n'$ types. 
For each type $i\leq n'$, let $m_i'$ be the number of beads with type $i$, and let $m'=\max_{i}m_i'$.
The objective is to divide the beads between $k'$ agents, such that (a) all agents receive exactly the same amount of beads from each type and (b) the number of splits to the necklace is minimized.



\stitle{Reduction} 
Given an instance of the fair hashmap problem, let $L=\langle p_1, p_2, \cdots, p_n\rangle$ be the ordering of $P$ based on an attribute $x$. 
The problem gets reduced to necklace splitting as following: 
The points in $P$ get mapped into the $T=n$ bead, distributed with the ordering $\langle p_1, p_2, \cdots, p_n\rangle$ in the necklace.
The $k$ groups in $\Gee$ get translated to the $n'=k$ bead types $\{\gee_1,\cdots, \gee_k\}$.
The $m$ buckets translate to the $k'=m$ agents.

Given the necklace splitting output, each split of the necklace translates into a bin.
The bin is assigned to the corresponding bucket of the agent who received the necklace split.
Using this reduction, an optimal solution to the necklace splitting is the fair hashmap with minimum number of cuts: first, since all party in necklace splitting receive equal number of each bead type, the corresponding hashmap satisfies {\bf 0-unfairness}, as well as the collision probability and single fairness requirements;
second, since the necklace splitting minimizes the number of splits to the necklace, it adds minimum number of bins to the fair hashmap problem, i.e., it finds the optimal fair hashmap on $L$, with minimum number of cuts.
Using this mapping, in the rest of the section we adapt the recent results for solving necklace splitting for fair hashmap.
In particular, Alon and Graur~\cite{alon2021efficient} propose polynomial time algorithms (with respect to number of beads) for the Necklace Splitting problem and the $\eps$-approximate version of the problem. 

\subsubsection{Binary groups}
The fair hashmap problem when there are two groups $\{\gee_1,\gee_2\}$, maps to the necklace splitting instance with two bead types.
While a straightforward implementation of the algorithm in~\cite{alon2021efficient}-(Proposition 2) leads to an $O(n(\log n+m))$ algorithm, in \necklaceb Algorithm~\ref{alg:nl2d}, we propose an optimal time algorithm that guarantees splitting a necklace with {\em at most $2(m-1)$} cuts in only $O(n\log n)$ time.

\stitle{Algorithm}
Without loss of generality, let $C=\langle p_1, p_2, \cdots, p_n\rangle$ be the sorted list of points in $P$ based on their values on an attribute $x$.
\necklaceb views $C$ as a circle by considering $p_n$ before $p_1$. It uses modulo to size of the list ($\%|C|$) to move along the circle.
The key idea is that the circle $C$ {\em always} has at least one consecutive window of size $\frac{n}{m}$ that contains $\frac{|\gee_1|}{m}$ points from $\gee_1$ (and hence $\frac{|\gee_2|}{m}$ points from $\gee_2$), see~\cite{alon2021efficient}.
Hence, we design an algorithm to find such windows efficiently.
We initialize a list $T$ such that $T[j]$ contains the number of items from group $\gee_1$ between $C[j]$ and $C[(j+n/m-1)\%|C|]$.
Furthermore, we initialize the list $X$ such that $X[j]$ is true if and only if $T[j]=|\gee_1|/m$, i.e., the window from $C[j]$ to $C[(j+n/m-1)\%|C|]$ is a good candidate for a cut.
All indexes are initialized in lines \ref{l1a}--\ref{l1b} of Algorithm~\ref{alg:nl2d}.
In order to bound the running time of the new algorithm, we assume that $X[j]$ also stores a pointer to the $j$-th elements in lists $C$ and $T$. Furthermore, we assume that all Boolean variables in $X$ are stored in a max heap $M_X$(if $X[j]=\text{true}$ and $X[i]=\text{false}$ then $X[j]>X[i]$). We use $M_X$ to call $M_X.top()$ that returns the top item in max heap, i.e., it returns a $j$ such that $X[j]=\text{true}$, in $O(1)$ time. 

The algorithm is executed in iterations until the list $C$ is non-empty. In each iteration, we find a window of size $n/m$ containing exactly $|\gee_1|/m$ items from group $\gee_1$ (so it also contains exactly $|\gee_2|/m$ items from group $\gee_2$). More specifically, in line~\ref{l:MaxHeap} we find $j$ such that $T[j]=|\gee_1|/m$. In line~\ref{l:cut} we define the new cut we find and in lines \ref{l:removea}--\ref{l:removeb} we remove the cut from our lists to continue with the next iteration. The points within the cut are marked for the bucket $bkt$. In lines \ref{l:updatea}--\ref{l:updateb} we update the values of $T$ (and hence the values of $X$) so that $T$ and $X$ have the correct values after removing the window from $C[j]$ to $C[(j+n/m-1)\%|C|]$.
Hence, we can continue searching for the next window containing $|\gee_1|/m$ tuples from group $\gee_1$ in the next iteration.
Finally, it sorts the discovered cuts and assigns the bin boundaries $B$ and the corresponding buckets $H$.

\begin{theorem}
\label{thm:neckSplitk2}
In the binary demographic group cases, there exists an algorithm that finds a $(0,2)$-hashmap satisfying the collision probability and the single fairness in $O(n\log n)$ time.
\end{theorem}
\begin{proof}
In each iteration of the algorithm, we remove a window containing $|\gee_1|/m$ items from group $\gee_1$ and $|\gee_2|/m$ items from group $\gee_2$. Hence, the correctness of our algorithm follows by the discrete intermediate value theorem and~\cite{alon2021efficient}.
Next, we show that the running time is $O(n\log n)$.
It takes $O(n\log n)$ to sort based on attribute $x$. Then it takes $O(n/m+n)=O(n)$ to initialize $T$ and  $X$. In each iteration of the while loop at line~\ref{l:bigLoop} we remove $n/m$ items, so in total it runs for $m$ iterations. In each iteration, we get $j$ at line~\ref{l:MaxHeap} in $O(1)$ time using the max heap. The for loop in line~\ref{l:updatea} is executed for $O(n/m)$ rounds. In each round, we need $O(1)$ time to update $T$. It also takes $O(1)$ time to update a value in $X$, and $O(\log n)$ time to update the max heap. Finally, the for loop in line~\ref{l:removea} runs for $O(n/m)$ rounds. In each round it takes $O(1)$ time to remove an item from lists $C, T, X$ and $O(\log n)$ time to update the max heap. Overall Algorithm~\ref{alg:nl2d} runs in $O(n\log n + m\frac{n}{m}\log n)=O(n\log n)$ time.
\end{proof}


\begin{algorithm}[!tb]
    \caption{\necklaceb} \label{alg:nl2d}
    \begin{algorithmic}[1] \small
    \Require{The set of points $P$ (with two groups $\{\gee_1,\gee_2\}$)}
    \Ensure{Bin boundaries $B$ and corresponding buckets $H$}
        \State $C=\langle p_1,p_2,\cdots,p_n\rangle\gets$ {\bf sort} $P$ based on an attribute $x$;
        \State {\bf for $i= 0$ to $n-1$} {\bf do} $\{T[i]\gets 0; X[i]\gets \text{false}\}$
        \State $M_X\gets $ max heap storing $X$;
        $\sigma_1\gets 0$; $bkt\gets 0$;
        \Statex{\tt\scriptsize //Initialize $T$}
        \State {\bf for $i= 0$ to $n/m$} {\bf do} ~\{{\bf if} $\gee(C[i])==\gee_1$ {\bf then} $\sigma_1\gets \sigma_1+1$\} \label{l1a}
        \For{$i= 0$ to $n-1$}
            \State $T[i]\gets\sigma_1$;
            \If{$T[i]==|\gee_1|/m$}
                $X[i]\gets\text{true}$; $M_X.\text{update}(X[i])$;
            \EndIf
            \If{$\gee(C[i])==\gee_1$}
                $\sigma_1\gets\sigma_1-1$;
            \EndIf
            \If{$\gee(C[(i+n/m) \% |C|])==\gee_1$}
                $\sigma_1\gets\sigma_1+1$;
            \EndIf
        \EndFor\label{l1b}
        \While{$|C|>0$}\label{l:bigLoop}
        \State $j\leftarrow M_x.top()$;\label{l:MaxHeap}
            {\tt\scriptsize //Find a window $[j,(j+n/m-1)\%|C|]$ with $X[j]=\text{true}$}
            \Statex{\tt\scriptsize \hspace{5mm}//Update $T$ and $X$ removing the window $[j,(j+n/m-1])\%|C|]$}
             \State $\sigma\gets T[(j+n/m)\%|C|]$;
            \For{$i=j-1$ to $j-n/m+1$ with step $-1$}\label{l:updatea}
               \If{$i<0$}
                     $i\gets |C|+i$;
                \EndIf
                \If{$\gee(C[(i+2\cdot n/m)\%|C|])==\gee_1$}\label{l:updatec}
                    $\sigma\gets \sigma-1$;
                \EndIf
                \If{$\gee(C[i])==\gee_1$}
                    $\sigma\gets \sigma+1$;
                \EndIf
                \State $T[i]\gets \sigma$;
                \State $X[i]\gets \text{true}$ {\bf if} $(T[i]==|\gee_1|/m)$ {\bf else} false
                \State $M_X.update(X[i])$;
            \EndFor\label{l:updateb}
            \State cuts $\gets$ cuts $\cup \big\{j, (j+n/m-1)\% |C|\big\}$\label{l:cut}
            \Statex{\tt\scriptsize \hspace{5mm}//Remove window $[j,(j+n/m-1)\%|C|]$}
            \For{$i\in[0,n/m)$}\label{l:removea}
                \State $H_{C[(i+j)\%|C|]}^{tmp}\leftarrow bkt$;
                 $Remove(C[(i+j)\%|C|])$;
                \State $Remove(T[(i+j)\%|C|])$;
                $Remove(X[(i+j)\%|C|])$;
            \EndFor\label{l:removeb}
            \State $bkt\leftarrow bkt+1$;
        \EndWhile
        \State {\bf sort}$($cuts$)$
        \For{$j=0$ to |cuts|}
            \State Let $p_i$ be the rightmost tuple in the $j$-th bin;
            \State $B_j\gets \frac{p_i[x]+p_{i+1}[x]}{2}$;  {\tt\scriptsize // the right boundary of the $j$-th bin}
            \State $H_j\gets H^{tmp}_i$; {\tt\scriptsize // the bucket assigned to the $j$-th bin}
        \EndFor
        \State {\bf return} $(B,H)$
    \end{algorithmic}
\end{algorithm}

\section{Discrepancy-based hashmaps}
\label{sec:extended}
{
A hashmap satisfies $\gamma$-discrepancy if and only if each bucket contains at least $(1-\gamma)\frac{|\gee_i|}{m}$ and at most $(1+\gamma)\frac{|\gee_i|}{m}$ points from each group $\gee_i$. In this section, we first show that a hashmap that satisfies $\gamma$-discrepancy has bounded collision probability, single fairness, and pairwise fairness. Then, we propose efficient algorithms that construct $\gamma$-discrepancy hashmaps, where $\gamma$ is close to the optimum.

Let $P_w$ be the ordering of points in $P$ based on a vector $w$ and let $\hashmap$ be a hashmap constructed on $P_w$. Recall that $Pr_i$ is defined as the pairwise fairness value of $\hashmap$ for group $\gee_i$. Let $Cp$ be the collision probability and $Sp_i$ the single fairness of group $\gee_i$.

\begin{lemma}
\label{lem:discrProbs}
    Let $\hashmap$ be a hashmap satisfying $\gamma$-discrepancy. Then $Cp\leq \frac{1+\gamma}{m}$,  $\frac{1-\gamma}{m}\leq Sp_i\leq \frac{1+\gamma}{m}$ and $Pr_i\leq \frac{1+\gamma}{m}$ for each group $\gee_i$.
\end{lemma}
\begin{proof} 
{
Recall that $n_j$ is the number of items in bucket $j$ and $\alpha_{i,j}$ is the number of items from group $i$ in bucket $j$.
Notice that $\sum_{j=1}^m\alpha_{i,j}=|\GroupPoints_i|$ and $\sum_{j=1}^m n_j=n$.
We have,

\begin{align*}
    Pr_i=\sum_{j=1}^m\left(\frac{\alpha_{i,j}}{|\GroupPoints_i|}\right)^2\leq \sum_{j=1}^m\frac{(1+\gamma)\frac{\GroupPoints_i}{m}}{|\GroupPoints_i|}\cdot\frac{\alpha_{i,j}}{|\GroupPoints_i|}=
    \frac{1+\gamma}{m}\sum_{j=1}^m\frac{\alpha_{i,j}}{|\GroupPoints_i|}=\frac{1+\gamma}{m}.
\end{align*}

Similarly we show that 
    \[
    Cp=\sum_{j=1}^m\left(\frac{n_j}{n}\right)^2\leq \frac{1+\gamma}{m}.
    \]
Using the same arguments it also holds that
$$Sp_i=\sum_{j=1}^m\frac{\alpha_{i,j}}{|\GroupPoints_i|}\frac{n_j}{n}\leq \frac{1+\gamma}{m}\quad\text{ and } \quad Sp_i\geq \frac{1-\gamma}{m}.$$
}
\end{proof}

{
Next, we describe a dynamic programming algorithm to find a hashmap with the smallest discrepancy. In the appendix, we show a faster randomized algorithm approximating the smallest discrepancy.
Finally, we describe a simple heuristic that works as a post-processing method to further improve the pairwise fairness.

\subsection{Dynamic Programming algorithm}
Let $P_w$ be the ordering of the items for a vector $w\in \mathcal{W}$.
Let $\mathsf{Disc}[i,j]$ be the discrepancy of the optimum partition among the first $i$ items in $P_w$ using $j$ buckets. Let also $\mathcal{D}(a,b)$ be the discrepancy of the bucket including all items in the window $[a,b]$ in $P_w$, i.e., $\{P_w[a],P_w[a+1],\ldots, P_w[b]\}$.
We define the recursive relation
$$\mathsf{Disc}[i,j]=\min_{1\leq x<i}\max\{\mathsf{Disc}[x-1,j-1], \mathcal{D}(x,i)\}$$
Given $i, j$, our algorithm computes the $j$-th bucket with right boundary $i$, trying all left boundaries $x<i$, that leads to a partition with minimum discrepancy over the first $i$ items.
In order to efficiently implement the algorithm, each time we try a new left boundary $x$, we do not compute $\mathcal{D}(x,i)$ from scratch. Instead, we maintain and update a max-heap of size $k$ storing the discrepancy of every group $\gee_i$ in the window $[x,i]$. When we compute $[x-1,i]$ we update one element in the max-heap and compute $\mathcal{D}(x-1,i)$ in $O(\log k)$ time. The table $\mathsf{Disc}$ has $O(nm)$ cells and for each cell we spend $O(n\log k)$ time. 
By definition, $\mathsf{Disc}[n,m]$ computes $\optDisc$. Doing standard modifications, it is straightforward to return the partition, instead of the discrepancy $\optDisc$.
By repeating the algorithm above for every $w\in \mathcal{W}$, we conclude to the next theorem.

\begin{theorem}
\label{thm:2dOpt}
Let $P$ be a set of $n$ tuples in $\Re^d$. There exists an algorithm that computes an $(\optDisc,1)$-hashmap in $O(n^{d+2}m\log n\log k)$ time,
satisfying $\optDisc$-approximation in collision probability and single fairness.
\end{theorem}

For parameters $\gamma, \delta$, in the appendix, we show a randomized algorithm, to compute a $((1+\delta)\optDisc+\gamma,1)$-hashmap with collision probability at most $\frac{1+(1+\delta)\optDisc+\gamma}{m}$
and single fairness in the range $[\frac{1-(1+\delta)\optDisc-\gamma}{m},\frac{1+(1+\delta)\optDisc+\gamma}{m}]$, in time $O(n+\poly(m,k,\delta, \gamma)$.

\subsection{Local-search based heuristic}
So far, we consider algorithms that return a hashmap satisfying (approximately) $\optDisc$-discrepancy. From Lemma~\ref{lem:discrProbs} we know that a hashmap satisfying $\optDisc$-discrepancy is a $(\optDisc,1)$-hashmap. However, there is no guarantee that $\optDisc\leq\optRank$. 

In this section, we design a practical algorithm that returns an $(\eps,1)$-hashmap with $\eps\leq \optRank$ allowing a slight increase in single fairness (and collision probability). In practice, as we see in Section~\ref{sec:experiments}, it holds that $\eps\ll \optRank$. The new algorithm is a local-search based algorithm and works as a post-processing procedure to any ranking-based algorithm (for example Algorithm~\ref{alg:ranking2d}).
The intuition is the same to other discrepancy-based algorithms: The fact that we use the same number of items per bucket, 
restricts our options to compute a fair hashmap. Given the buckets computed by a ranking-based algorithm, we try to (slightly) modify the boundaries of the buckets to compute a new hasmap with smaller unfairness.

The high level idea is that in each iteration of the algorithm we slightly move one of the boundaries that improves the unfairness the most, maintaining a sufficient single fairness and collision probability. Let $T$ be the maximum number of iterations we execute our algorithm, and let $f^-, f^+, c^+$ be the minimum single fairness, the maximum single fairness, and the maximum collision probability, respectively, that the returned hashmap should satisfy.
Let $B$ be the boundaries returned by Algorithm~\ref{alg:ranking2d}. For an iteration $i\leq T$, for every boundary $B_j\in B$ we move $B_j$ one position to the left or to the right. For each movement of the boundary $B_j$, we compute the unfairness $\eps_j$, single fairness $f_j$, and collision probability $c_j$ of the new partition. If $f^-\leq f_j\leq f^+$ and $c_j\leq c^+$ then this is a valid partition/hashmap satisfying the requested single fairness and collision probability. In the end of each iteration, we modify the boundary $B_{j^*}$ that leads to a valid hashmap with the smallest unfairness, i.e., $j^*=\argmin_{j:f^-\leq f_j\leq f^+, c_j\leq c^+}\eps_j$.

By definition, in each iteration, we find a partition having at most the same unfairness as before. In practice, we expect to find a hashmap with much smaller unfairness. This is justified in our experiments, Figure~\ref{fig:adult_local_search}. For the running time, the algorithm runs in $T$ iterations. In each iteration, we go through all the $O(m)$ boundaries, and we compute $\eps_j, f_j, c_j$. Using a max heap to maintain the unfairness for every group $\gee_i$ (similar to the max-heap we used in the previous dynamic programming algorithm) we can compute the unfairness $\eps_j$ in $O(\log k)$ time.
We need the same time to compute $f_j$ and $c_j$.
Our algorithm runs in $O(n+Tm\log k)$.
}

{
Due to space limitations, we show the cut-based algorithm that finds a $\big(\eps,k(4+\log \frac{1}{\eps})\big)$-hashmap in the appendix.
}

\section{Experiments}
\label{sec:experiments}
In addition to the theoretical analysis, we conduct extensive experiments on a variety of settings to confirm the fairness and memory/time efficiency of our proposed algorithms. In short, aligned with the theoretical guarantees shown in the previous sections, the results of our experiments demonstrate the effectiveness and efficiency of our algorithms in real-world settings.
\subsection{Experiments Setup}
The experiments were conducted on a 3.5 GHz Intel Core i9 processor, 128 GB memory, running Ubuntu. The algorithms were implemented in Python 3.

For evaluation purposes, we used three real-world and one semi-synthetic datasets to evaluate our algorithms. With the importance of the scalability of our proposed methods to large settings in mind, we chose datasets that are large enough to represent real-world applications. For each dataset, we selected the two columns that were most uncorrelated to construct the hashmaps. The values in either column are normalized to be in the $[0,1]$ range. As the sensitive attribute, we follow the existing literature on group fairness and study fairness over demographic information such as {\em sex} and {\em race}. A summary of the datasets is presented in Table~\ref{tbl:datasets}. For a more detailed description of the datasets, refer to the appendix.

\subsection{Evaluation Plan}
We evaluate our proposed algorithms based on three metrics: 1) unfairness, 2) space, and 3) efficiency (preprocessing time and query time). 
For each of the above metrics, we study the effect of varying three variables: dataset size $n$, minority-to-majority ratio, and number of buckets $m$. In our experiments, we vary $n$ from 0.2 to 1.0 fraction of the original dataset with an increasing step of 0.2. We vary $m$ from 100 to 1000 increasing by 100 at each step and finally, the minority-to-majority ratio from 0.25 to 1.0 increasing by 0.25 at each step.
Throughout our experiments, while varying a variable, we fix the others as follows: $\text{dataset size }n=0.2\times|\text{original dataset}|$;
     $\text{number of buckets }m=100$; 
     $\text{minority-to-majority ratio}=0.25$.
Due to the space limitations, we present the results for two datasets for each setting and present the extended results in the appendix.
Particularly, for fairness and space evaluation, we report the results for \compas and \adult, the fairness benchmark datasets. For run-time evaluation though, we report on the larger-scale datasets, \diabetes and \popsim.
We confirm that we obtained similar results for all datasets.

\subsubsection{Evaluated Algorithms} In our experiments, we evaluate \ranker, \necklaceb, \pd algorithms, and \cdf hashmap \cite{kraska2018case} (referred as \fag) as the baseline, for all of the datasets using the binary sensitive attributes ({\tt sex} or binary {\tt race}). 
We also evaluate the \ranker and \pd algorithms using \compas dataset with {\tt race} attribute to demonstrate that our algorithms extend to non-binary sensitive attributes.
Given the potentially large number of rankings, we use the sampled vectors approach for \ranker.
Particularly, we report the results based on two samples of vectors of size 100 and 1000.
Finally, we evaluate the effectiveness of our local-search based heuristic on the output the \ranker algorithm.

\begin{figure*}[!tb] 
\centering
    \includegraphics[width=0.9\textwidth]{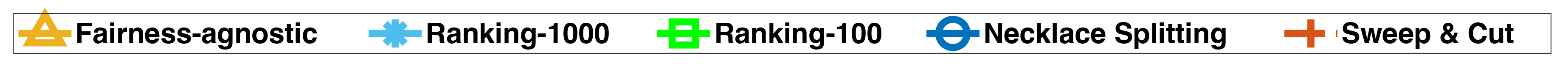}
    \begin{minipage}[t]{0.24\linewidth}
        \centering
        \includegraphics[width=\textwidth]{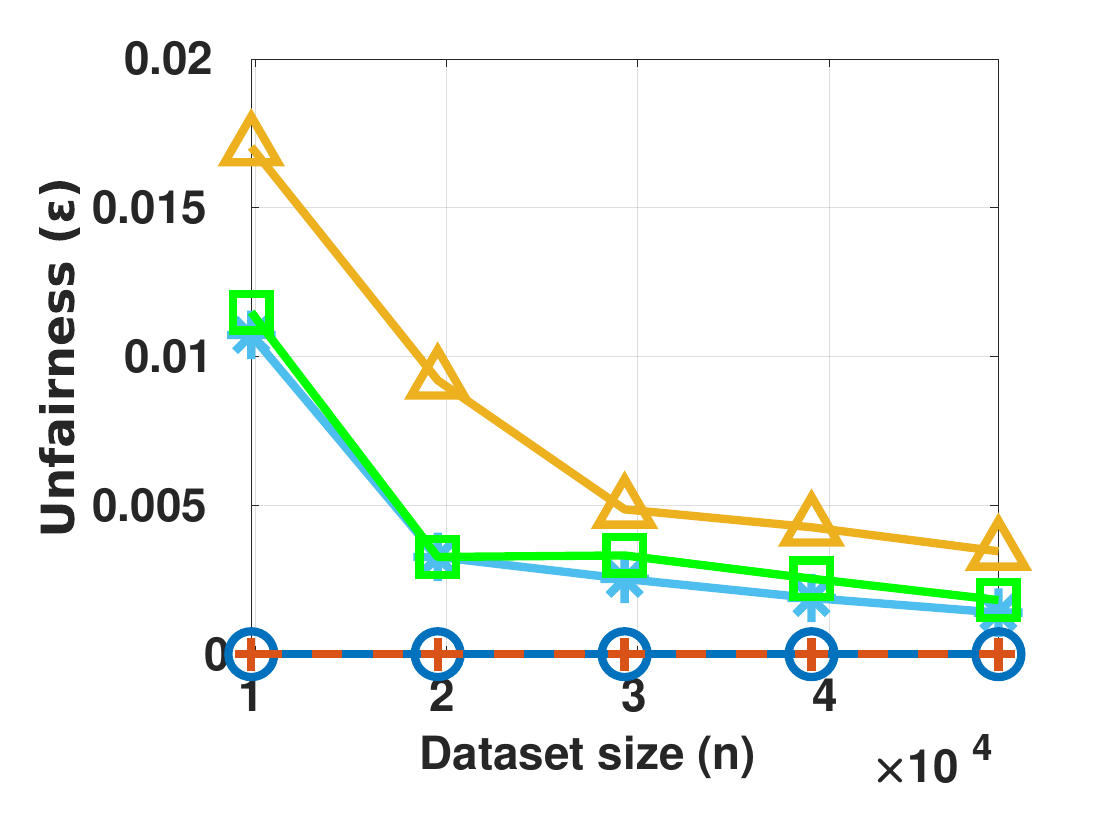}
        \vspace{-2.5em}
        \caption{{Effect of varying dataset size $n$ on unfairness, \adult, {\tt sex}}}
        \vspace{-2.5em}
        \label{fig:adult_n_vs_unfairness}
    \end{minipage}
    \hfill
    \begin{minipage}[t]{0.24\linewidth}
        \centering
        \includegraphics[width=\textwidth]{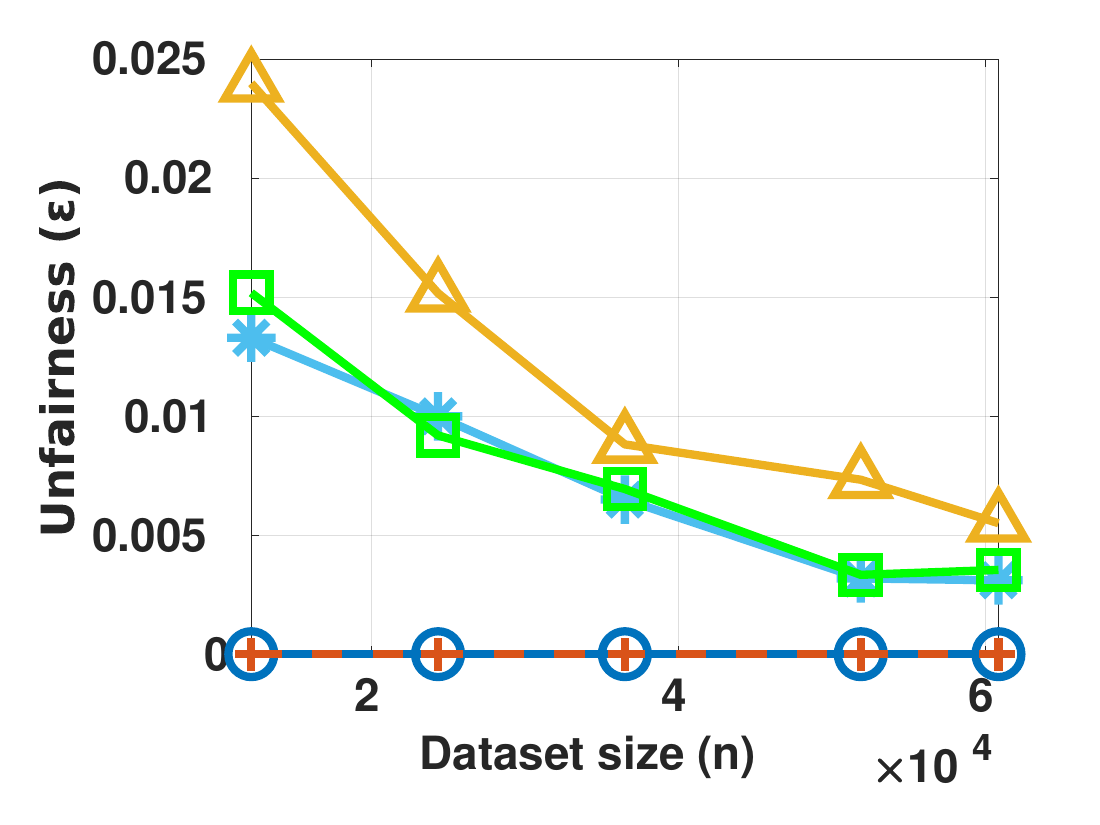}
        \vspace{-2.5em}
        \caption{Effect of varying dataset size $n$ on unfairness, \compas, {\tt sex}}
        
        \label{fig:compas_n_vs_unfairness}
    \end{minipage}
    \hfill
    \begin{minipage}[t]{0.24\linewidth}
        \centering
        \includegraphics[width=\textwidth]{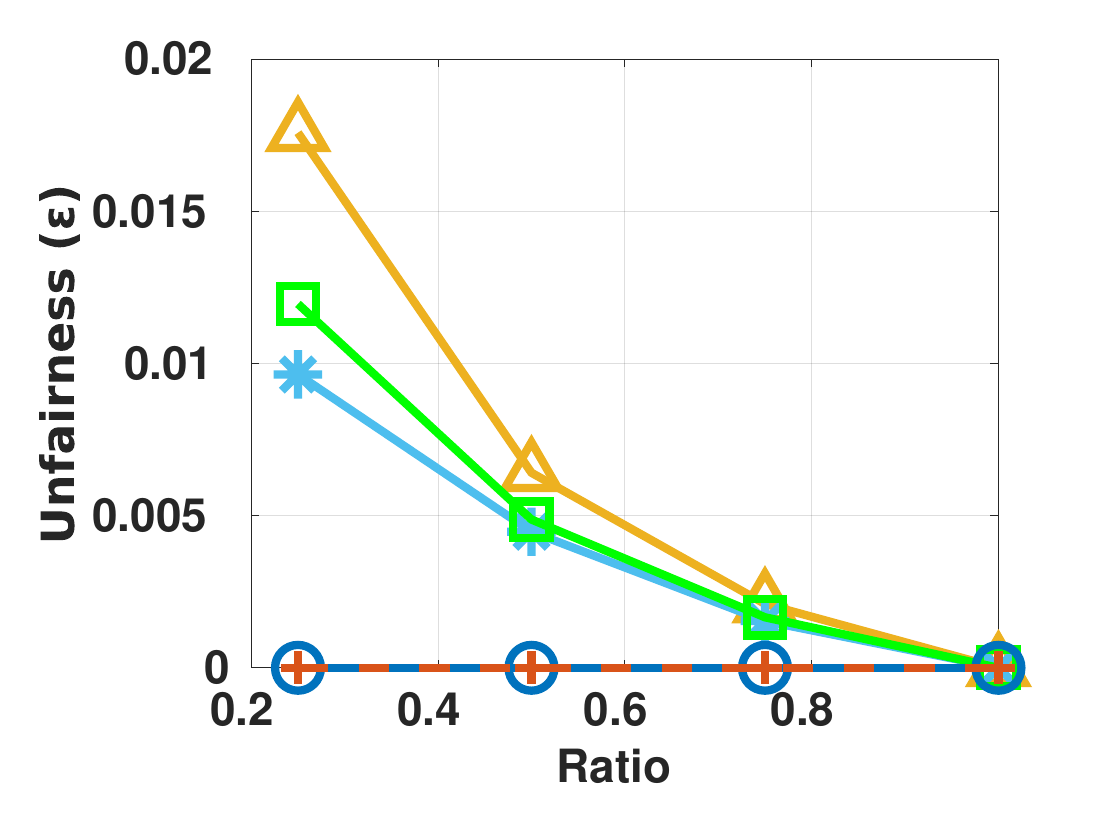}
        \vspace{-2.5em}
        \caption{Effect of varying minority-to-majority ratio on unfairness, \adult, {\tt sex}}
        
        \label{fig:adult_ratio_vs_unfairness}
    \end{minipage}
    \hfill
    \begin{minipage}[t]{0.24\linewidth}
        \centering
        \includegraphics[width=\textwidth]{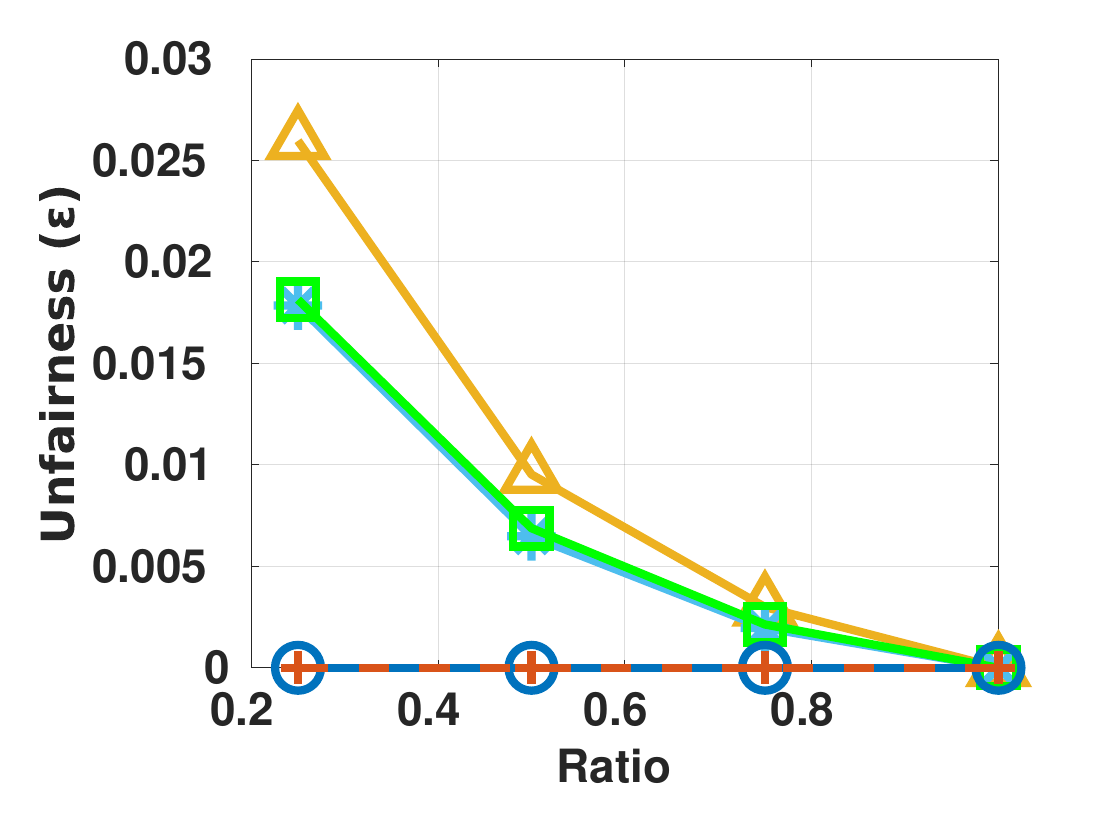}
        \vspace{-2.5em}
        \caption{Effect of varying minority-to-majority ratio on unfairness, \compas, {\tt sex}}
        
        \label{fig:compas_ratio_vs_unfairness}
    \end{minipage}
    \hfill 
    \begin{minipage}[t]{0.24\linewidth}
        \centering
        \includegraphics[width=\textwidth]{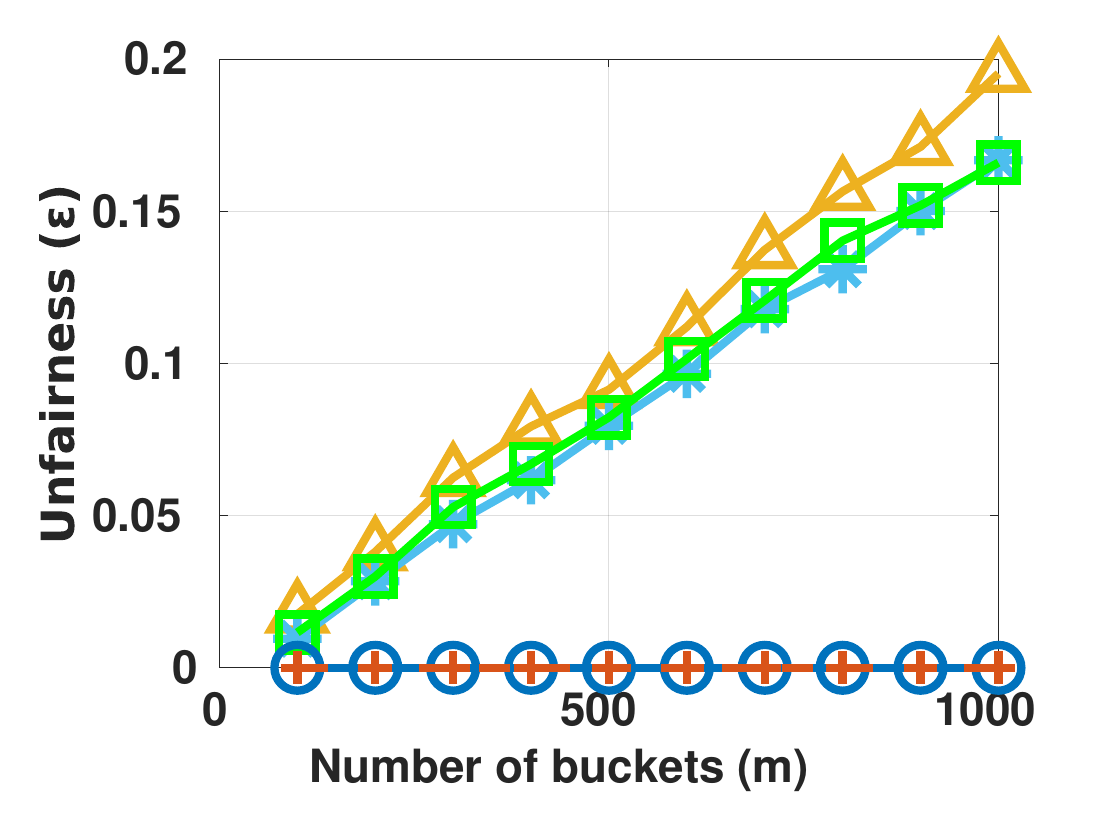}
        \vspace{-2.5em}
        \caption{Effect of varying number of buckets $m$ on unfairness, \adult, {\tt sex}}
        
        \label{fig:adult_m_vs_unfairness}
    \end{minipage}
    \hfill
    \begin{minipage}[t]{0.24\linewidth}
        \centering
        \includegraphics[width=\textwidth]{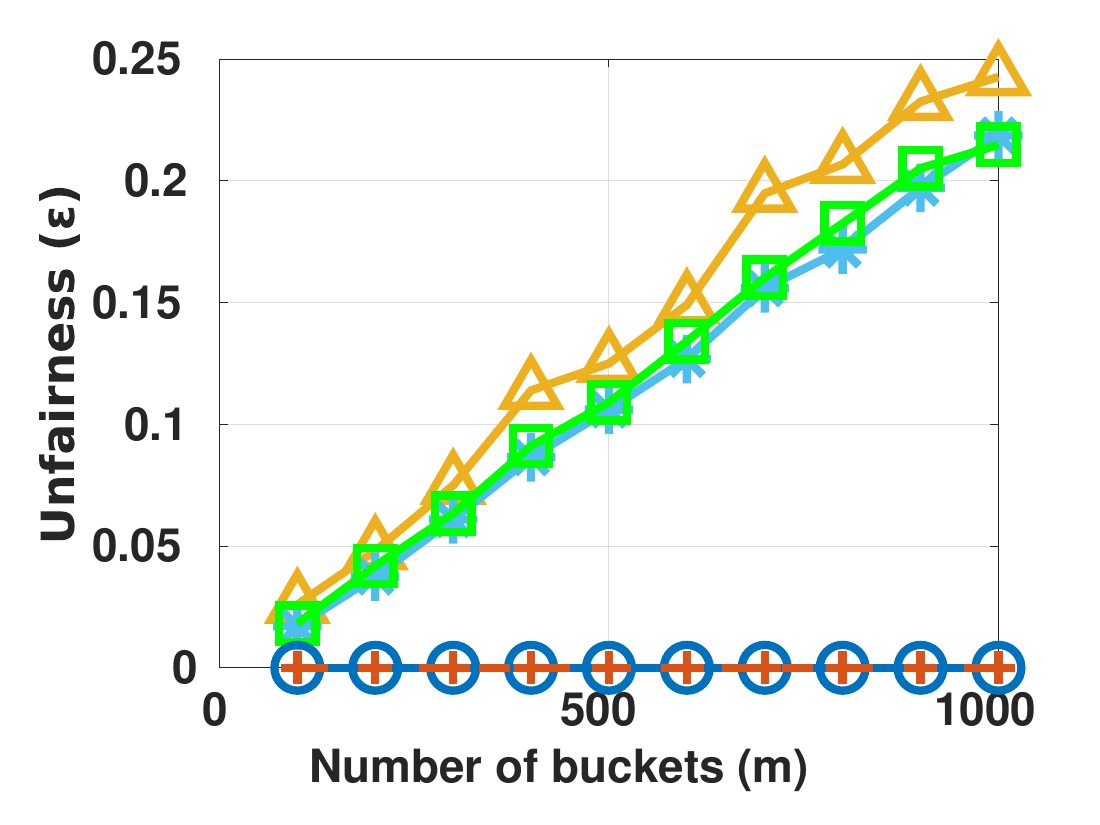}
        \vspace{-2.5em}
        \caption{{Effect of varying number of buckets $m$ on unfairness, \compas, {\tt sex}}}
        
        \label{fig:compas_m_vs_unfairness}
    \end{minipage}
    \hfill
    \begin{minipage}[t]{0.24\linewidth}
        \centering
        \includegraphics[width=\textwidth]{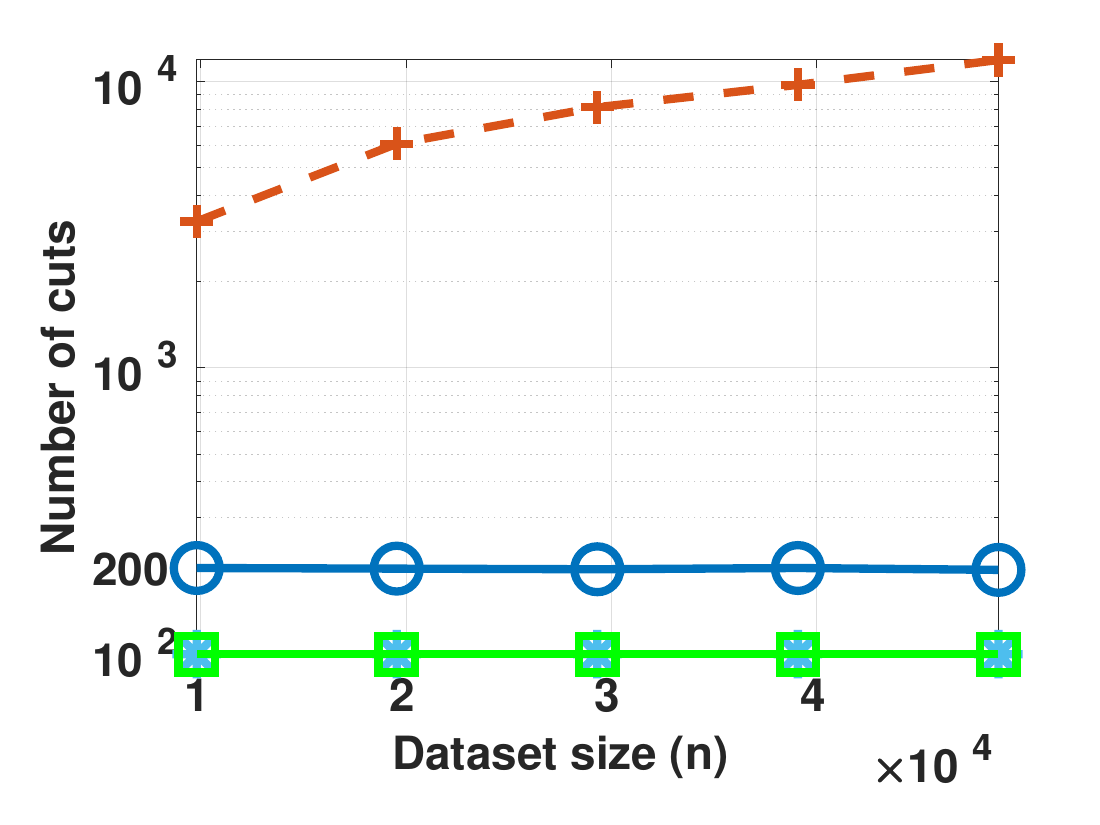}
        \caption{Effect of varying dataset size $n$ on space, \adult}
        
        \label{fig:adult_n_vs_space}
    \end{minipage}
    \hfill
    \begin{minipage}[t]{0.24\linewidth}
        \centering
        \includegraphics[width=\textwidth]{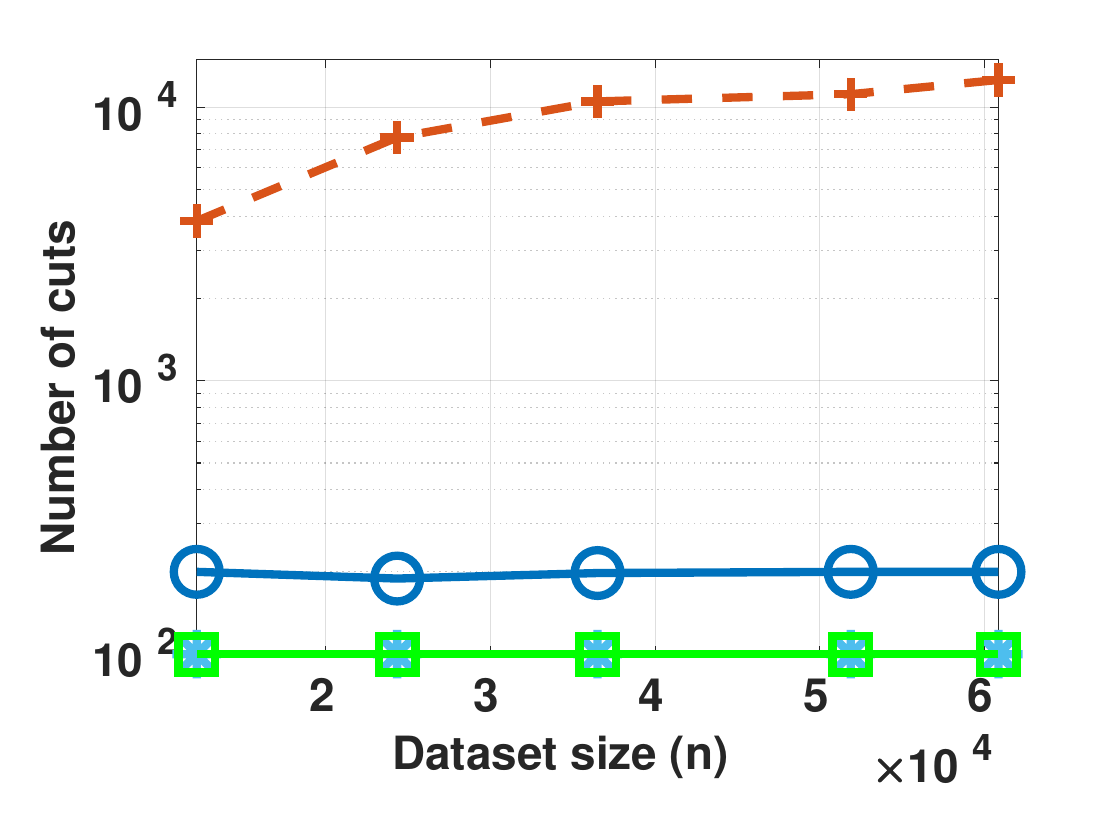}
        \caption{{Effect of varying dataset size $n$ on space, \compas}}
        
        \label{fig:compas_n_vs_space}
    \end{minipage}
    \hfill
    \begin{minipage}[t]{0.24\linewidth}
        \centering
        \includegraphics[width=\textwidth]{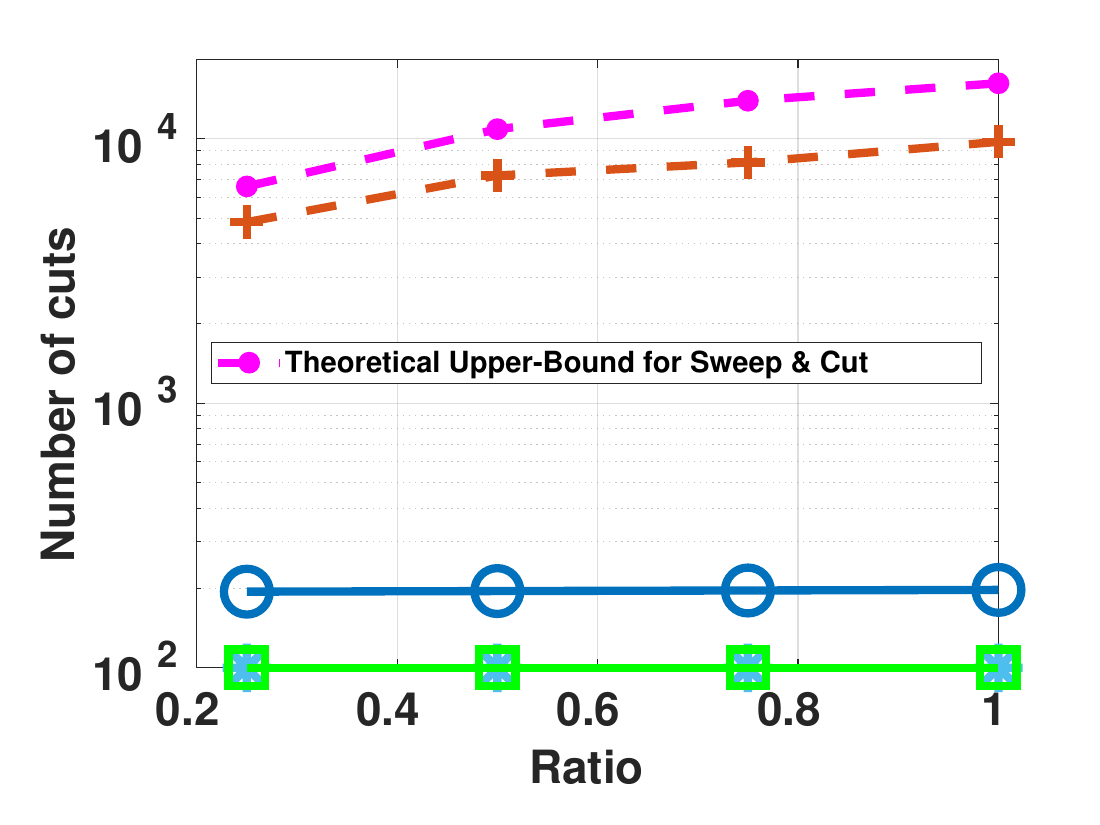}
        \caption[]{Effect of varying minority-to-majority ratio on space, \adult}
        
        \label{fig:adult_ratio_vs_space}
    \end{minipage}
    \hfill
    \begin{minipage}[t]{0.24\linewidth}
        \centering
        \includegraphics[width=\textwidth]{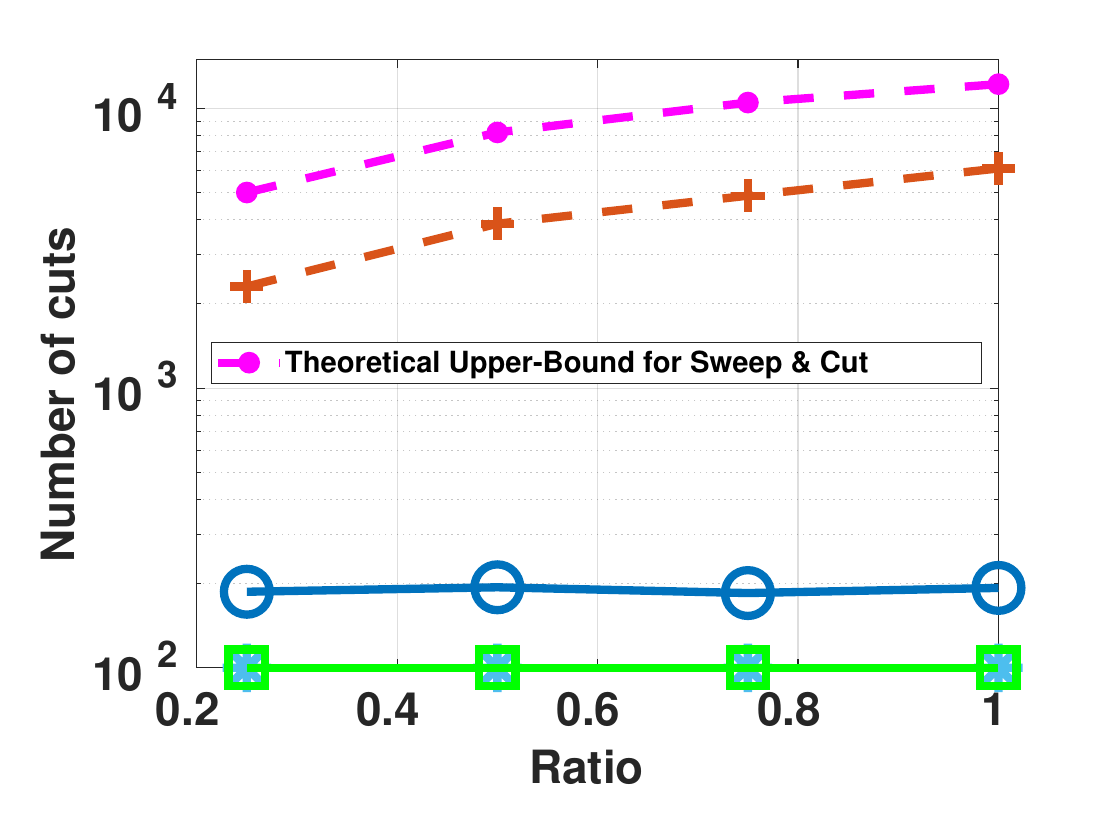}
        \caption[]{{Effect of varying minority-to-majority ratio on space, \compas}}
        
        \label{fig:compas_ratio_vs_space}
    \end{minipage}
    \hfill
    \begin{minipage}[t]{0.24\linewidth}
        \centering
        \includegraphics[width=\textwidth]{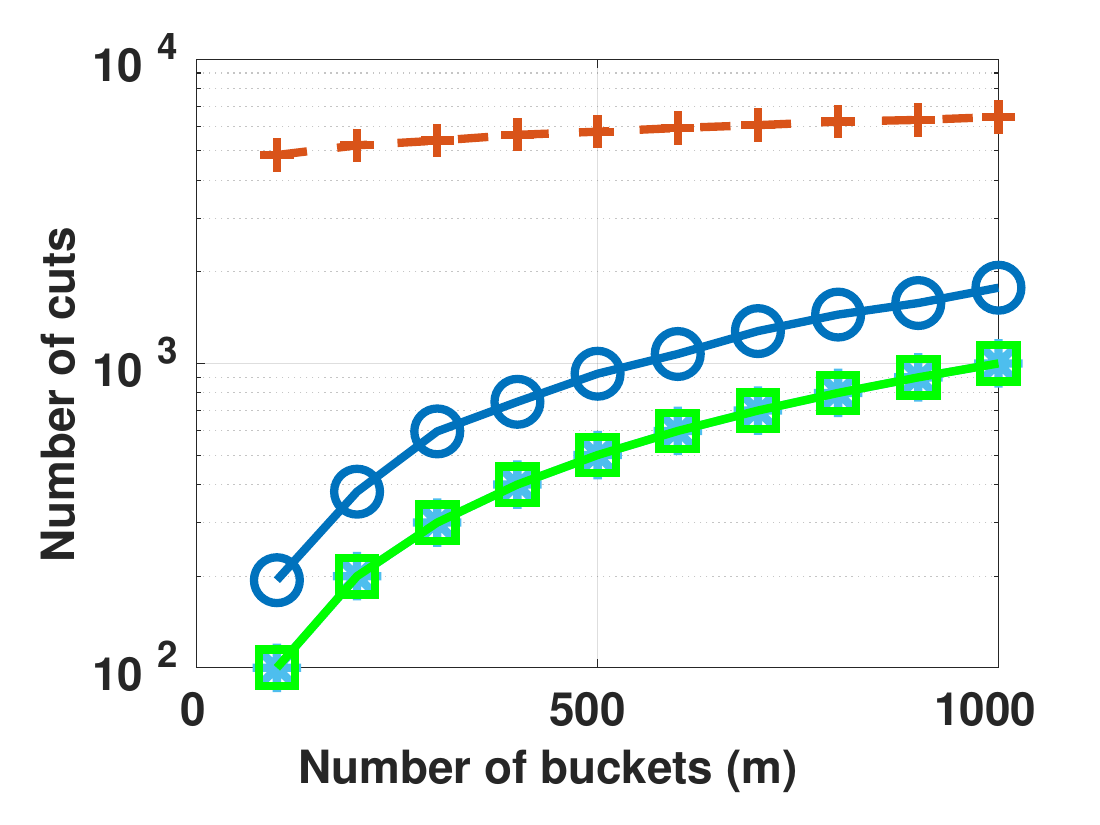}
        \caption{Effect of varying number of buckets $m$ on space, \adult}
        
        \label{fig:adult_m_vs_space}
    \end{minipage}
    \hfill
    \begin{minipage}[t]{0.24\linewidth}
        \centering
        \includegraphics[width=\textwidth]{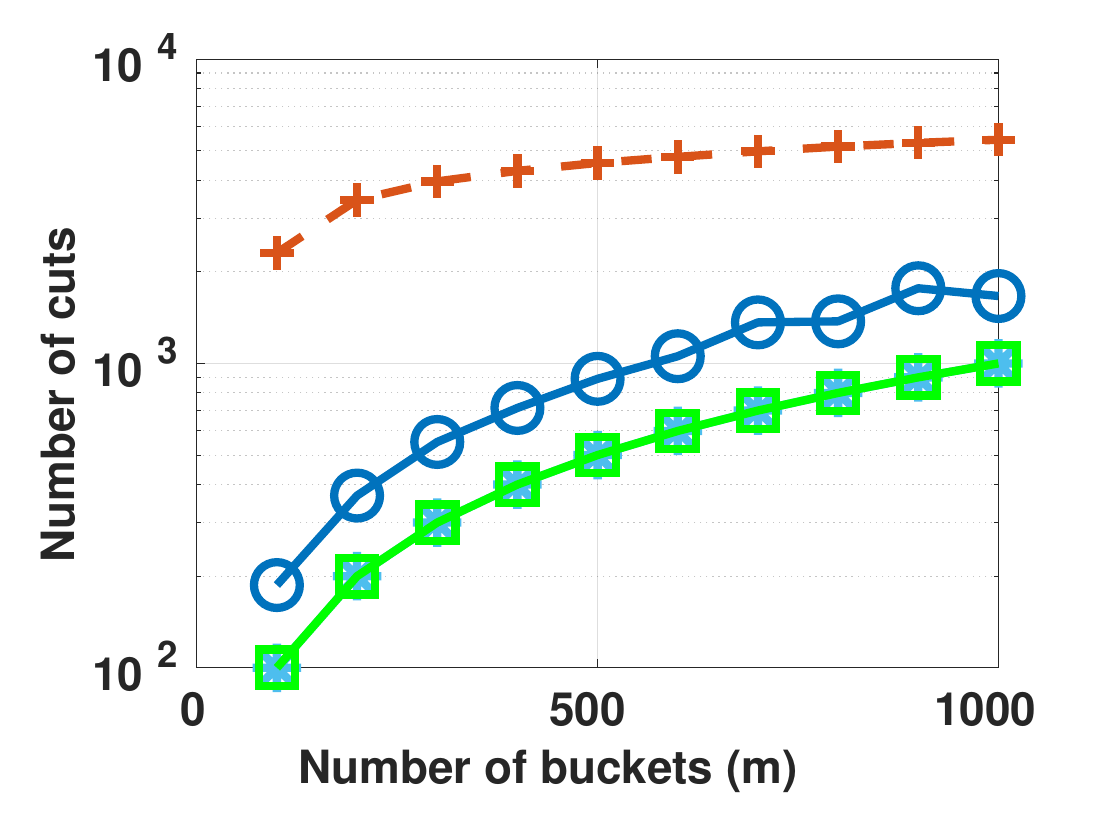}
        \caption{Effect of varying number of buckets $m$ on space, \compas} 
        
        \label{fig:compas_m_vs_space}
    \end{minipage}
\end{figure*}

\begin{figure*}
    \centering
    \includegraphics[width=0.9\textwidth]{plots/legend.png}
    
\hfill
    \begin{minipage}[t]{0.24\linewidth}
        \centering
        \includegraphics[width=\textwidth]{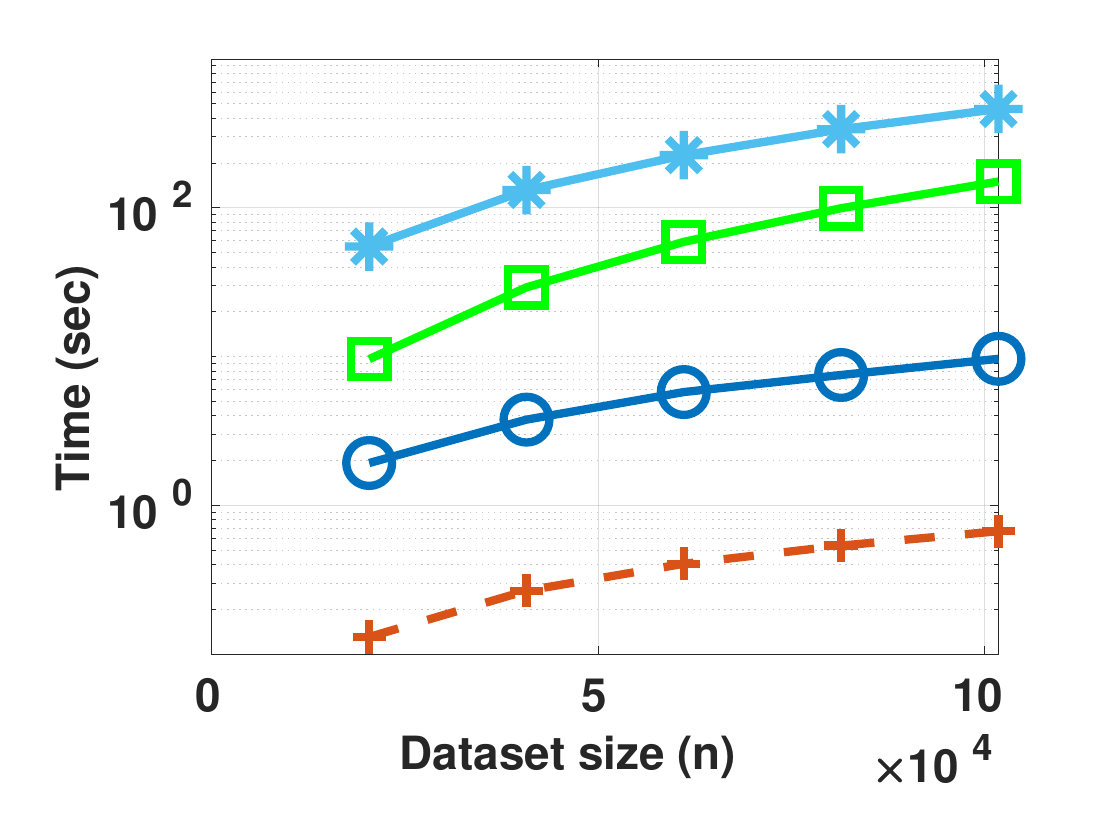}
        \caption{{Effect of varying dataset size $n$ on preprocessing time, \diabetes}}
        
        \label{fig:diabetes_n_vs_prep_time}
    \end{minipage}
    \hfill
    \begin{minipage}[t]{0.24\linewidth}
        \centering
        \includegraphics[width=\textwidth]{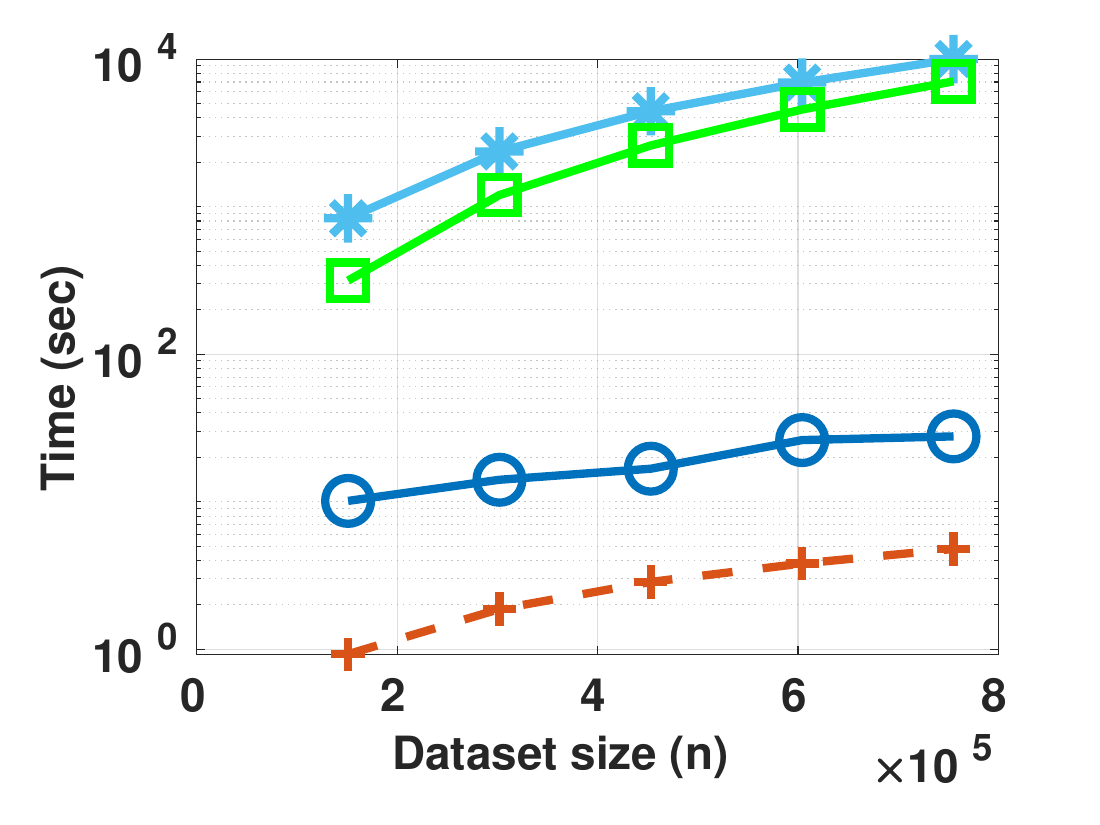}
        \caption{{Effect of varying dataset size $n$ on preprocessing time, \popsim}} 
        
        \label{fig:popsim_n_vs_prep_time}
    \end{minipage}
    \hfill
    \begin{minipage}[t]{0.24\linewidth}
        \centering
        \includegraphics[width=\textwidth]{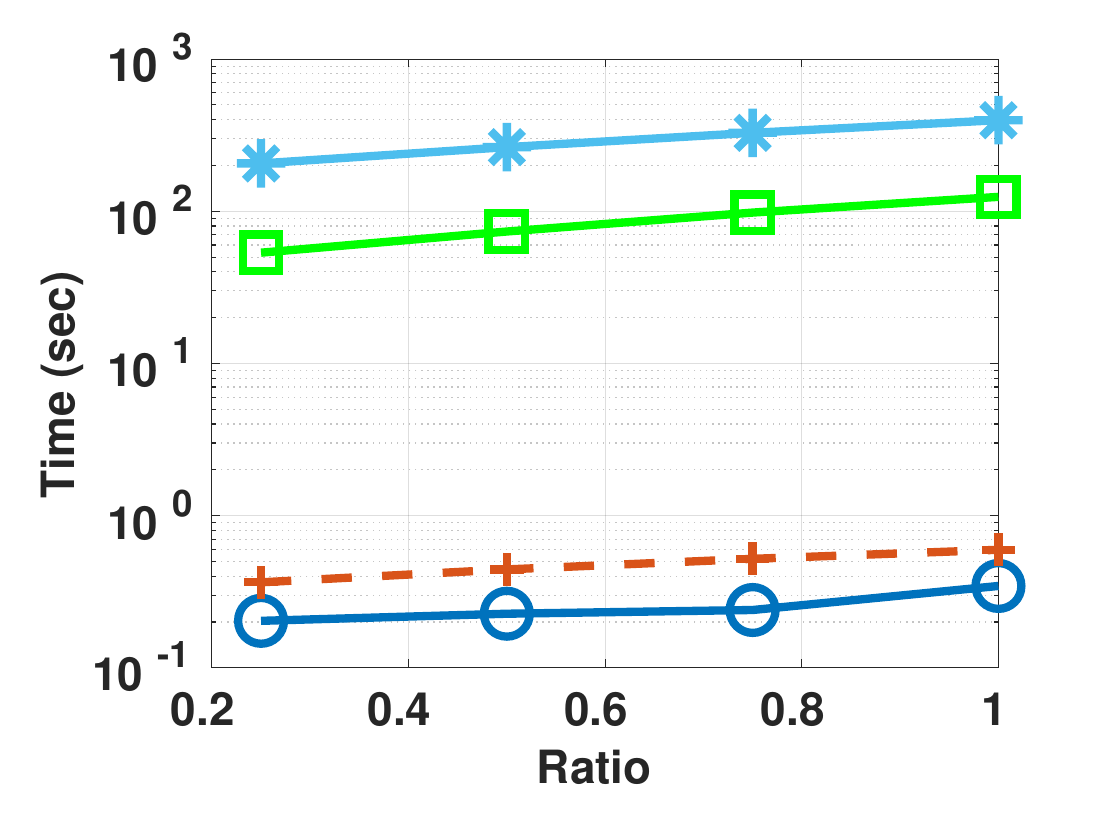}
        \caption{Effect of varying minority-to-majority ratio on preprocessing time, \diabetes}
        
        \label{fig:diabetes_ratio_vs_prep_time}
    \end{minipage}
    \hfill
    \begin{minipage}[t]{0.24\linewidth}
        \centering
        \includegraphics[width=\textwidth]{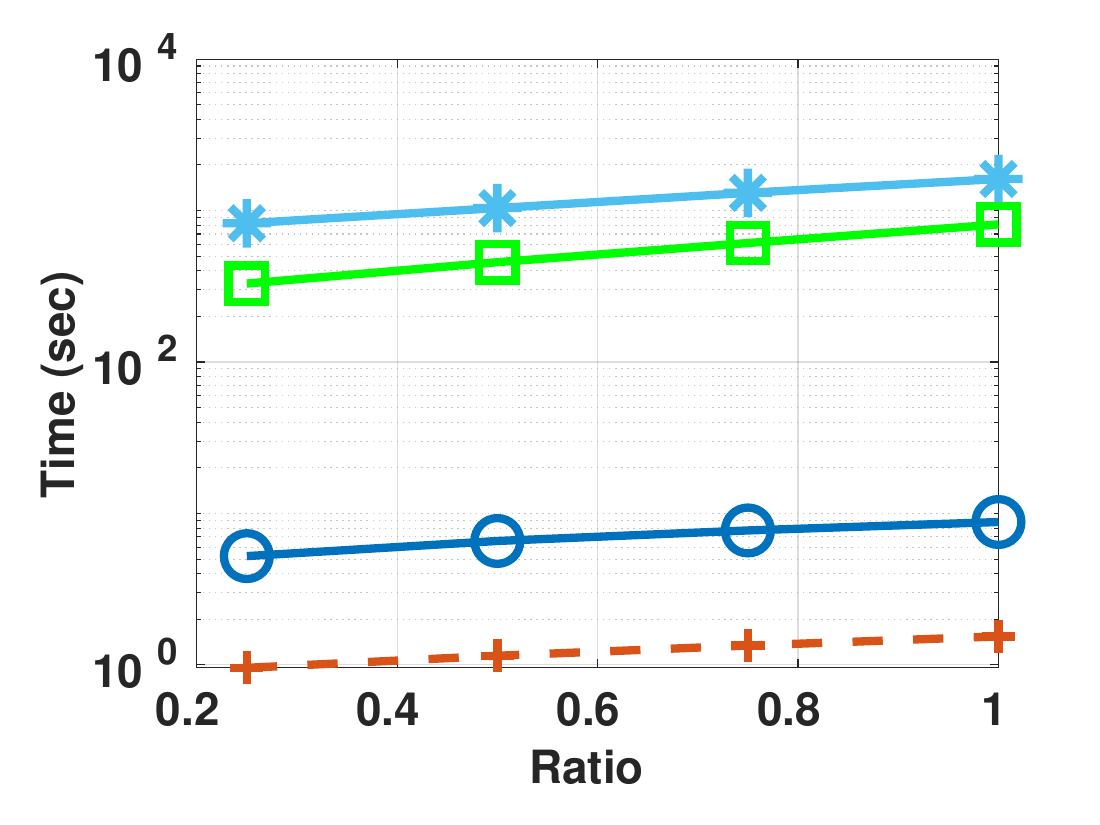}
        \caption[]{Effect of minority-to-majority ratio on preprocessing time, \popsim}
        
        \label{fig:popsim_ratio_vs_prep_time}
    \end{minipage}
    \hfill
    \begin{minipage}[t]{0.24\linewidth}
        \centering
        \includegraphics[width=\textwidth]{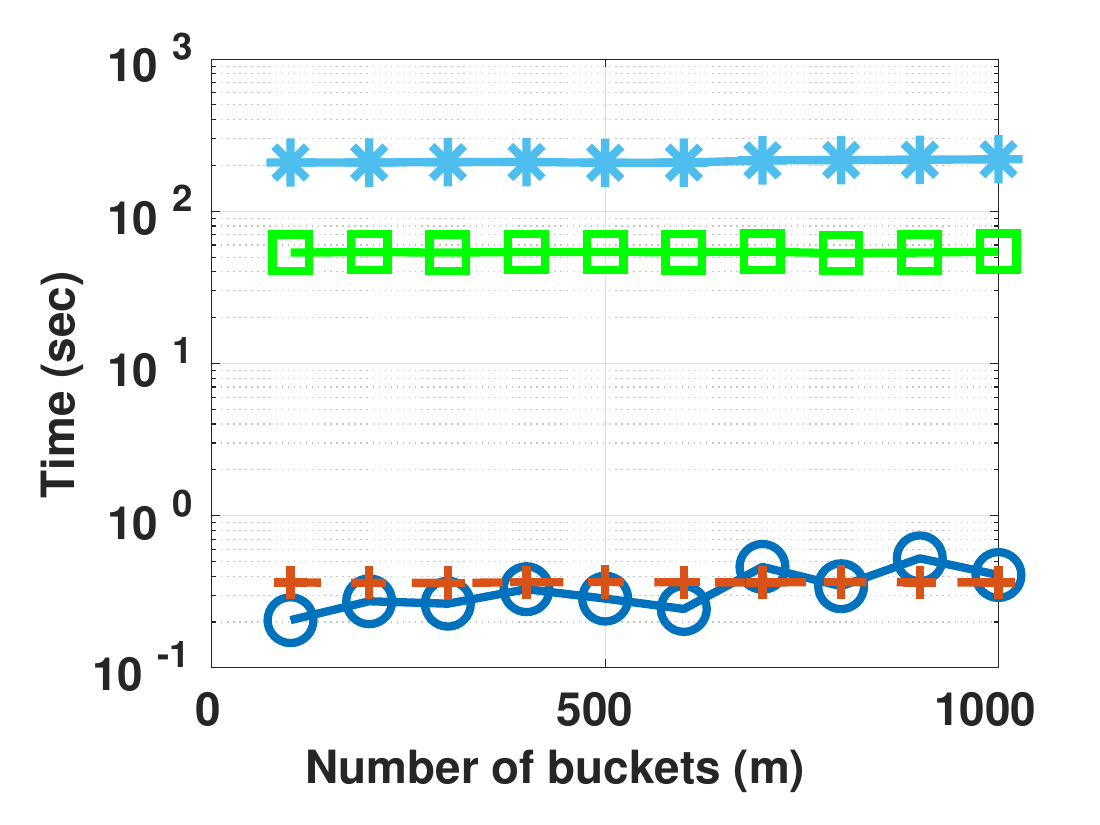}
        \caption{Effect of varying number of buckets $m$ on preprocessing time, \diabetes}
        
        \label{fig:diabetes_m_vs_prep_time}
    \end{minipage}
    \hfill
    \begin{minipage}[t]{0.24\linewidth}
        \centering
        \includegraphics[width=\textwidth]{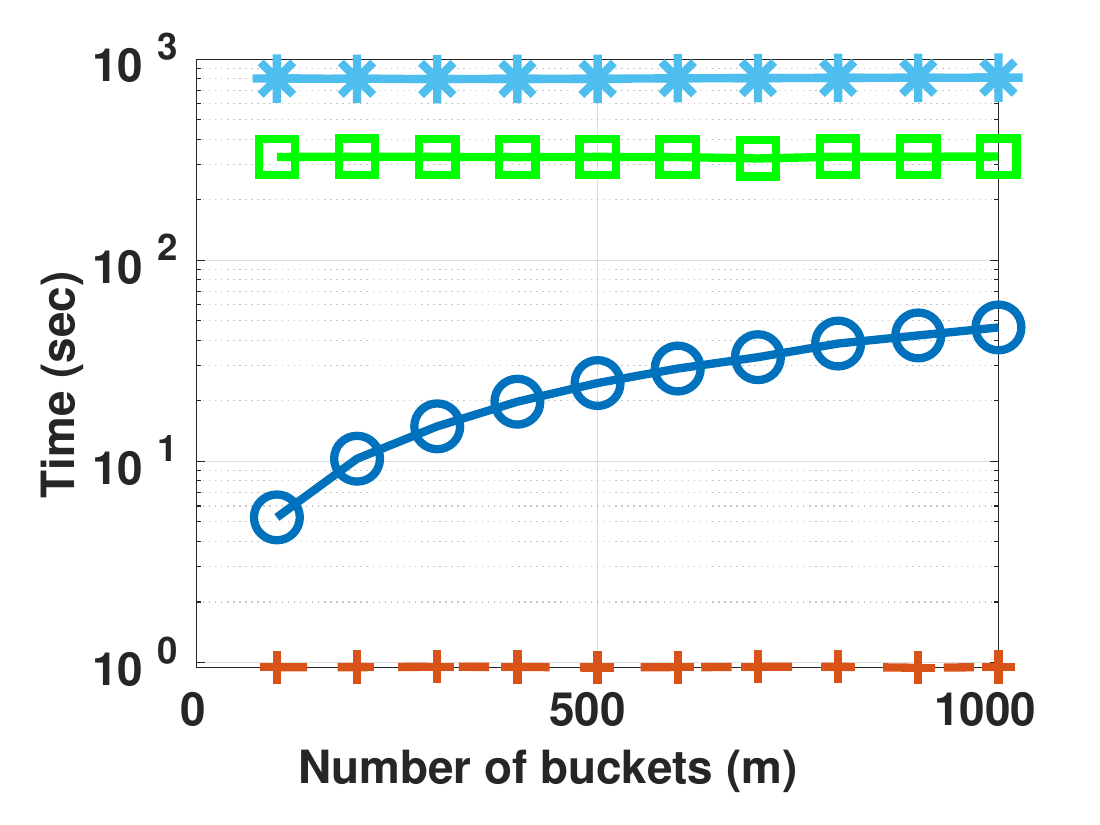}
        \caption{{Effect of varying number of buckets $m$ on preprocessing time, \popsim}} 
        
        \label{fig:popsim_m_vs_prep_time}
    \end{minipage}
    \hfill
    \begin{minipage}[t]{0.24\linewidth}
        \centering
        \includegraphics[width=\textwidth]{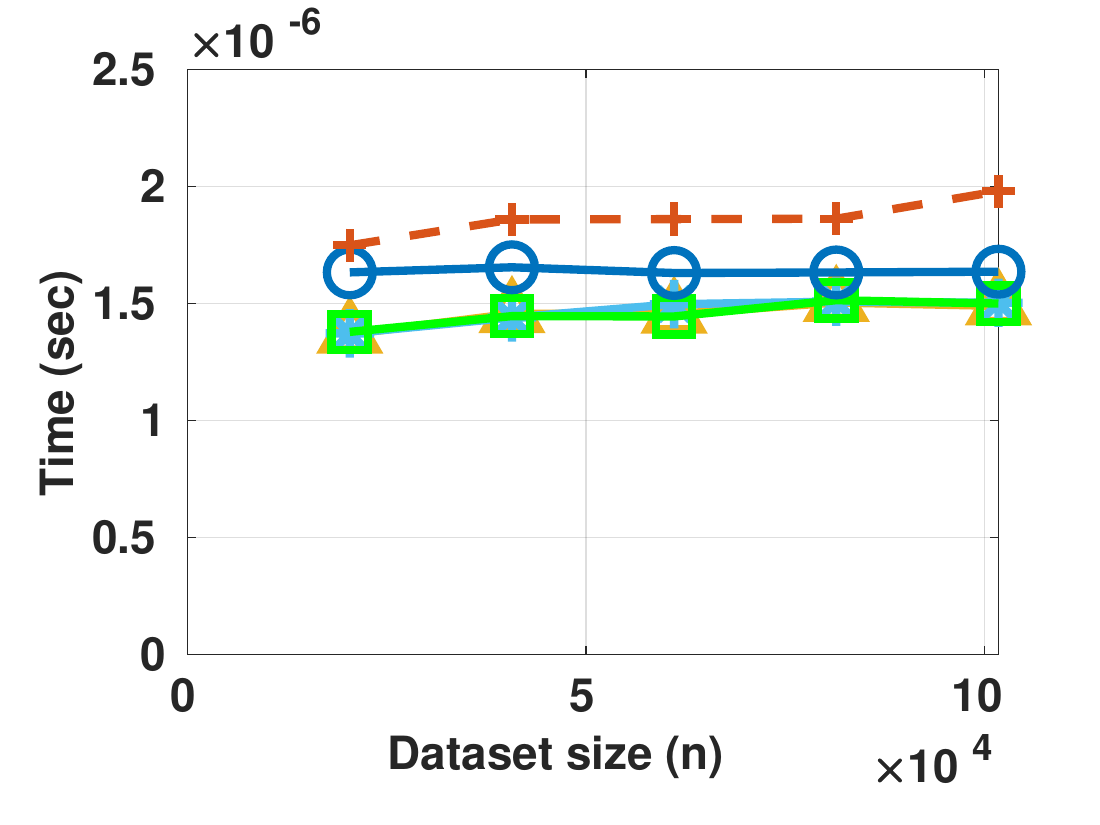}
        \caption{{Effect of varying dataset size $n$ on query time, \diabetes}}
        
        \label{fig:diabetes_n_vs_query_time}
    \end{minipage}
    \hfill
    \begin{minipage}[t]{0.24\linewidth}
        \centering
        \includegraphics[width=\textwidth]{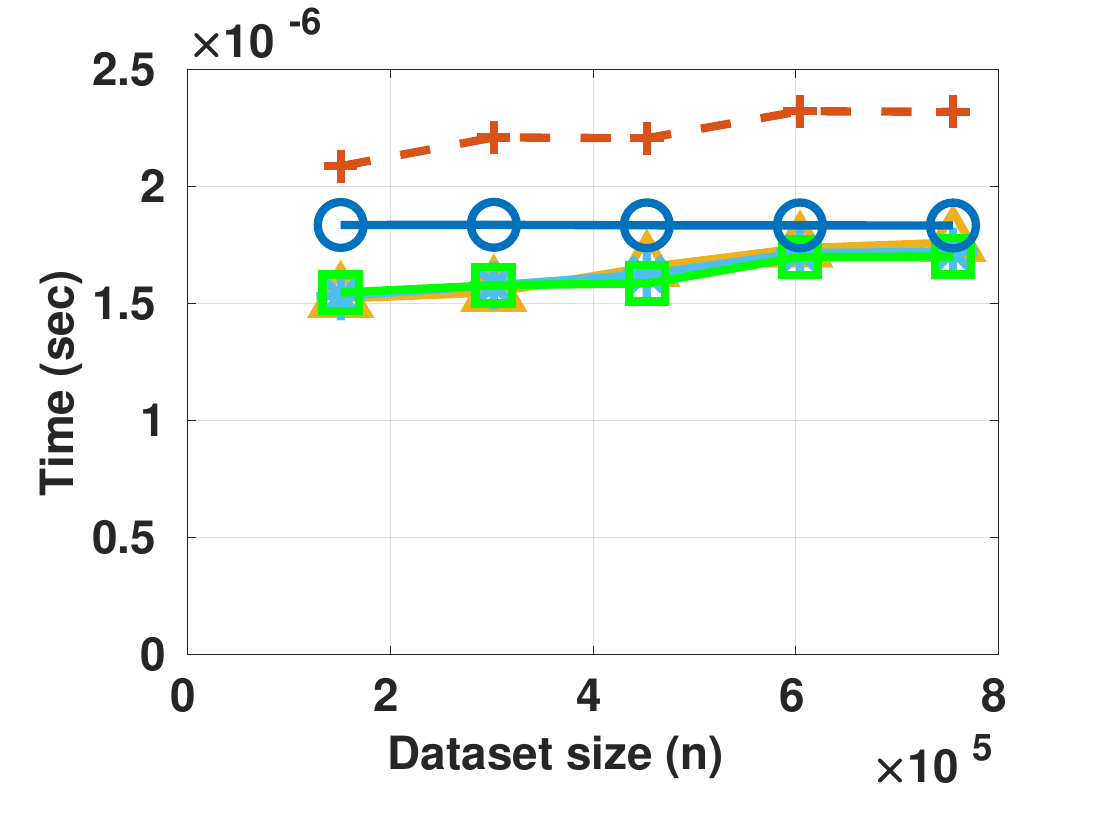}
        \caption[]{{Effect of varying dataset size $n$ on query time, \popsim}}
        
        \label{fig:popsim_n_vs_query_time}
    \end{minipage}
\hfill
    \begin{minipage}[t]{0.24\linewidth}
        \centering
        \includegraphics[width=\textwidth]{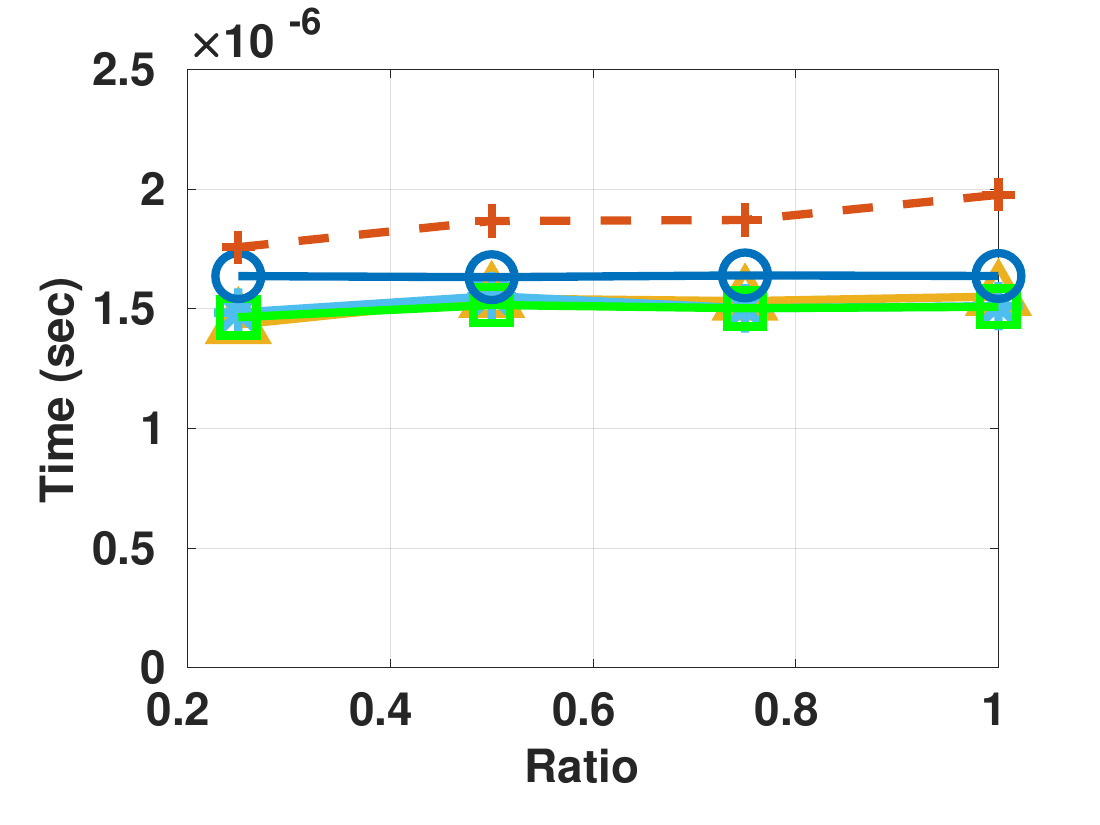}
        \caption{{Effect of varying minority-to-majority on query time, \diabetes}}
        
        \label{fig:diabetes_ratio_vs_query_time}
    \end{minipage}
    \hfill
    \begin{minipage}[t]{0.24\linewidth}
        \centering
        \includegraphics[width=\textwidth]{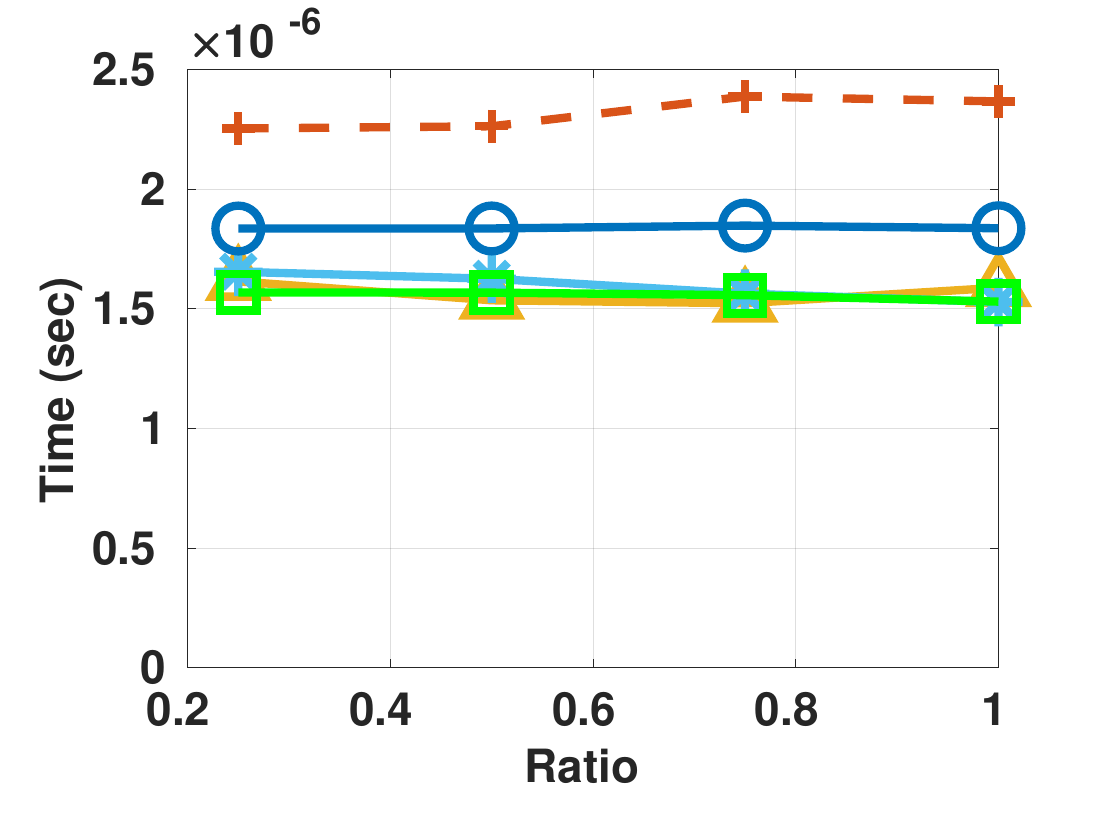}
        \caption{{Effect of varying minority-to-majority on query time, \popsim}} 
        
        \label{fig:popsim_ratio_vs_query_time}
    \end{minipage}
    \hfill
    \begin{minipage}[t]{0.24\linewidth}
        \centering
        \includegraphics[width=\textwidth]{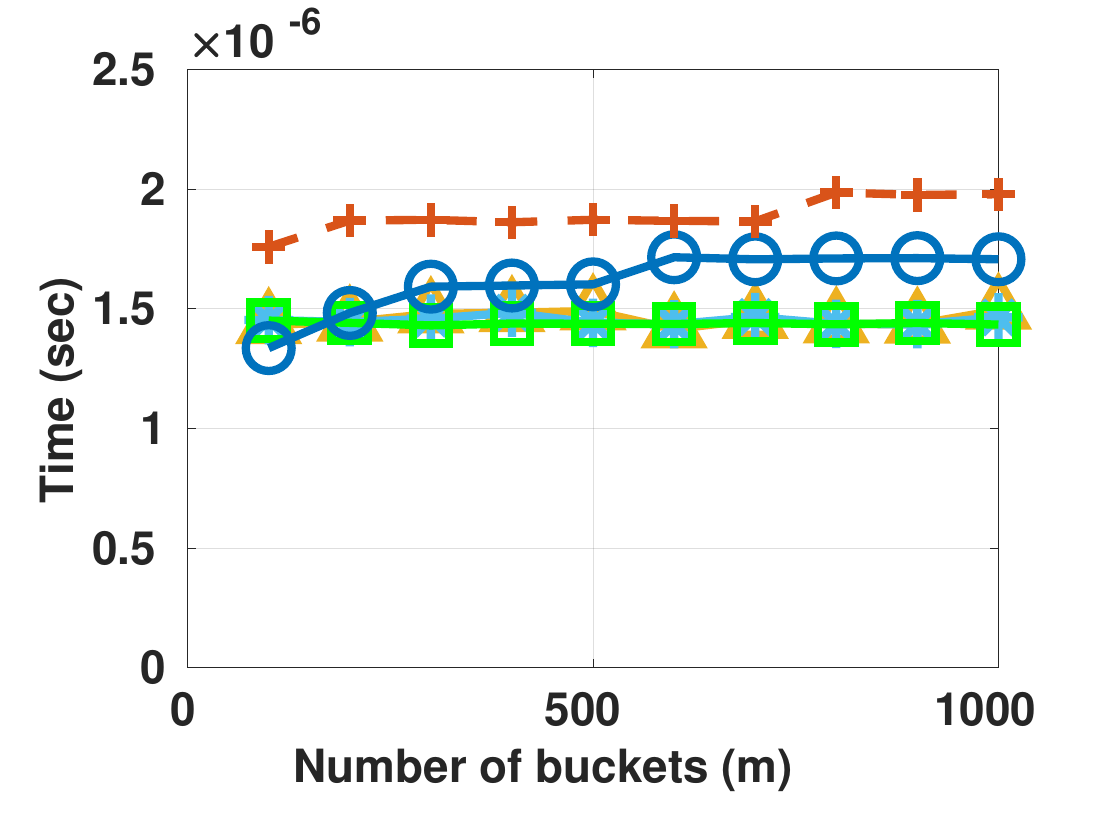}
        \caption{{Effect of varying number of buckets $m$ on query time, \diabetes}}
        
        \label{fig:diabetes_m_vs_query_time}
    \end{minipage}
    \hfill
    \begin{minipage}[t]{0.24\linewidth}
        \centering
        \includegraphics[width=\textwidth]{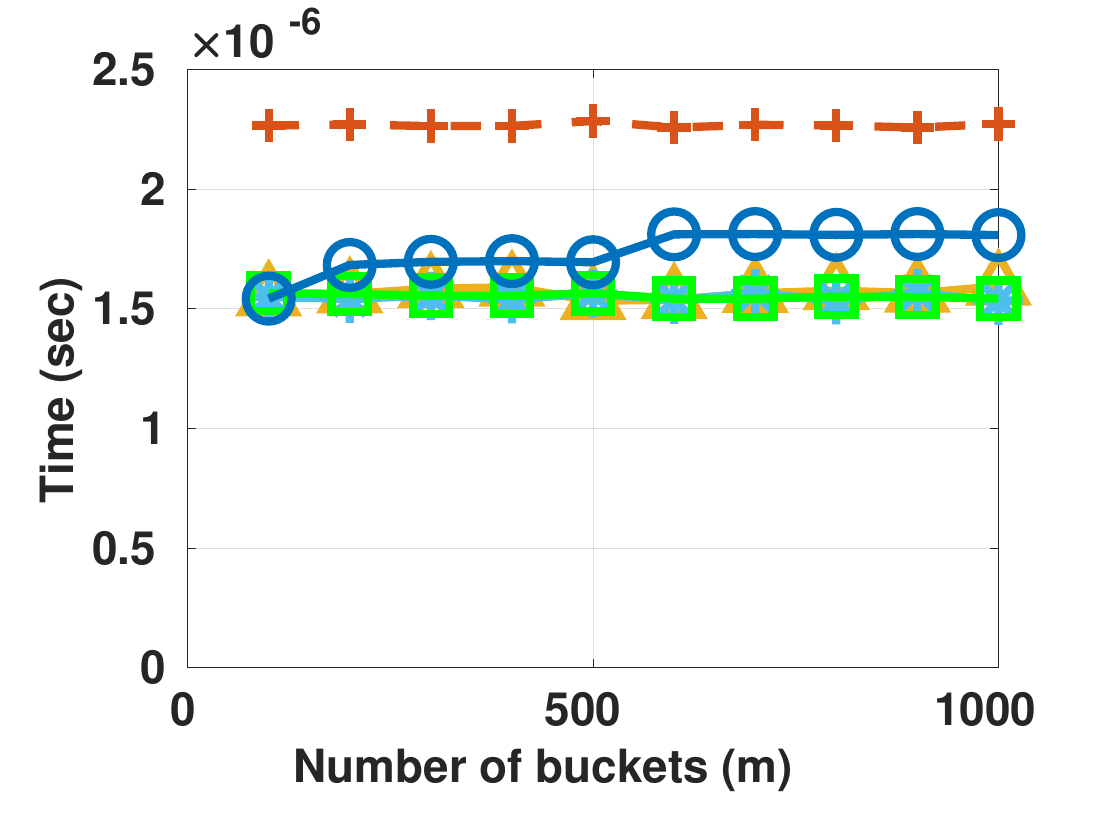}
        \caption{{Effect of varying number of buckets $m$ on query time, \popsim}} 
        \label{fig:popsim_m_vs_query_time}
    \end{minipage}
\end{figure*}

\begin{figure*}
    \begin{minipage}[t]{0.24\linewidth}
        \centering
        \includegraphics[width=\textwidth]{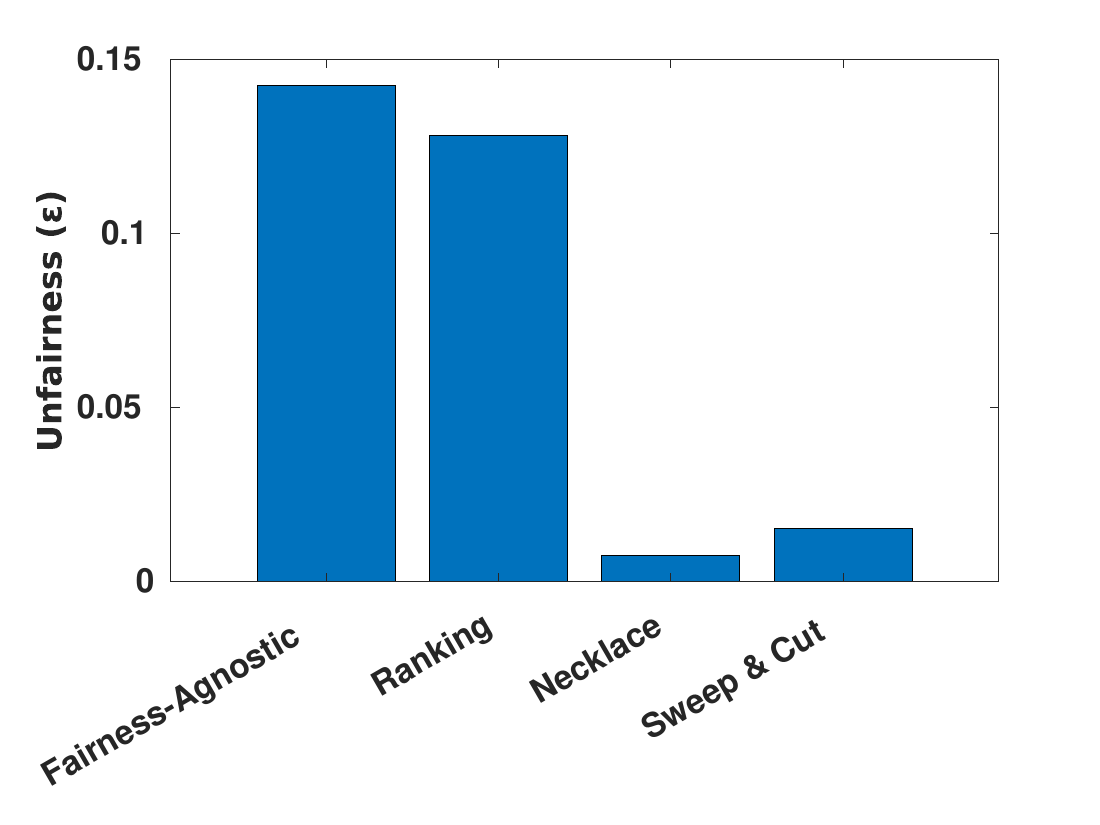}
        \caption{Learning setting: Unfairness evaluation over held out data, \adult} 
        
        \label{fig:learned_adult}
    \end{minipage}
    \hfill
    \begin{minipage}[t]{0.24\linewidth}
        \centering
        \includegraphics[width=\textwidth]{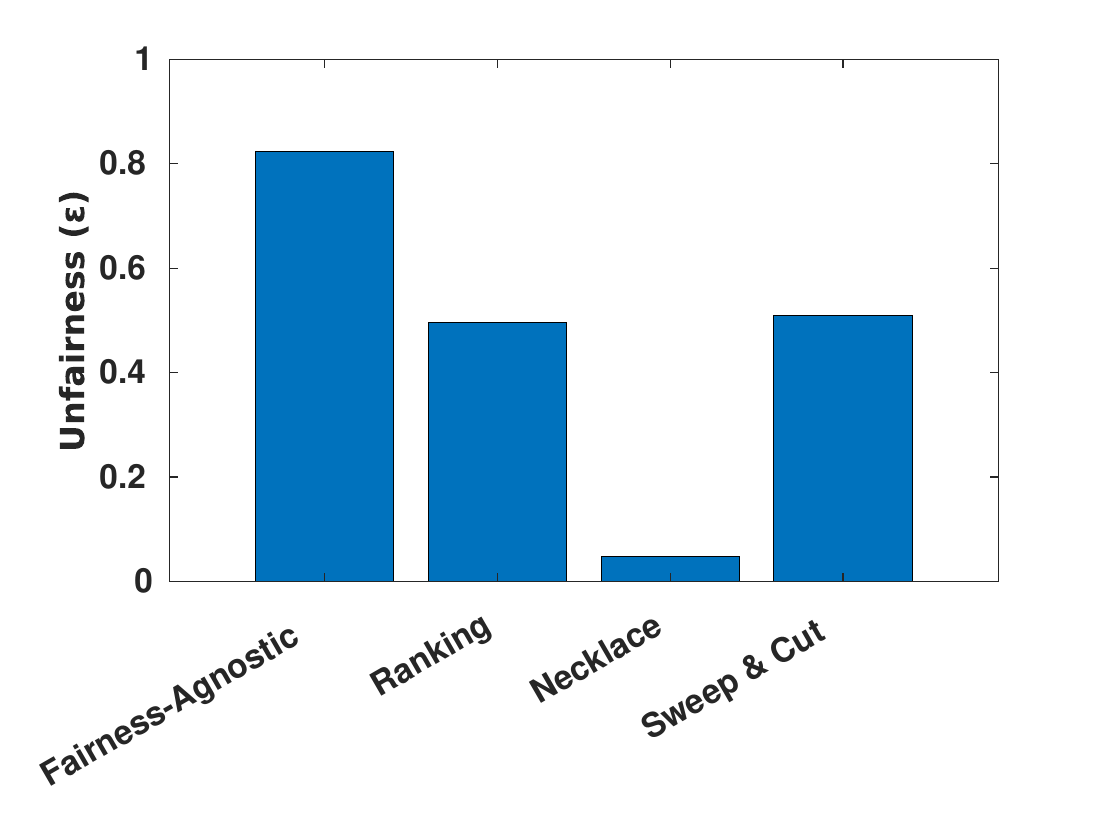}
        \caption{Learning setting: Unfairness evaluation over held out data, \popsim}
        
        \label{fig:learned_popsim}
    \end{minipage}
    \hfill
    \begin{minipage}[t]{0.24\linewidth}
        \centering
        \includegraphics[width=\textwidth]{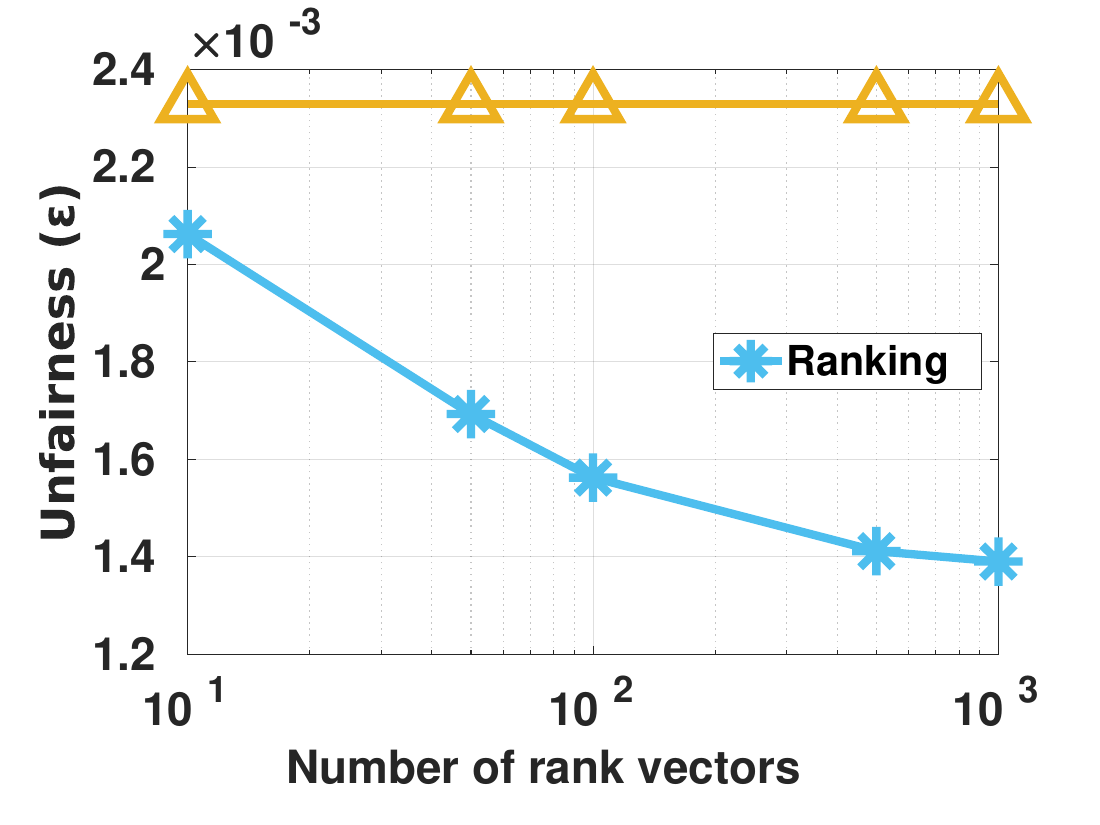}
        \caption{{Effect of varying number of sampled vectors on unfairness, \adult}}
        
        \label{fig:adult_vector_vs_unfairness}
    \end{minipage}
    \hfill
    \begin{minipage}[t]{0.24\linewidth}
        \centering
        \includegraphics[width=\textwidth]{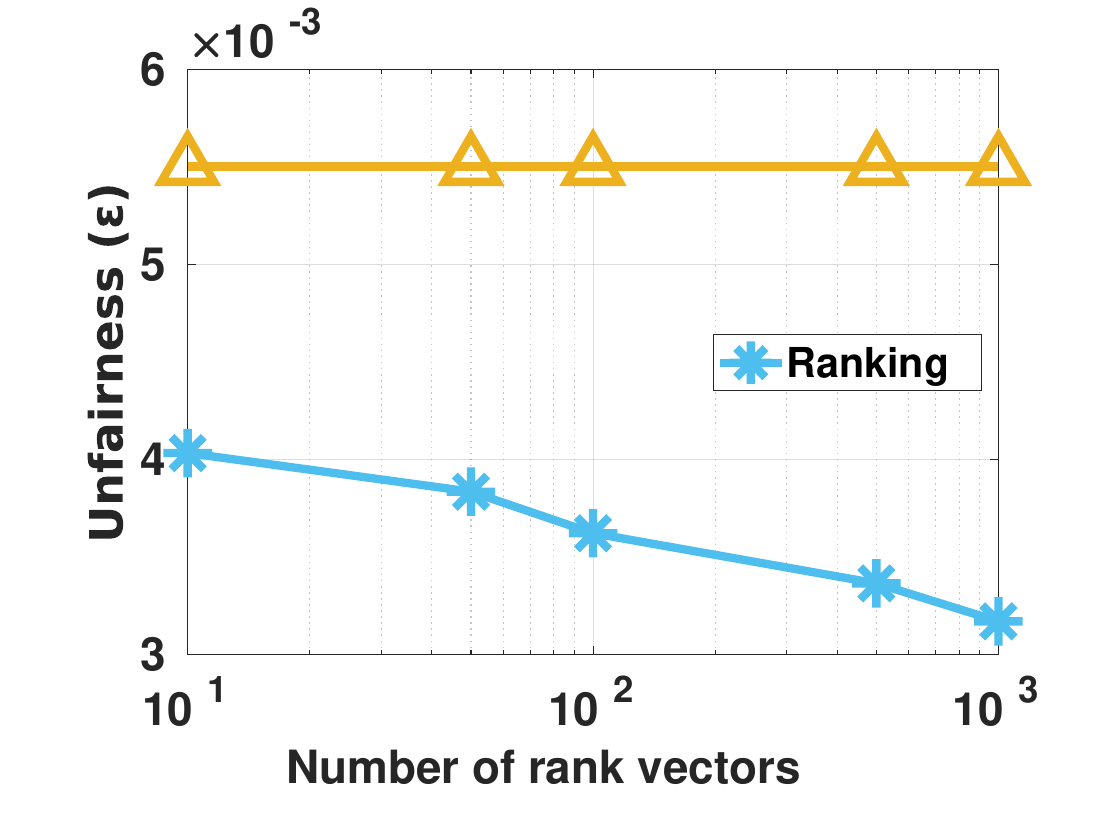}
        \caption{{Effect of varying number of sampled vectors on unfairness, \compas, {\tt sex}}}
        
        \label{fig:compas_vector_vs_unfairness}
    \end{minipage}
    \end{figure*}
    
    \begin{figure}
    \begin{minipage}[t]{0.49\linewidth}
        \centering
        \includegraphics[width=\textwidth]{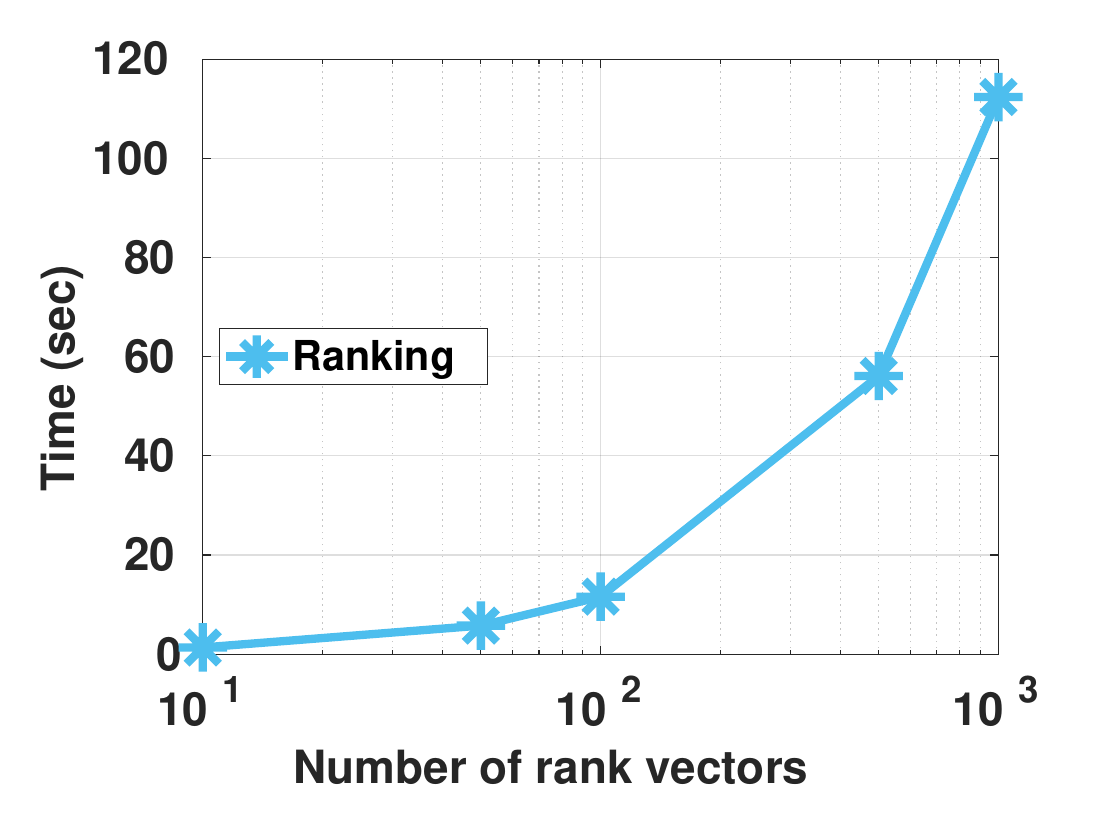}
        \caption{{Effect of varying number of sampled vectors on preprocessing time, \adult}} 
        
        \label{fig:adult_vector_vs_prep_time}
    \end{minipage}
    \hfill
    \begin{minipage}[t]{0.49\linewidth}
        \centering
        \includegraphics[width=\textwidth]{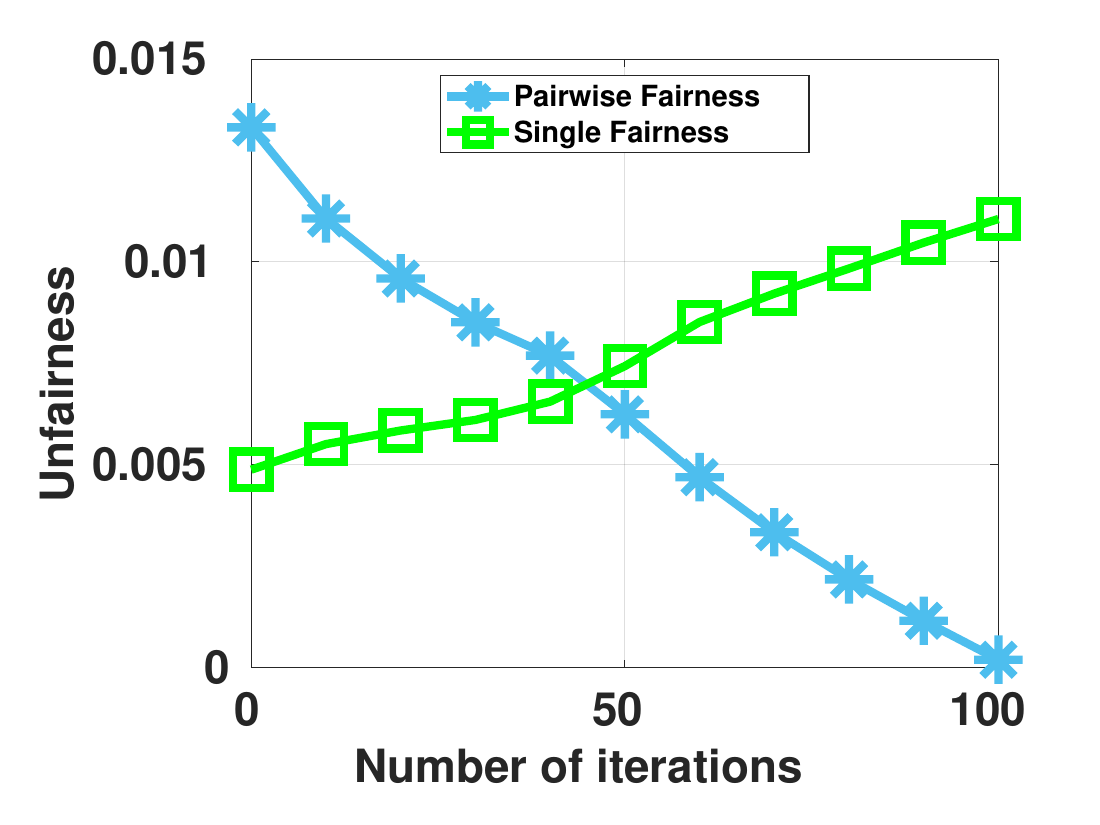}
        \caption{{Pairwise unfairness reduction using the local-search heuristic, \adult} }
        
        \label{fig:adult_local_search}
    \end{minipage}
\end{figure}

\begin{table}[!tb]
    \small
    \centering
    \caption{Overview of datasets}
    \label{tbl:datasets}
    \vspace{-3mm}
    {
    \begin{tabular}{c|c|c|c}
        {\bf Dataset} &{\bf Size}& {\bf Sensitive Attributes}&{\bf No. of Attributes}\\
        \hline
        {\bf\adult} ~\cite{misc_adult_2}&$\sim$49K&{\tt sex}&15
        \\\hline
        {\bf\compas} ~\cite{COMPAS}&$\sim$61K&${\tt sex, race}$&29
        \\\hline
        {\bf\diabetes} ~\cite{strack2014impact}&$\sim$102K&${\tt sex}$&49
        \\\hline
        {\bf\popsim} ~\cite{khanh2023popsim} &1M&{\tt race}&5
    \end{tabular}}
\end{table}

\subsection{Unfairness Evaluations}
We start our experiments by evaluating our algorithms for unfairness. Recall that \ranker returns an $\optRank$-unfair hashmap while \pd and \necklaceb output 0-unfair hashmaps, i.e., $\eps=0$. 
In the first experiment, we study the effect of varying dataset size $n$ on the unfairness. As shown in Figures~\ref{fig:adult_n_vs_unfairness} and~\ref{fig:compas_n_vs_unfairness}, irrespective of the dataset size, \pd and \necklaceb always exhibit zero unfairness. On the other hand, \ranker, while improving compared to \fag, still shows a small degree of unfairness that, similar to the baseline, decreases as the size of the dataset grows.
We also studied the impact of increasing the number of sampled vectors in Figures~\ref{fig:adult_vector_vs_unfairness} and~\ref{fig:compas_vector_vs_unfairness} and noticed a consistent decrease in the unfairness with the increasing number of sampled vectors. 
\color{black} 
Next, we study the effect of the minority-to-majority ratio on unfairness. The results are illustrated in Figures~\ref{fig:adult_ratio_vs_unfairness} and~\ref{fig:compas_ratio_vs_unfairness} with \pd and \necklaceb showing no unfairness, while \ranker reducing the unfairness compared to \fag. 
It is worth mentioning that all unfairness values approach to zero when 
the dataset includes an equal number of records from each group. This further accentuates the role of unequal base rate~\cite{kleinberg2016inherent} in unfairness.
Last but not least, we evaluate the effect of increasing the number of buckets $m$ on unfairness (Figures~\ref{fig:adult_m_vs_unfairness} and~\ref{fig:compas_m_vs_unfairness}). \pd and \necklaceb are independent of the number of buckets and show zero unfairness as $m$ increases. \fag and \ranker unfairness values however increase in a linear fashion as $m$ grows. Consistent with the two previous experiments, we observe that \ranker methods moderately improve the unfairness. 
Overall, 
confirming our theoretical analysis, \pd and \necklaceb are preferred from the fairness perspective.

\subsection{Space Evaluations}
Next, we evaluate our algorithms for memory demands a.k.a. space. Recall that \ranker is a $1$-memory hashmaps, meaning that no additional memory is required and the number of boundary points is exactly $m-1$. \necklaceb guarantees the number of cuts to be at most $2(m-1)$ while \pd can create up to $O(n)$ cuts in the worst-case. 
We investigated the effect of varying dataset size $n$ (Figures~\ref{fig:adult_n_vs_space} and~\ref{fig:compas_n_vs_space}), minority-to-majority ratio (Figures~\ref{fig:adult_ratio_vs_space} and~\ref{fig:compas_ratio_vs_space}), and number of buckets (Figures~\ref{fig:adult_m_vs_space} and~\ref{fig:compas_m_vs_space}) on the required space for each algorithm. 
As expected the results are consistent with the theoretical bounds. The number of boundaries created by \ranker is independent of the dataset size $n$ and minority-to-majority ratio and only depends on the number of buckets $m$ and therefore it is always a constant ($m-1$). Our experiments also verify similar results for \necklaceb being independent of $n$ and minority-to-majority ratio. Interestingly, in almost all settings, the actual number of cuts created by \necklaceb is close to the upper-bound $2(m-1)$. The results for \pd however, verify that the major drawback of this algorithm is the memory demands with close to the worst case of $O(n)$ cuts. Lemma~\ref{th:pd:mem:avg} provides an upper-bound on the expected number of cuts as a function of the minority ratio in the data set.
To evaluate how tight this upper-bound is in practice, in Figures~\ref{fig:adult_ratio_vs_space} and~\ref{fig:compas_ratio_vs_space}, we also present the actual number of cuts while varying the minority-to-majority ratios. At least in this experiment, the upper-bound was tight as it was always less than 30\% larger than the actual number. In general, if an application requires maintaining the space at a minimum while satisfying fairness constraints, as empirically observed, \ranker is the leading alternative. \necklaceb is also a favorable choice as it provides a practical trade-off between fairness (0-unfair) and space ($\leq2(m-1)$ cuts).

\subsection{Efficiency Evaluations}
In this set of experiments, we evaluate our proposed algorithms for efficiency. More specifically, we measure efficiency from two perspectives: 1) the preprocessing time that is required to construct the fair hashmap, and 2) the query time needed to return a hash (bucket) for new records when the hashmap is constructed.

\subsubsection{Preprocessing Time}
We start our efficiency experiments by revisiting the preprocessing time complexity of the proposed algorithms. \ranker has a time complexity of $O(n^2\log n)$ in 2D while \pd and \necklaceb both run in $O(n\log n)$. In our first experiment, we study the impact of varying the dataset size $n$. First, in Figures~\ref{fig:diabetes_n_vs_prep_time} and~\ref{fig:popsim_n_vs_prep_time}, one can confirm that, overall, the run-time increases with the dataset size. For \ranker, the exact time depends on the number of rankings generated.
This is also evident in Figures~\ref{fig:adult_vector_vs_prep_time}, where we study the impact of increasing the number of sampled vectors on the preprocessing time.
Next, as demonstrated in Figures~\ref{fig:diabetes_ratio_vs_prep_time} and ~\ref{fig:popsim_ratio_vs_prep_time}, we confirm that the preprocessing time is independent of the minority-to-majority ratio.
As with the preceding experiment, we expect that varying $m$ should not impact the run-time of any of the algorithms, which is consistent with our experiment results in Figures~\ref{fig:diabetes_m_vs_prep_time} and ~\ref{fig:popsim_m_vs_prep_time}.
Overall, in time-sensitive applications, both \pd and \necklaceb offer the fastest results, all the while ensuring 0-unfairness.
\subsubsection{Query Time}
Recall that the output of our algorithms consists of a sequence of boundaries along with a corresponding set of hash buckets. After constructing the hashmap, obtaining the hash bucket for a new query is a simple process: just execute a binary search on the boundaries and retrieve the bucket linked to the boundaries within which the query point resides. Therefore, query time is in $O(\log|B|)$ and only depends on the number of boundaries. Our empirical results are consistent with the preceding analysis, confirming that query time remains independent of both dataset size (see Figures~\ref{fig:diabetes_n_vs_query_time} and~\ref{fig:popsim_n_vs_query_time}) and the minority-to-majority ratio (refer to Figures~\ref{fig:diabetes_ratio_vs_query_time} and~\ref{fig:popsim_ratio_vs_query_time}). The increase in the number of buckets ($m$) leads to a logarithmic growth in query time across all our algorithms, as depicted in Figures~\ref{fig:diabetes_m_vs_query_time} and~\ref{fig:popsim_m_vs_query_time}. Thanks to the logarithmic efficiency of binary search, all our algorithms offer remarkably fast query times, with with practically negligible variations.
\subsection{Local-search-based Heuristic Evaluation}
In this experiment, we apply the local-search-based heuristic to the boundaries generated by the \ranker algorithm on an instance of \adult dataset. We run the heuristic in 1000 iterations, with single fairness lower and upper bounds as $f^-=0, f^+=0.05$ respectively and the collision probability upper bound as $c^+=0.05$. The results are shown in Figure \ref{fig:adult_local_search}. The local-search-based heuristic effectively enhances pairwise fairness by making minimal adjustments to the bin boundaries, incurring a slight cost in single fairness within 100 iterations before it concludes.
\subsection{Learning Settings Evaluation}\label{sec:exp:learning}
So far in our experiments, we assumed that the algorithms have access to the entire input set. In this experiment, we demonstrate that our methods can work in expectation if an unbiased sample set from the input set is provided. To do so, we partition the input datasets into training and test sets with a ratio of 0.8 to 0.2. Next, we utilize our algorithms to create a hashmap on the training set.
We then use the test set for evaluation: each test entry is queried on the constructed hashmap to identify their buckets. Having created a hashmap that exclusively contains the test entries, we proceed to measure the pairwise unfairness. The results are illustrated in Figures~\ref{fig:learned_adult} and~\ref{fig:learned_popsim}. Although all methods demonstrate an enhancement in unfairness compared to the \fag baseline, the most notable improvement is observed with \necklaceb, where the unfairness decreases from 0.81 to 0.03 for the \popsim and from 0.15 to 0.007 for \adult. Aligned with our previous results, \ranker only moderately improves the unfairness. Although \pd consistently improves the unfairness in the learned settings, depending on the number of cuts it generates, it may exhibit signs of overfitting based on the number of cuts it generates. This tendency becomes more apparent, especially when dealing with large training data, as illustrated in Figure~\ref{fig:learned_popsim}. However, this overfitting phenomenon is mitigated when the size of training data is smaller and the number of cuts is reduced, resulting in a substantial decrease in unfairness, as depicted in Figure~\ref{fig:learned_adult}.
\section{Related work}\label{sec:related}
\stitle{Hashing} Hashing has a long history in computer science~\cite{wang2014hashing, chi2017hashing}. Hashing-based algorithms and data structures find many applications in various areas such as theory, machine learning, computer graphics, computational geometry and databases~\cite{knuth1997art, marcus2020benchmarking, charikar2019multi, ahle2020subsets, chen2022truly, czumaj2022streaming, duan2022faster}. Due to its numerous applications, the design of efficient hash functions with theoretical guarantees are of significant importance~\cite{damgaard1989design, dhar2022linear, kuszmaul2022hash, aamand2019non, aamand2020fast, xing2023beating, bender2022optimal, assadi2023tight, bercea2022extendable, kacham2023pseudorandom}.
In traditional hashmaps, the goal is to design a hash function that maps a key to a random value in a specified output range. The goal is to minimize the number of collisions, where a collision occurs when multiple keys get mapped to the same output value. There are several well-known schemes such as chaining, probing, and cuckoo hashing to handle collisions. Recently, machine learning is used to learn a proper hash function~\cite{sabek2022can, kraska2018case, mitzenmacher2018model}. In a typical scenario, a set of samples is received and they learn the CDF of the underlying data distribution. Then the hashmap is created by partitioning the range into equal-sized buckets.
It has been shown that such learned index structures~\cite{sabek2022can, kraska2018case}, can outperform traditional hashmaps on practical workloads.
However, to the best of our knowledge, none of these hashmap schemes can handle fair hashing with theoretical guarantees.
Finally, learning has been used to obtain other data structures as well, such as $B$-trees~\cite{kraska2018case} or bloom filters~\cite{kraska2018case, vaidya2020partitioned, mitzenmacher2018model}.

\stitle{Algorithmic Fairness}
Fairness in data-driven systems has been studied by various research communities but mostly in machine learning (ML)~\cite{barocas2017fairness,mehrabi2021survey,pessach2022review}.
Most of the existing work is on training a ML model that satisfies some fairness constraints. Some pioneering fair-ML efforts include~\cite{dwork2012fairness,zafar2017fairness,celis2019classification,calmon2017optimized,feldman2015certifying,hardt2016equality,kamiran2012data}.
Biases in data has also been studied extensively~\cite{stoyanovich2022responsible,shahbazi2023representation,olteanu2019social,salimi2019data,asudeh2019assessing,asudeh2021identifying} to ensure data has been prepared responsibly~\cite{salimi2019interventional,nargesian2021tailoring,salimi2020database,shetiya2022fairness}. 
Recent studies of fair algorithm design include
fair
clustering~\cite{chierichetti2017fair,bera2019fair,schmidt2020fair,ahmadian2020fair,bohm2020fair,makarychev2021approximation,chlamtavc2022approximating,hotegni2023approximation},
fairness in resource allocation and facility location problem~\cite{blanco2023fairness,jiang2021rawlsian,donahue2020fairness,he2020inherent,zhou2019public,gorantla2023fair},
min cut~\cite{li2023near},
max cover~\cite{asudeh2023maximizing},
game theoretic approaches~\cite{zhao2021fairness,donahue2023fairness,algaba2019shapley},
hiring~\cite{raghavan2020mitigating,aminian2023fair},
ranking~\cite{asudeh2019designing,singh2018fairness,singh2019policy,zehlike2022fairness,guan2019mithraranking,asudeh2018obtaining},
recommendation~\cite{chen2023bias,li2022fairness, swift2022maximizing},
representation learning~\cite{he2020geometric},
etc.

Fairness in data structures is significantly under studied, with the existing work being limited to ~\cite{aumuller2021fair,aumuller2022sampling,aumuller2020fair}, which study {\em individual fairness} in near-neighbor search. Particularly, a rejection sampling technique has been added on top of local sensitive hashing (LSH) that equalizes the retrieval chance for all points in the $\rho$-vicinity of a query point, independent of how close those are to the query point.
To the best of our knowledge, we are the first to study {\em group fairness} in hashing and more generally in data structure design.


\vspace{1mm}
\section{Final Remarks and Future Work}
In this paper we studied hashmaps through the lens of fairness and proposed several fair and memory/time efficient algorithms.
Some the interesting directions for future work are as following.

{\em Memory-efficient 0-unfair hashmaps for more than two groups}:
    Our ranking-based algorithms do not depend on the number of groups, however, those are not 0-unfair. The \pd algorithm also does not depend on the number of groups and it is even 0-unfair. Its memory requirement, however, can be as high as $O(n)$.
    The performance of the necklace splitting algorithm, on the other hand, depends on the number of groups. While for two groups, the number of boundaries is independent from $n$ and at most $2(m-1)$, the state-of-the-art algorithm for more than two groups suddenly increases this requirement by a factor of $O(\log(n))$. Developing a fair $(0,c)$-hashmap for this case, for a constant value of $c$, remains an interesting open problem for future work.

{\em Beyond the class of Linear Ranking functions}:
    We developed our ranking-based algorithms using the class of linear ranking functions. 
    It would be interesting to expand the scope to more general non-linear classes (such as monotonic functions).
    Indeed one can add non-linear attributes before running our ranking-based algorithms. But this will increase the number of dimensions, exponentially reducing its run time.

{\em Lower-bounds and trade-offs on $(\eps,\alpha)$}: Our ranking-based algorithms satisfy 1-memory requirement but cannot achieve 0-unfairness. The cut-based algorithms, on the other hand, are 0-unfair but require additional memory. This suggests a trade-off between fairness and memory requirements. 
    Last but not least, formally studying this trade-off and identifying the lower-bound Pareto frontier for fairness and memory is an interesting future work.

\vspace{2mm}
\section*{Acknowledgements}
This work was supported in part by 
the National Science Foundation, Grant No. 2107290, and the CAHSI-Google IRP Award.
The authors would like to thank the anonymous reviewers and the meta-reviewer for their invaluable feedback.

\bibliographystyle{ACM-Reference-Format}
\bibliography{ref}

\section*{Appendix}
\appendix
\section{Extended Details of Discrepancy-based hashmaps}
\subsection{Faster randomized algorithm for small $m$}
The dynamic programming algorithm we designed depends on $n^{d+2}$, which is too large. Here we propose a faster randomized algorithm to return an $((1+\delta)\optDisc+\gamma, 1)$-hashmap, with running time which is strictly linear on $n$.
We use the notion of $\gamma$-approximation.
Let $w$ be a vector in $\Re^d$, and let $P_w$ be the projections of all points in $P$ onto $w$. We also use $P_w$ to denote the ordering of the points on $w$.
Let $\mathcal{I}_w$ be the set of all possible intervals on the line supporting $w$. Let $(P_w, \mathcal{I}_{w})$ be the range space with elements $P_w$ and ranges $\mathcal{I}_w$.
Let $A_i$ be a random (multi-)subset of $\gee_i$ of size $O(\gamma^{-2}\log \phi^{-1})$ where at each step a point from $\gee_i$ is chosen uniformly at random (with replacement). It is well known~\cite{har2011geometric, chazelle2000discrepancy} that $A_i$ is a $\gamma$-approximation of $(\gee_i,\mathcal{I}_w)$ with probability at least $1-\phi$, satisfying
$|\frac{|\gee_i\cap I|}{|\gee_i|}-\frac{|A_i\cap I|}{|A_i|}|\leq \gamma$,
for every $I\in \mathcal{I}_w$.

We discretize the range $[0,1]$ creating the discrete values of discrepancy $E=\{0,\frac{1}{n}, (1+\delta)\frac{1}{n},\ldots, 1\}$ and run a binary search on $E$.
Let $\alpha\in E$ be the discrepancy we consider in the current iteration of the binary search. We design a randomized algorithm such that, if $\alpha>\optDisc$ it returns a valid hashmap. Otherwise, it does not return a valid hashmap. The algorithm returns the correct answer with probability at least $1-1/n^{O(1)}$.

\paragraph{Algorithm.}
For each $i$, we get a sample set $A_i\subseteq \gee_i$ of size $O(\frac{m^2}{\gamma^2}\log n^{d+2})$. Let $A=\cup_{i}A_i$. Then we run the binary search we described in the previous paragraph. Let $\alpha$ be the discrepancy we currently check in the binary search.
Let $\mathcal{W}_A$ be the set of all combinatorially different vectors with respect to $A$ as defined in Section~\ref{sec:ranking}. For each $w\in \mathcal{W}_A$ we define the ordering $A_w$.
Using a straightforward dynamic programming algorithm we search for a partition of $m$ buckets on $A_w$ such that every bucket contains between $[(1-\alpha-\gamma/2)\frac{|A_i|}{m}, (1+\alpha+\gamma/2)\frac{|Ai|}{m}]$ points for each $\gee_i$. If it returns one such partition, then we continue the binary search for smaller values of $\alpha$. Otherwise, we continue the binary search for larger values of $\alpha$. In the end, we return the partition/hashmap of the last vector $w$ that the algorithm returned a valid hashmap.

\paragraph{Analysis.}
Let $w$ be a vector in $\Re^d$ and the ordering $P_w$. Let $\mathcal{I}_{w}$ be the intervals defined on the line supporting the vector $w$. 
By definition, $A_i$ is a $\gamma/(2m)$-approximation of $(P_{w}, \mathcal{I}_{w})$ with probability at least $1-1/n^{d+2}$.
Since there are $k<n$ groups and $O(n^d)$ combinatorially different vectors (with respect to $P$) using the union bound we get the next lemma.
\begin{lemma}
 \label{lem:approxAnew}
$A_i$ is a $\gamma/(2m)$-approximation in $(P_{w},\mathcal{I}_w)$ for every $w\in \Re^d$ and every $i\in [1,k]$ with probability at least $1-1/(2n)$.
\end{lemma}

Let $w^*$ be a vector in $\Re^d$ such that there exists an $\optDisc$-discrepancy partition of $P_{w^*}$.
Let $w'$ be the vector in $\mathcal{W}_A$ that belongs in the same cell of the arrangement of $\Arrang(\Lambda_A)$ with $w^*$, where $\Lambda_A$ is the set of dual hyperplanes of $A$. We notice that the ordering $A_{w'}$ is exactly the same with the ordering $A_{w^*}$.

By definition, there exists a partition of $m$ buckets/intervals on $P_{w^*}$ such that each interval contains $[(1-\optDisc)\frac{|\gee_i|}{m}, (1+\optDisc)\frac{|\gee_i|}{m}]$ items from group $\gee_i$, for every group $\gee_i$.
Let $I_1, I_2, \ldots I_m$ be the $m$ intervals defining the $\optDisc$-discrepancy partition. By the definition of $A$, it holds that $$|\frac{|\gee_i\cap I_j|}{|\gee_i|}-\frac{|A_i\cap I_j|}{|A_i|}|\leq \gamma/(2m)\Leftrightarrow$$
$$\frac{1-\optDisc-\gamma/2}{m}\leq \frac{|A_i\cap I_j|}{|A_i|}|\leq \frac{1+\optDisc+\gamma/2}{m}$$
for every $i\in[1,k]$ and $j\in[1,m]$ with probability at least $1-1/n^{d}$.
Since, the ordering of $A_{w*}$ is the same with the ordering of $A_{w'}$ it also holds that there are $I_1', \ldots, I_m'$ on the line supporting $w'$ such that 
$$\frac{1-\optDisc-\gamma/2}{m}\leq \frac{|A_i\cap I_j'|}{|A_i|}|\leq \frac{1+\optDisc+\gamma/2}{m}$$
with probability at least $1-1/n^d$.
Let $\alpha\in E$ be the discrepancy we currently check in the binary search such that $\alpha>\optDisc$.
Since there exists a $(\optDisc+\gamma/2)$-discrepancy partition in $A_{w'}$ it also holds that the straightforward dynamic programming algorithm executed with respect to vector $w'$ will return a partition satisfying $[(1-\alpha-\gamma/2)\frac{|A_i|}{m}, (1+\alpha+\gamma/2)\frac{|Ai|}{m}]$ for each $\gee_i$ with probability at least $1-1/n^d$.

Next, we show that any hashmap returned by our algorithm satisfies $(1+\delta)\optDisc+\gamma$ discrepancy, with high probability. 
Let $\alpha\in E$ be the discrepancy in the current iteration of the binary search such that $\alpha\in [\optDisc, (1+\delta)\optDisc]$ (there is always such $\alpha$ in $E$).
Let $\bar{w}\in\mathcal{W}_A$ be the vector such that our algorithm returns a valid partition with respect to $A_{\bar{w}}$ (as shown previously). Let $\bar{I}_1,\ldots, \bar{I}_m$ be the intervals on the returned valid partition defined on the line supporting $\bar{w}$ such that
\begin{equation}
 \label{eq:AAnew}
    (1-\alpha-\gamma/2)\frac{|A_i|}{m}\leq |\bar{I}_j\cap A_i|\leq (1+\alpha+\gamma/2)\frac{|A_i|}{m},
\end{equation}
for each $\gee_i$ and $j\in[1,m]$. Using Lemma~\ref{lem:approxAnew}, we get that 
\begin{equation}
 \label{eq:BBnew}
    |\frac{|\gee_i\cap \bar{I}_j|}{|\gee_i|}-\frac{|A_i\cap \bar{I}_j|}{|A_i|}|\leq \gamma/(2m)
\end{equation}
for every $i\in[1,k]$ and every $j\in[1,m]$
with probability at least $1-1/(2n)$. From Equation~\eqref{eq:AAnew} and Equation~\eqref{eq:BBnew} it follows that
$$\frac{1-(1+\delta)\optDisc-\gamma}{m}\leq \frac{|\gee_i\cap \bar{I}_j|}{|\gee_i|}\leq \frac{1+(1+\delta)\optDisc+\gamma}{m}.$$
for every $i\in[1,k]$ and $j\in[1,m]$ with probability at least $1-1/n$. The correctness of the algorithm follows.

We construct $A$ in $O(\frac{km^2}{\gamma^2}\log n)$ time. Then the binary search runs for $O(\log \frac{\log n}{\delta})$ rounds. In each round of the binary search we spend $O(\frac{k^dm^{2d}}{\gamma^{2d}}\log^{d+1} n)$ time to construct $\Arrang(\Lambda_A)$. The straightforward dynamic programming algorithm takes $O(\frac{k^2m^5}{\gamma^{4}}\log^2 n)$ time. Overall, the algorithm runs in $O(n+\log(\frac{\log n}{\delta})\frac{k^{d+2}m^{2d+5}}{\gamma^{2d+4}}\log^{d+2} n)$ time.

\begin{theorem}
Let $P$ be a set of $n$ tuples in $\Re^d$ and parameters $\delta, \gamma$. There exists an $O(n+\log(\frac{\log n}{\delta})\frac{k^{d+2}m^{2d+5}}{\gamma^{2d+4}}\log^{d+2} n)$ time algorithm such that, 
with probability at least $1-\frac{1}{n}$ it returns an $((1+\delta)\optDisc+\gamma,1)$-hashmap $\hashmap$ with collision probability at most $\frac{1+(1+\delta)\optDisc+\gamma}{m}$
and single fairness in the range $[\frac{1-(1+\delta)\optDisc-\gamma}{m},\frac{1+(1+\delta)\optDisc+\gamma}{m}]$.
\end{theorem}

\subsection{Cut-based}
We use a direct implementation of the necklace splitting algorithm proposed in~\cite{alon2021efficient} for $k>2$. Even though the authors argue that their algorithm runs in polynomial time, they do not provide an exact running time analysis.
We slightly modify their technique to work in our setting providing theoretical guarantees on the collision probability, single fairness, pairwise fairness, and pre-processing construction time. Skipping the details we give the next result.

\begin{theorem}
\label{thm:neckSplitk2}
In the non-binary demographic group cases, there exists an algorithm that finds a $\big(\eps,k(4+\log \frac{1}{\eps})\big)$-hashmap with collision probability within $[\frac{1}{m},\frac{1+\eps}{m}]$ and single fairness within $[\frac{1-\eps}{m}, \frac{1+\eps}{m}]$ in $O(mk^3\log\frac{1}{\eps}+knm(n+m))$ time. If $\eps=\frac{1}{3nm}$, the algorithm finds a $\big(0,k(4+\log n)\big)$-hashmap satisfying the collision probability and the single fairness in $O\big(mk^3\log n+knm(n+m)\big)$ time.
\end{theorem}
\pagebreak
\section{Extended Experiment Results}
\subsection{Datasets}
\noindent{\bf\adult} ~\cite{misc_adult_2} is a commonly used dataset in the fairness literature \cite{bellamy2018ai}, with column {\tt sex=\{male, female\}} as the sensitive attribute and {\tt fnlwgt} and {\tt education-num} columns with continuous real and discrete ordinal values for ordering the tuples. \adult dataset includes 15 census data related attributes with $\sim$49000 records describing people and is primarily used to predict whether an individual's income exceeds \$50K per year.

\noindent{\bf\compas} ~\cite{COMPAS} is another benchmark dataset widely used in the fairness literature. This dataset includes $\sim$61000 records of defendants from Broward County from 2013 and 2014 with 29 attributes including demographics and criminal history. \compas was originally used to predict the likelihood of re-offense by the criminal defendants. In our experiments, we picked {\tt sex=\{male, female\}} and {\tt race=\{White, Black, Hispanic, other\}} as the sensitive attributes and {\tt Person\_ID} and {\tt Raw\_Score} with discrete ordinal and continuous real values as the columns used for building the hashmap.

\noindent{\bf\diabetes} ~\cite{strack2014impact} is a dataset of $\sim$102K records of patients with diabetes collected over 10 years from 130 US hospitals and integrated delivery networks. \diabetes has 49 attributes from which we chose {\tt sex=\{male, female\}} as the sensitive attribute and {\tt encounter\_id} and {\tt patient\_nbr} columns with discrete ordinal values for constructing the hashmap.

\noindent{\bf\popsim} ~\cite{khanh2023popsim} is a semi-synthetic dataset used to simulate individual-level data with demographic information for the city of Chicago. This dataset is available in various sizes and in our experiments we use the variant with 1M entries. It has 5 attributes among which we use {\tt race=\{White, Black\} as the sensitive attribute and {\tt latitude} and {\tt longitude}} columns with continuous real values for building the hashmap.
\subsection{Non-binary Sensitive Attribute Evaluation}
We repeated our experiments for \ranker and \pd algorithms using the \compas dataset, with the non-binary sensitive attribute chosen as {\tt race}. The results, as illustrated in Figures~\ref{fig:compas_non_binary_n_vs_unfairness}-\ref{fig:compas_non_binary_n_vs_query_time}, align with the findings from our earlier experiments. Consistent with our expectations, our algorithms exhibit independence from the number of demographic groups and seamlessly extend to non-binary sensitive attributes without compromising efficiency or memory requirements. 

\begin{figure*}[!tbh]
    \centering
    \includegraphics[width=0.6\textwidth]{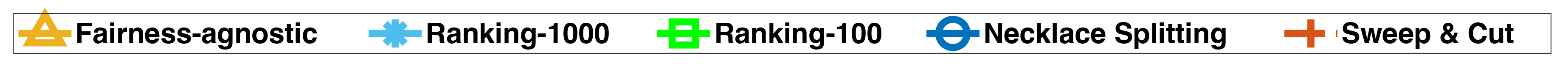}
    \vspace{-3mm}
\end{figure*}
\begin{figure*}[!tb]
\begin{minipage}[t]{\linewidth}
    \begin{subfigure}[t]{0.24\textwidth}
        \centering
        \includegraphics[width=\textwidth]{plots/adult/adult_unfairness_varying_size.pdf}
        \caption{\adult}
    \end{subfigure}
    \hfill
    \begin{subfigure}[t]{0.24\textwidth}
        \centering
        \includegraphics[width=\textwidth]{plots/compas/compas_unfairness_varying_size.pdf}
        \caption{\compas}
    \end{subfigure}
    \hfill
    \begin{subfigure}[t]{0.24\textwidth}
        \centering
        \includegraphics[width=\textwidth]{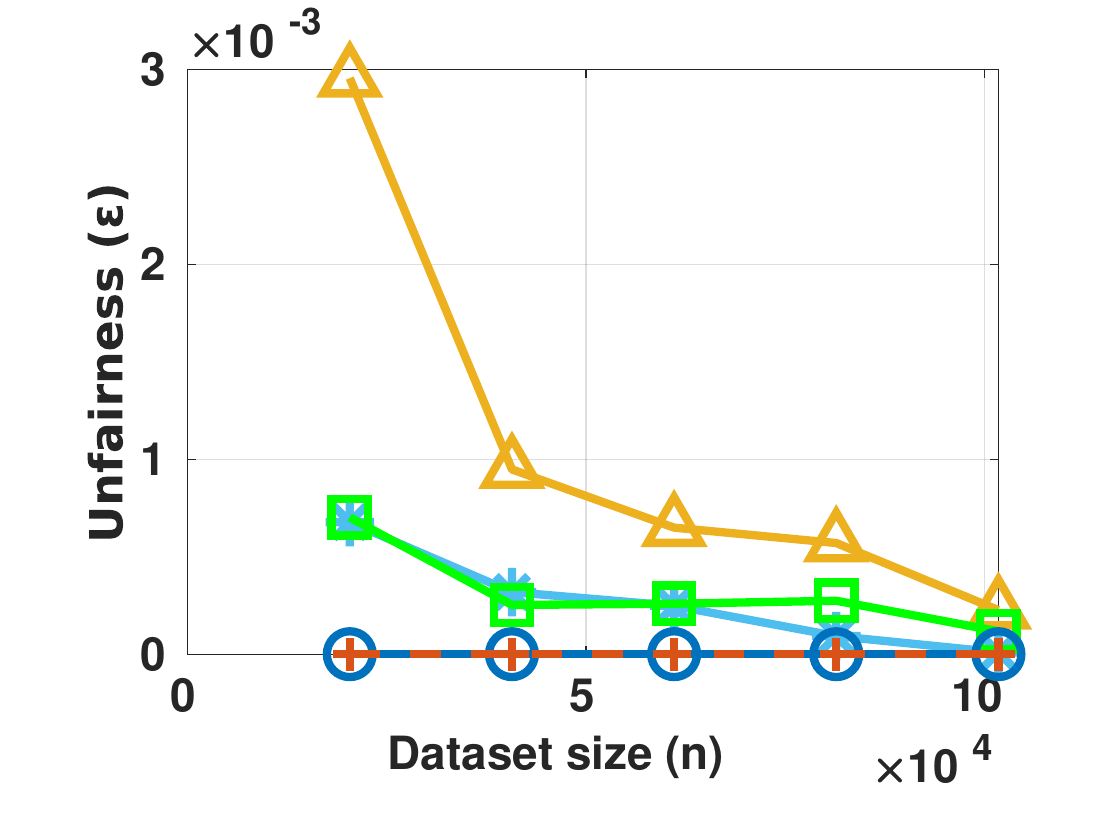}
        \caption{\diabetes}
    \end{subfigure}
    \hfill
    \begin{subfigure}[t]{0.24\textwidth}
        \centering
        \includegraphics[width=\textwidth]{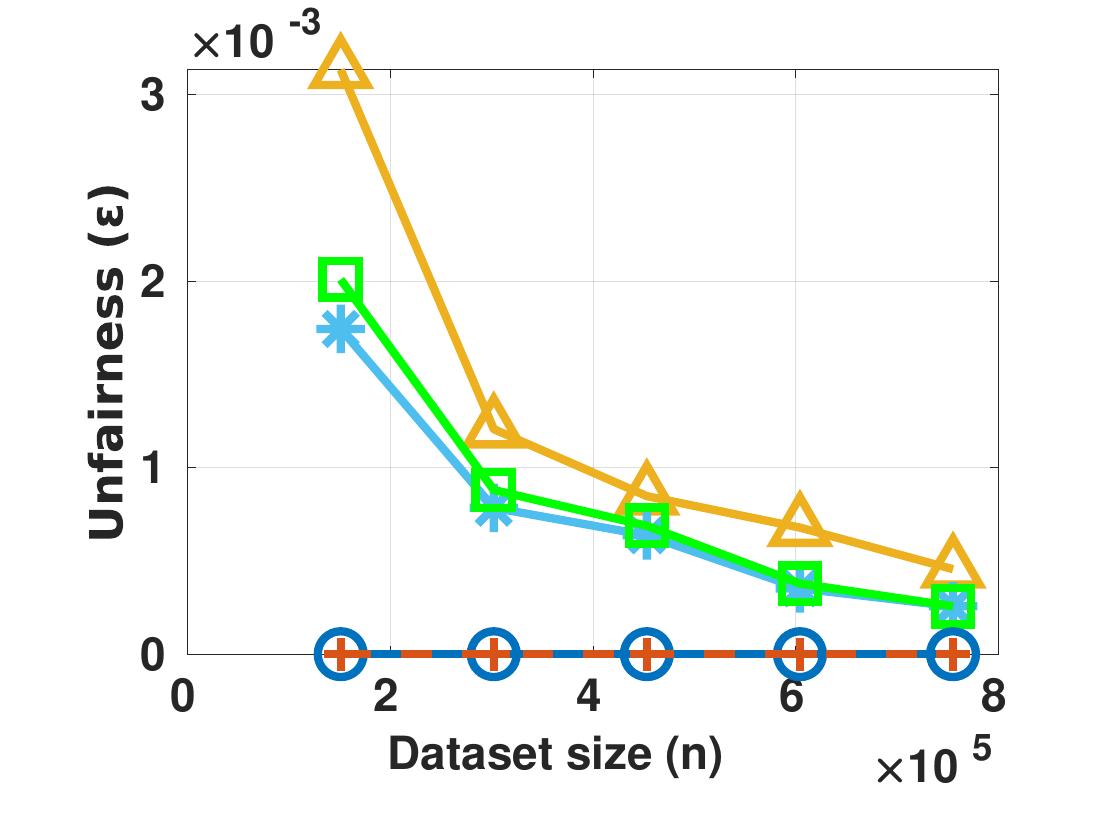}
        \caption{\popsim}
    \end{subfigure}
\caption{Effect of varying dataset size $n$ on unfairness}
\end{minipage}
\end{figure*}

\begin{figure*}[!tb]
\begin{minipage}[t]{\linewidth}
    \begin{subfigure}[t]{0.24\textwidth}
        \centering
        \includegraphics[width=\textwidth]{plots/adult/adult_unfairness_varying_ratio.pdf}
        \caption{\adult}
    \end{subfigure}
    \hfill
    \begin{subfigure}[t]{0.24\textwidth}
        \centering
        \includegraphics[width=\textwidth]{plots/compas/compas_unfairness_varying_ratio.pdf}
        \caption{\compas}
    \end{subfigure}
    \hfill
    \begin{subfigure}[t]{0.24\textwidth}
        \centering
        \includegraphics[width=\textwidth]{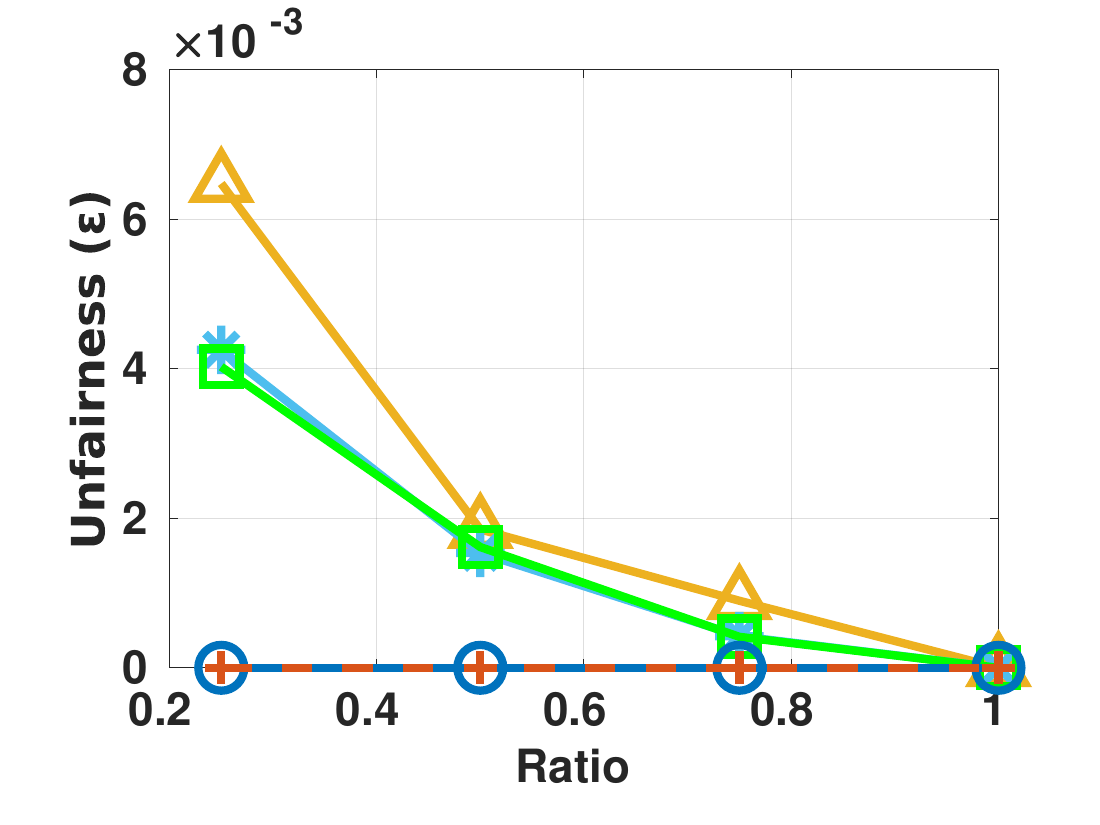}
        \caption{\diabetes}
    \end{subfigure}
    \hfill
    \begin{subfigure}[t]{0.24\textwidth}
        \centering
        \includegraphics[width=\textwidth]{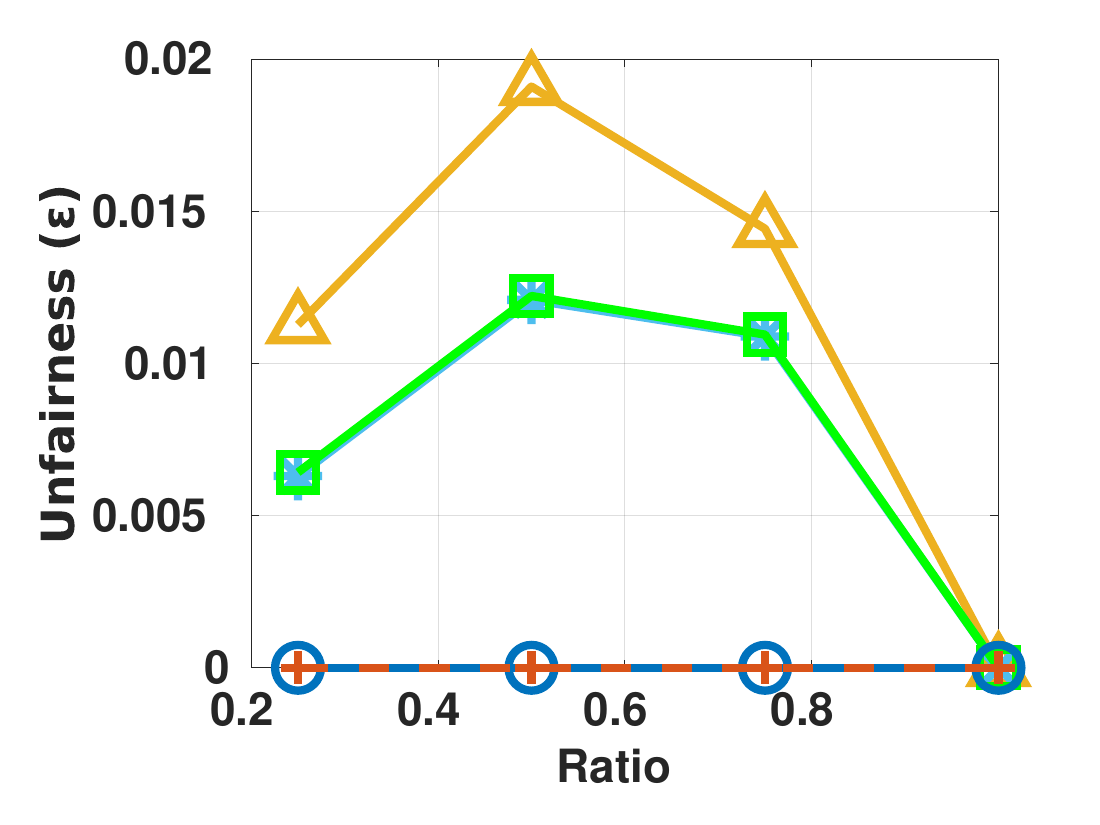}
        \caption{\popsim}
    \end{subfigure}
    \caption{Effect of varying majority to minority ratio on unfairness}
\end{minipage}
\end{figure*}

\begin{figure*}[!tb]
\begin{minipage}[t]{\linewidth}
    \begin{subfigure}[t]{0.24\textwidth}
        \centering
        \includegraphics[width=\textwidth]{plots/adult/adult_unfairness_varying_number_of_buckets.pdf}
        \caption{\adult}
    \end{subfigure}
    \hfill
    \begin{subfigure}[t]{0.24\textwidth}
        \centering
        \includegraphics[width=\textwidth]{plots/compas/compas_unfairness_varying_number_of_buckets.pdf}
        \caption{\compas}
    \end{subfigure}
    \hfill
    \begin{subfigure}[t]{0.24\textwidth}
        \centering
        \includegraphics[width=\textwidth]{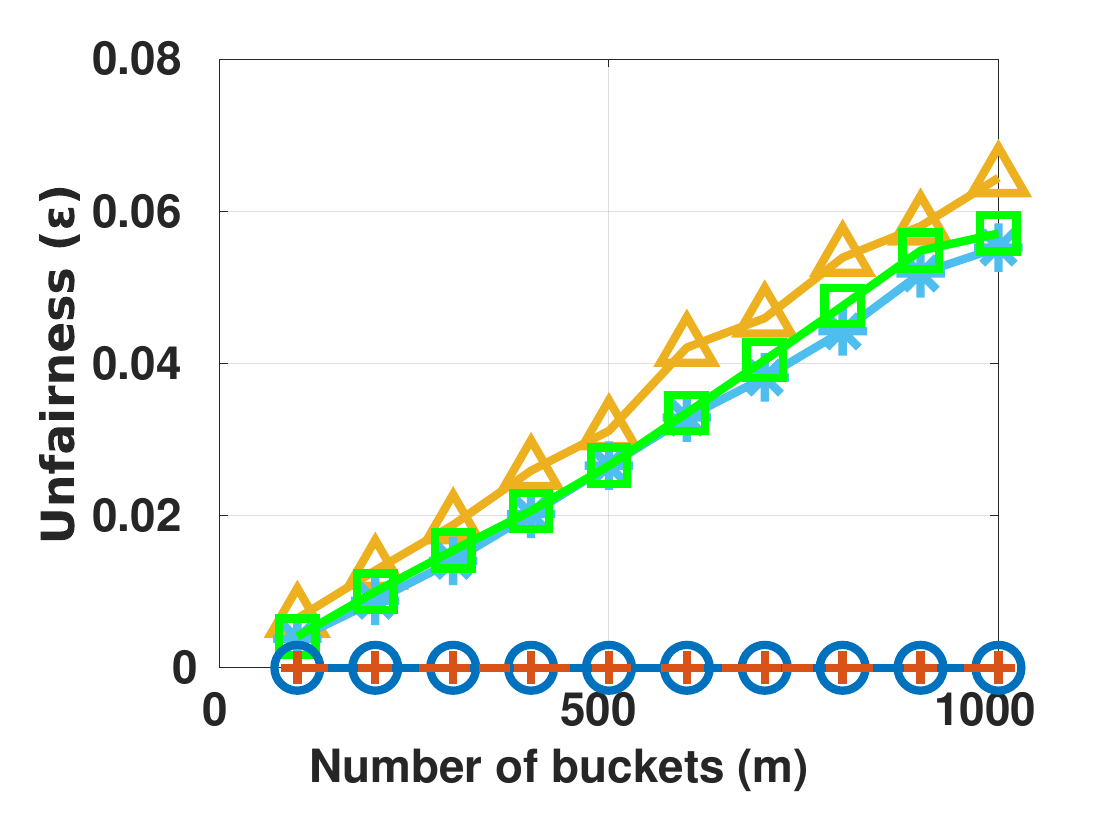}
        \caption{\diabetes}
    \end{subfigure}
    \hfill
    \begin{subfigure}[t]{0.24\textwidth}
        \centering
        \includegraphics[width=\textwidth]{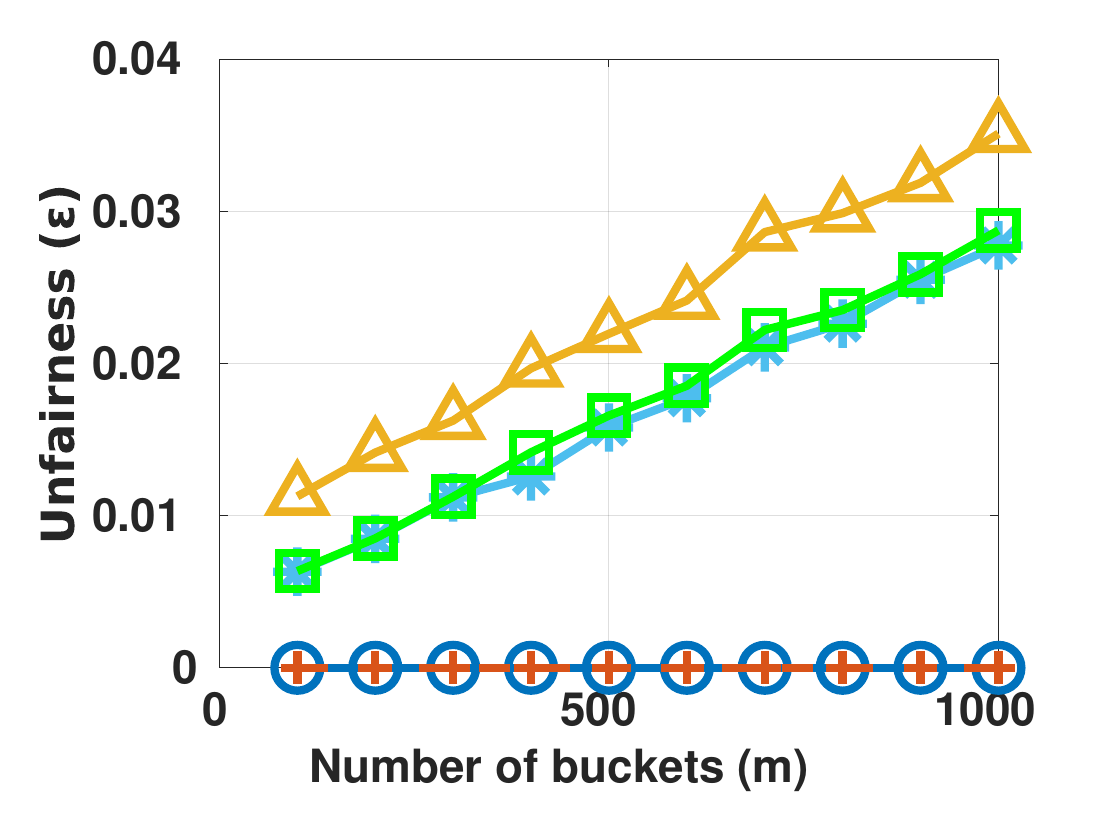}
        \caption{\popsim}
    \end{subfigure}
    \caption{Effect of varying number of buckets $m$ on unfairness}
\end{minipage}
\end{figure*}
\begin{figure*}[!tb]
\begin{minipage}[t]{\linewidth}
    \begin{subfigure}[t]{0.24\textwidth}
        \centering
        \includegraphics[width=\textwidth]{plots/adult/adult_space_varying_size.pdf}
        \caption{\adult}
    \end{subfigure}
    \hfill
    \begin{subfigure}[t]{0.24\textwidth}
        \centering
        \includegraphics[width=\textwidth]{plots/compas/compas_space_varying_size.pdf}
        \caption{\compas}
    \end{subfigure}
    \hfill
    \begin{subfigure}[t]{0.24\textwidth}
        \centering
        \includegraphics[width=\textwidth]{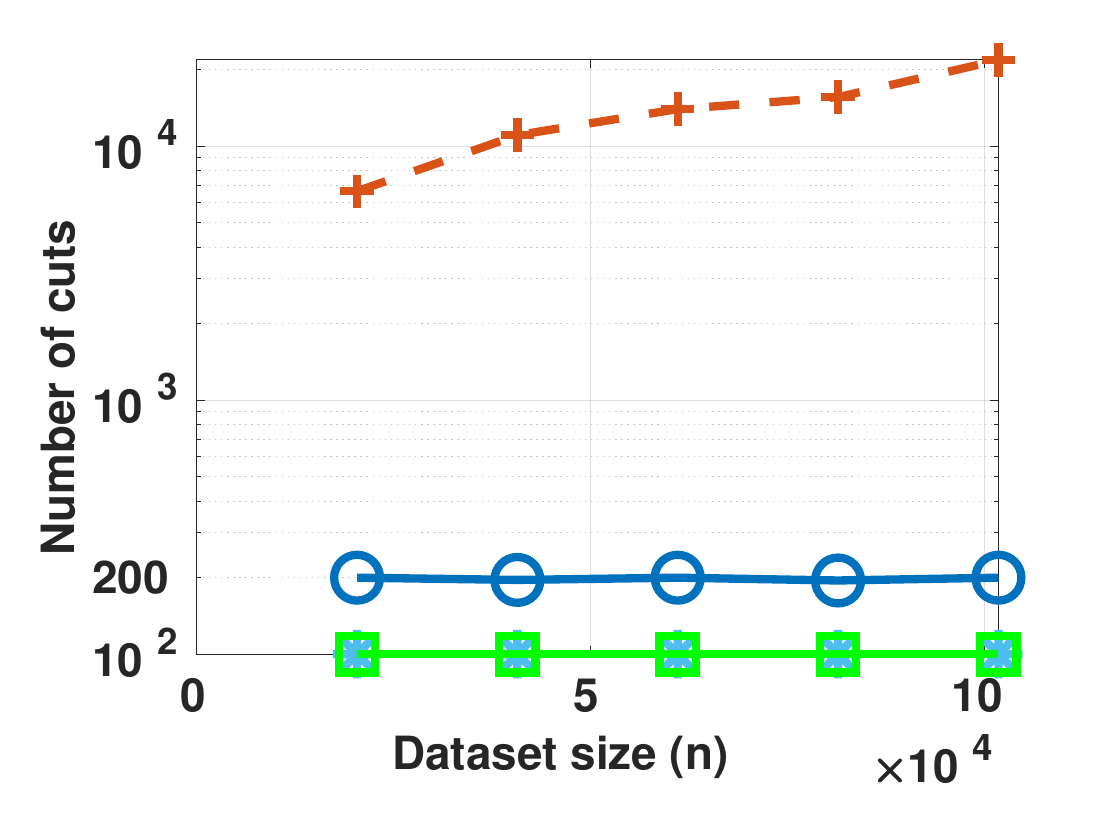}
        \caption{\diabetes}
    \end{subfigure}
    \hfill
    \begin{subfigure}[t]{0.24\textwidth}
        \centering
        \includegraphics[width=\textwidth]{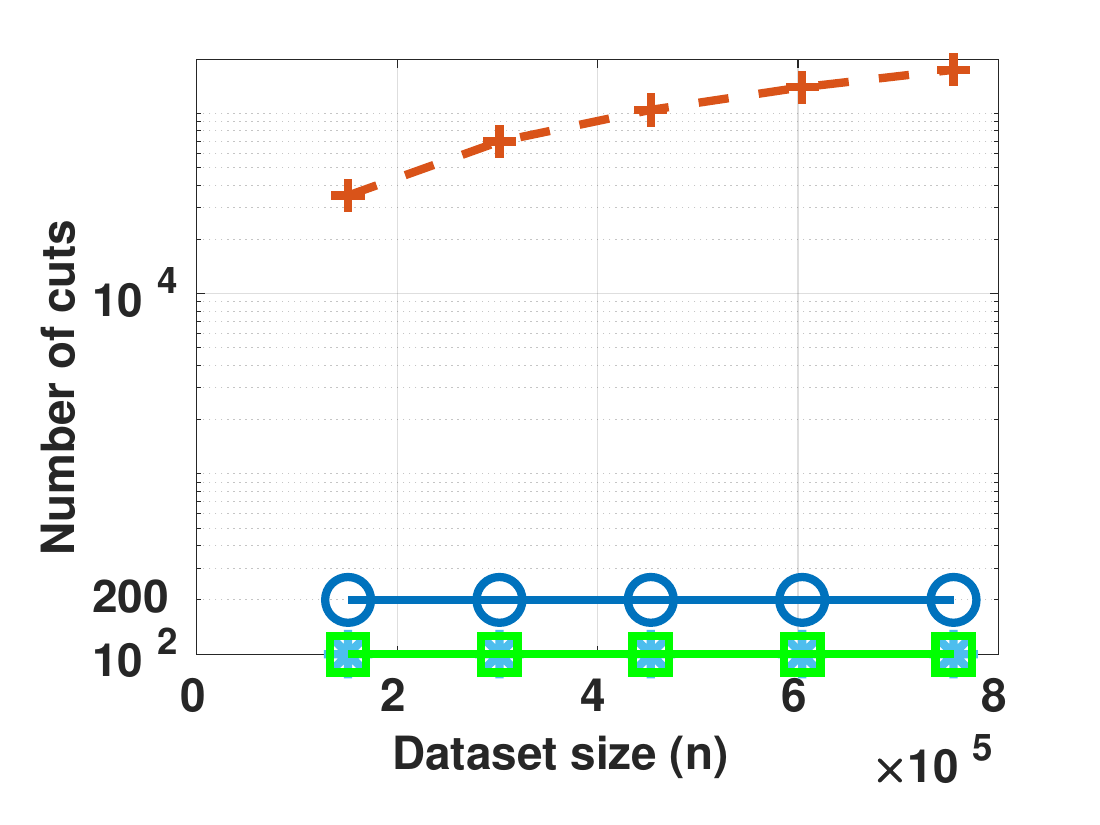}
        \caption{\popsim}
    \end{subfigure}
\caption{Effect of varying dataset size $n$ on space}
\end{minipage}
\end{figure*}

\begin{figure*}[!tb]
\begin{minipage}[t]{\linewidth}
    \begin{subfigure}[t]{0.24\textwidth}
        \centering
        \includegraphics[width=\textwidth]{plots/adult/adult_space_varying_ratio.pdf}
        \caption{\adult}
    \end{subfigure}
    \hfill
    \begin{subfigure}[t]{0.24\textwidth}
        \centering
        \includegraphics[width=\textwidth]{plots/compas/compas_space_varying_ratio.pdf}
        \caption{\compas}
    \end{subfigure}
    \hfill
    \begin{subfigure}[t]{0.24\textwidth}
        \centering
        \includegraphics[width=\textwidth]{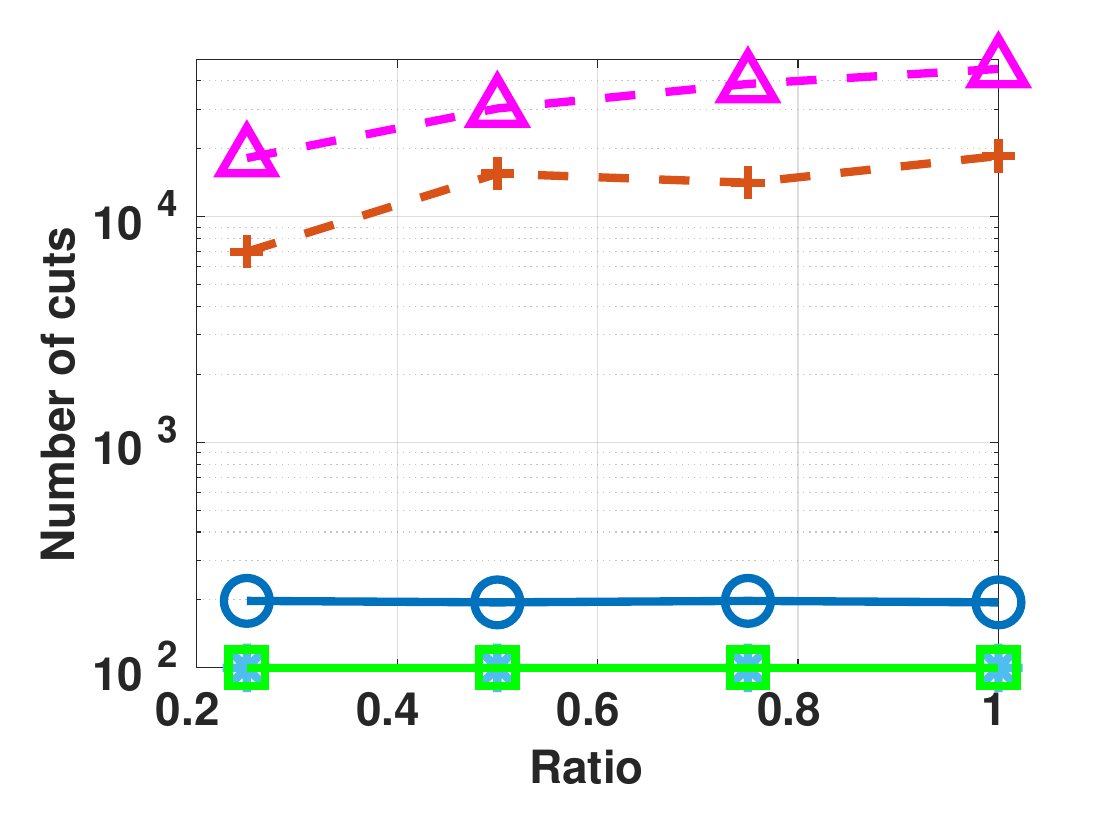}
        \caption{\diabetes}
    \end{subfigure}
    \hfill
    \begin{subfigure}[t]{0.24\textwidth}
        \centering
        \includegraphics[width=\textwidth]{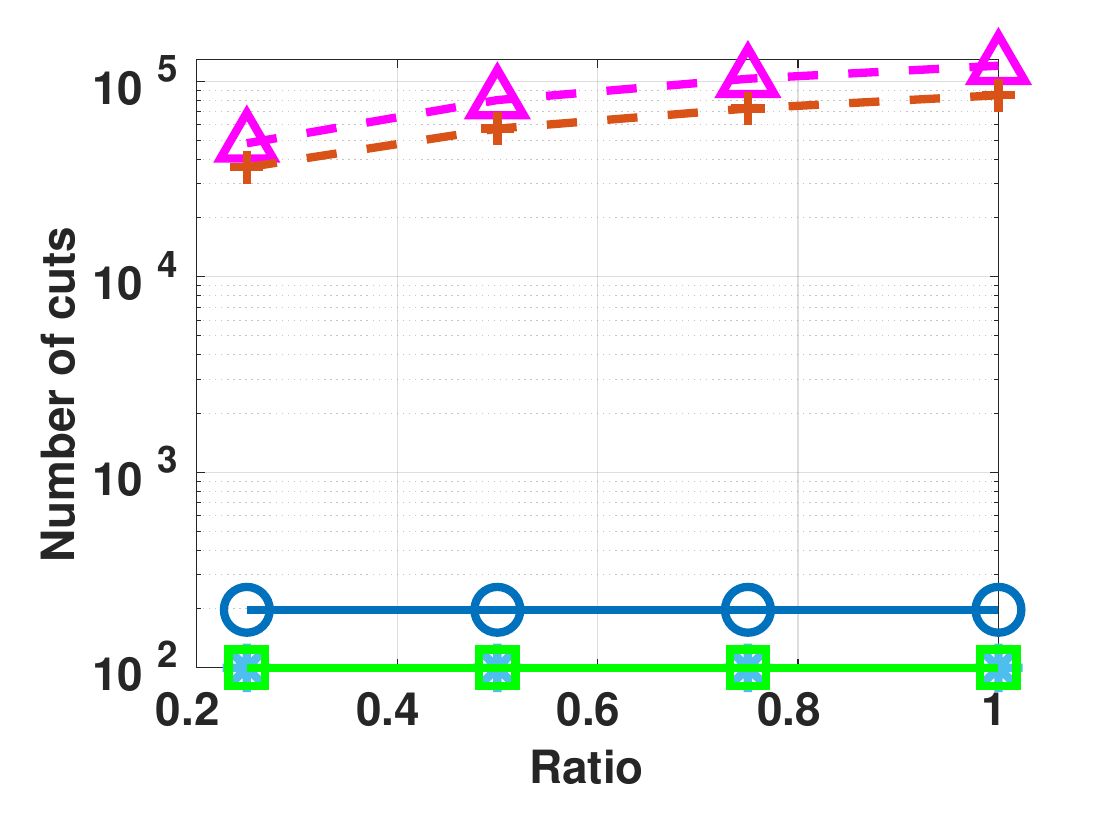}
        \caption{\popsim}
    \end{subfigure}
    \caption{Effect of varying majority to minority ratio on space}
\end{minipage}
\end{figure*}

\begin{figure*}[!tb]
\begin{minipage}[t]{\linewidth}
    \begin{subfigure}[t]{0.24\textwidth}
        \centering
        \includegraphics[width=\textwidth]{plots/adult/adult_space_varying_number_of_buckets.pdf}
        \caption{\adult}
    \end{subfigure}
    \hfill
    \begin{subfigure}[t]{0.24\textwidth}
        \centering
        \includegraphics[width=\textwidth]{plots/compas/compas_space_varying_number_of_buckets.pdf}
        \caption{\compas}
    \end{subfigure}
    \hfill
    \begin{subfigure}[t]{0.24\textwidth}
        \centering
        \includegraphics[width=\textwidth]{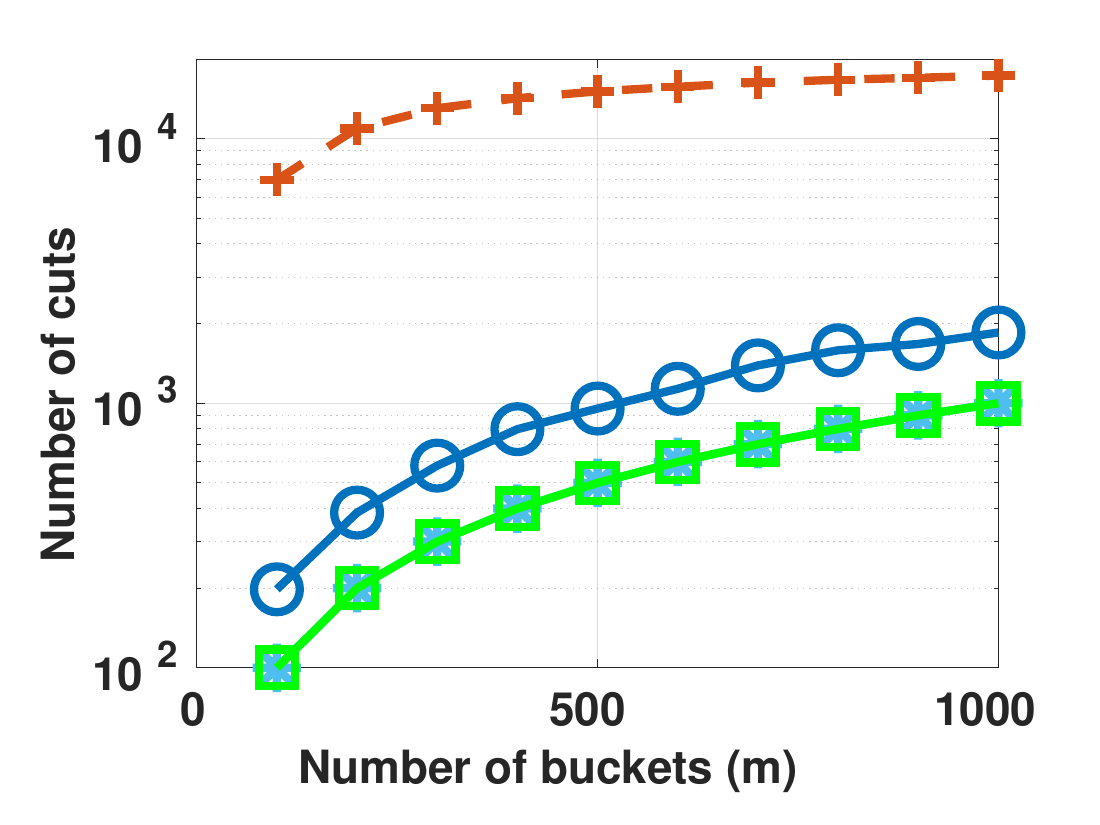}
        \caption{\diabetes}
    \end{subfigure}
    \hfill
    \begin{subfigure}[t]{0.24\textwidth}
        \centering
        \includegraphics[width=\textwidth]{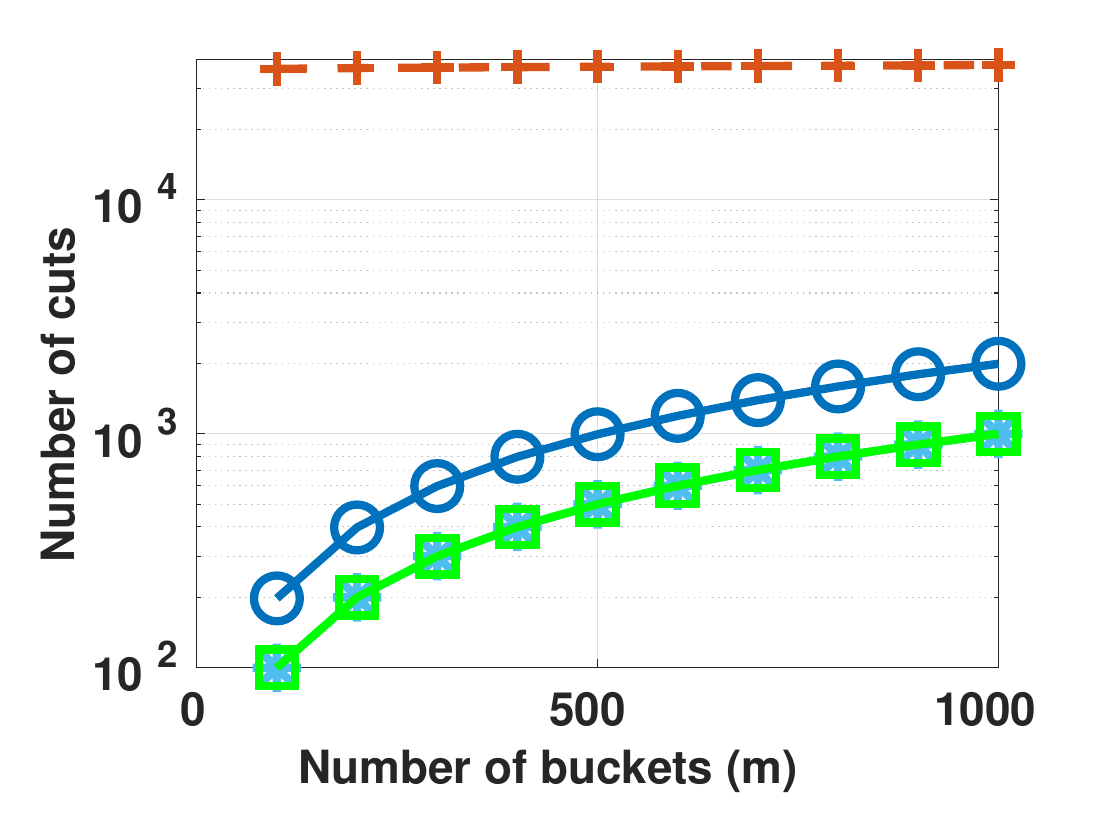}
        \caption{\popsim}
    \end{subfigure}
    \caption{Effect of varying number of buckets $m$ on space}
\end{minipage}
\end{figure*}
\begin{figure*}[!tb]
\begin{minipage}[t]{\linewidth}
    \begin{subfigure}[t]{0.24\textwidth}
        \centering
        \includegraphics[width=\textwidth]{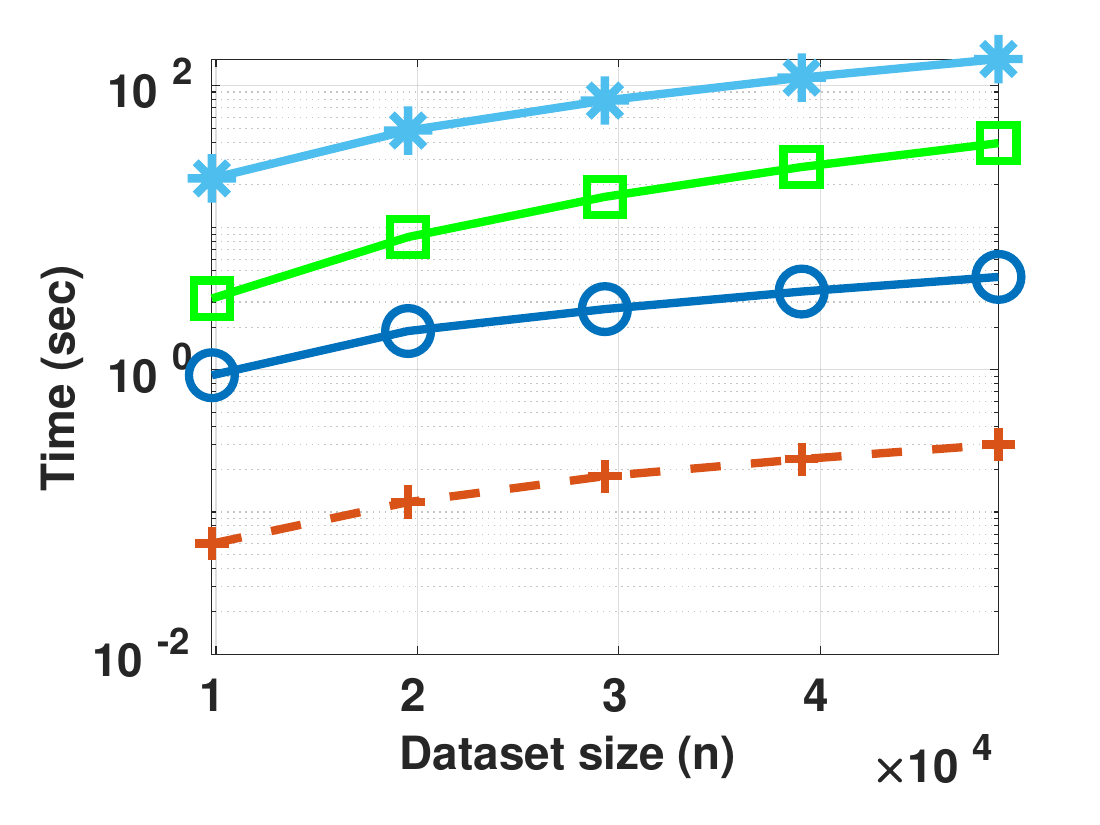}
        \caption{\adult}
    \end{subfigure}
    \hfill
    \begin{subfigure}[t]{0.24\textwidth}
        \centering
        \includegraphics[width=\textwidth]{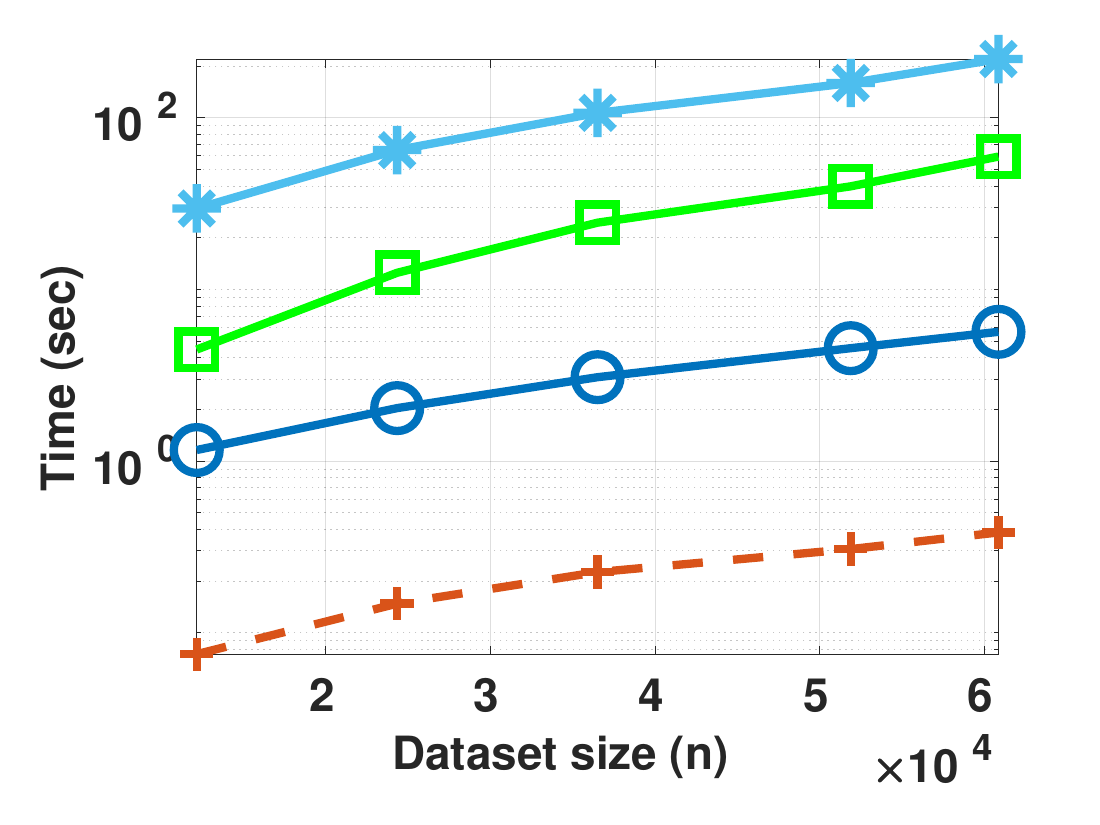}
        \caption{\compas}
    \end{subfigure}
    \hfill
    \begin{subfigure}[t]{0.24\textwidth}
        \centering
        \includegraphics[width=\textwidth]{plots/diabetes/diabetes_prep_time_varying_size.pdf}
        \caption{\diabetes}
    \end{subfigure}
    \hfill
    \begin{subfigure}[t]{0.24\textwidth}
        \centering
        \includegraphics[width=\textwidth]{plots/popsim/popsim_prep_time_varying_size.pdf}
        \caption{\popsim}
    \end{subfigure}
\caption{Effect of varying dataset size $n$ on preprocessing time}
\end{minipage}
\end{figure*}

\begin{figure*}[!tb]
\begin{minipage}[t]{\linewidth}
    \begin{subfigure}[t]{0.24\textwidth}
        \centering
        \includegraphics[width=\textwidth]{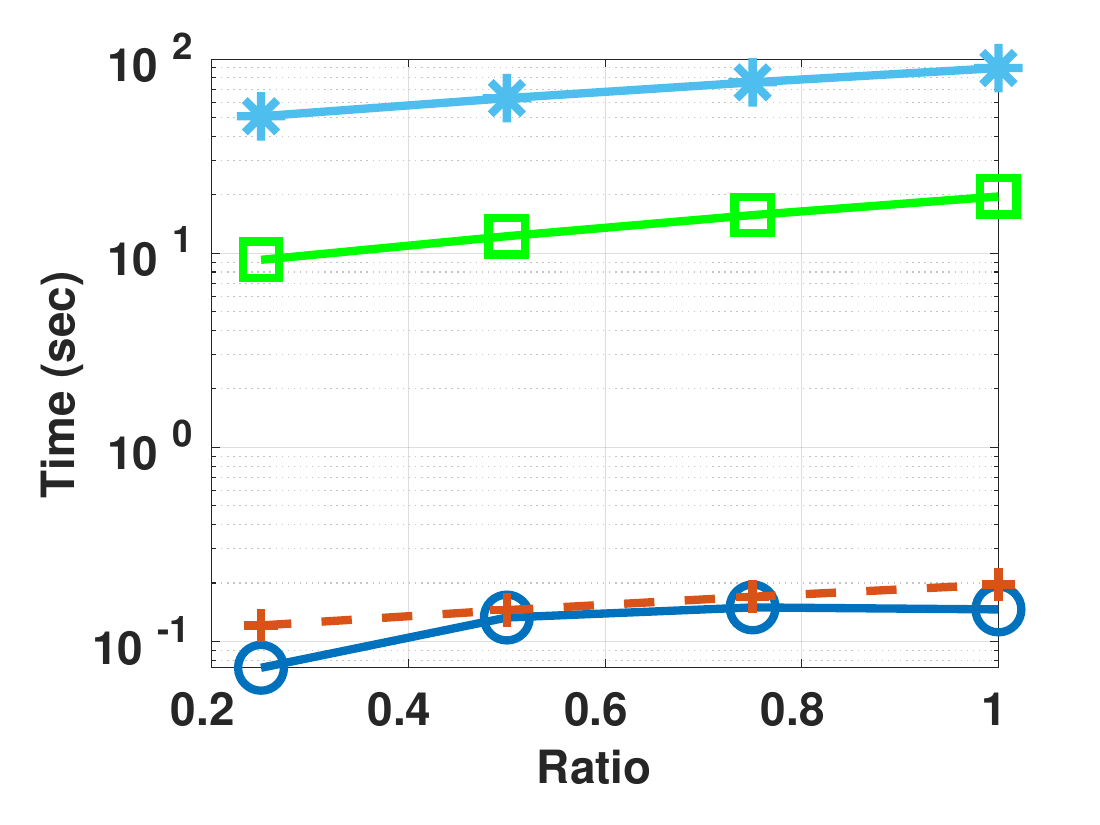}
        \caption{\adult}
    \end{subfigure}
    \hfill
    \begin{subfigure}[t]{0.24\textwidth}
        \centering
        \includegraphics[width=\textwidth]{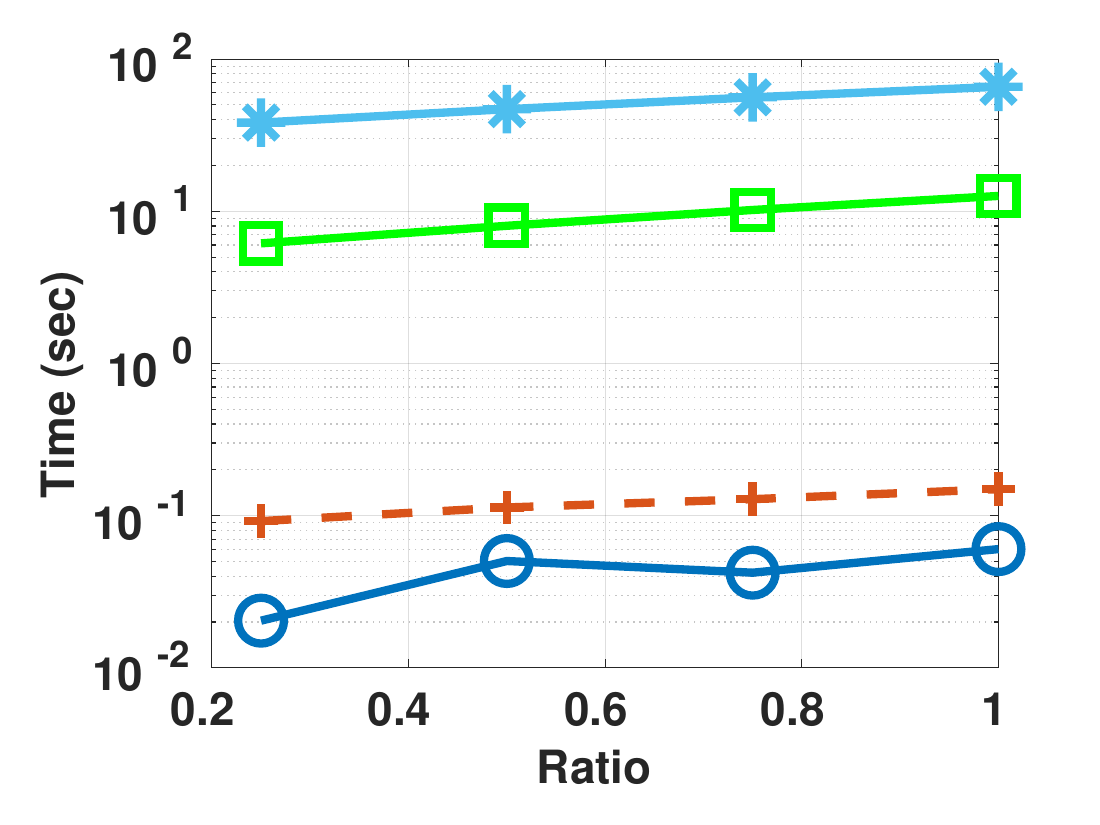}
        \caption{\compas}
    \end{subfigure}
    \hfill
    \begin{subfigure}[t]{0.24\textwidth}
        \centering
        \includegraphics[width=\textwidth]{plots/diabetes/diabetes_prep_time_varying_ratio.pdf}
        \caption{\diabetes}
    \end{subfigure}
    \hfill
    \begin{subfigure}[t]{0.24\textwidth}
        \centering
        \includegraphics[width=\textwidth]{plots/popsim/popsim_prep_time_varying_ratio.pdf}
        \caption{\popsim}
    \end{subfigure}
    \caption{Effect of varying majority to minority ratio on preprocessing time}
\end{minipage}
\end{figure*}

\begin{figure*}[!tb]
\begin{minipage}[t]{\linewidth}
    \begin{subfigure}[t]{0.24\textwidth}
        \centering
        \includegraphics[width=\textwidth]{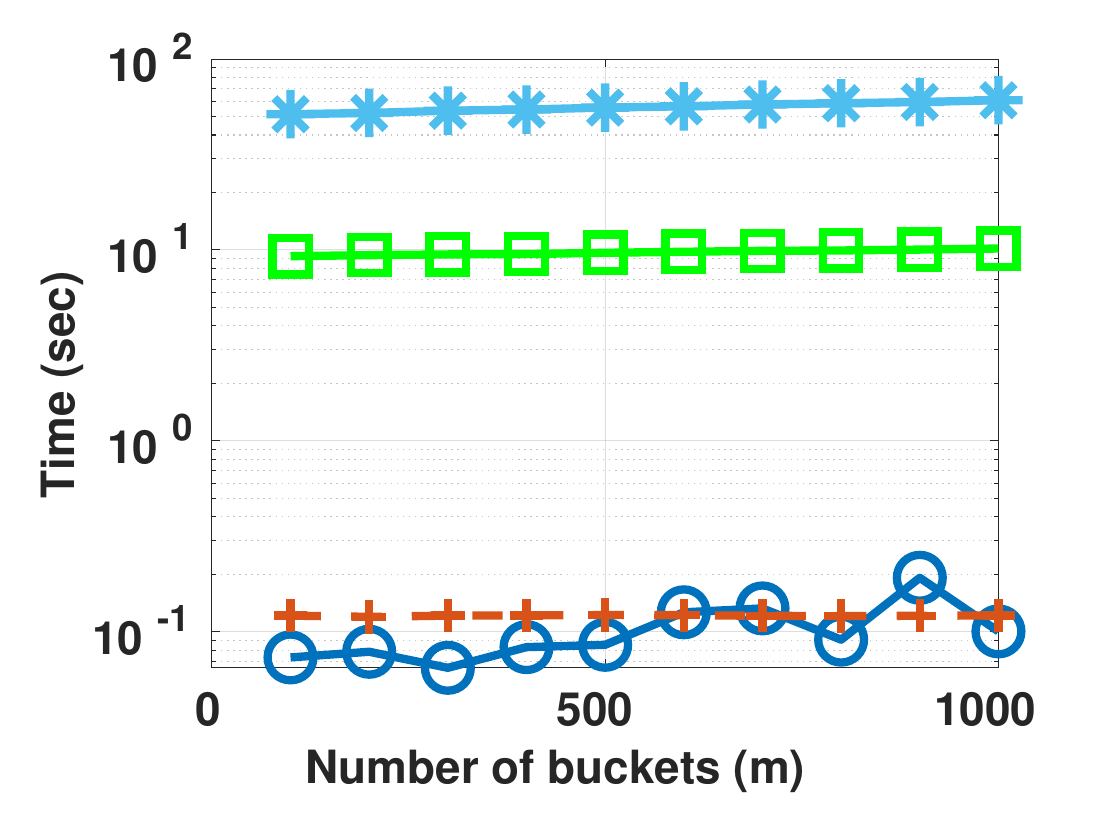}
        \caption{\adult}
    \end{subfigure}
    \hfill
    \begin{subfigure}[t]{0.24\textwidth}
        \centering
        \includegraphics[width=\textwidth]{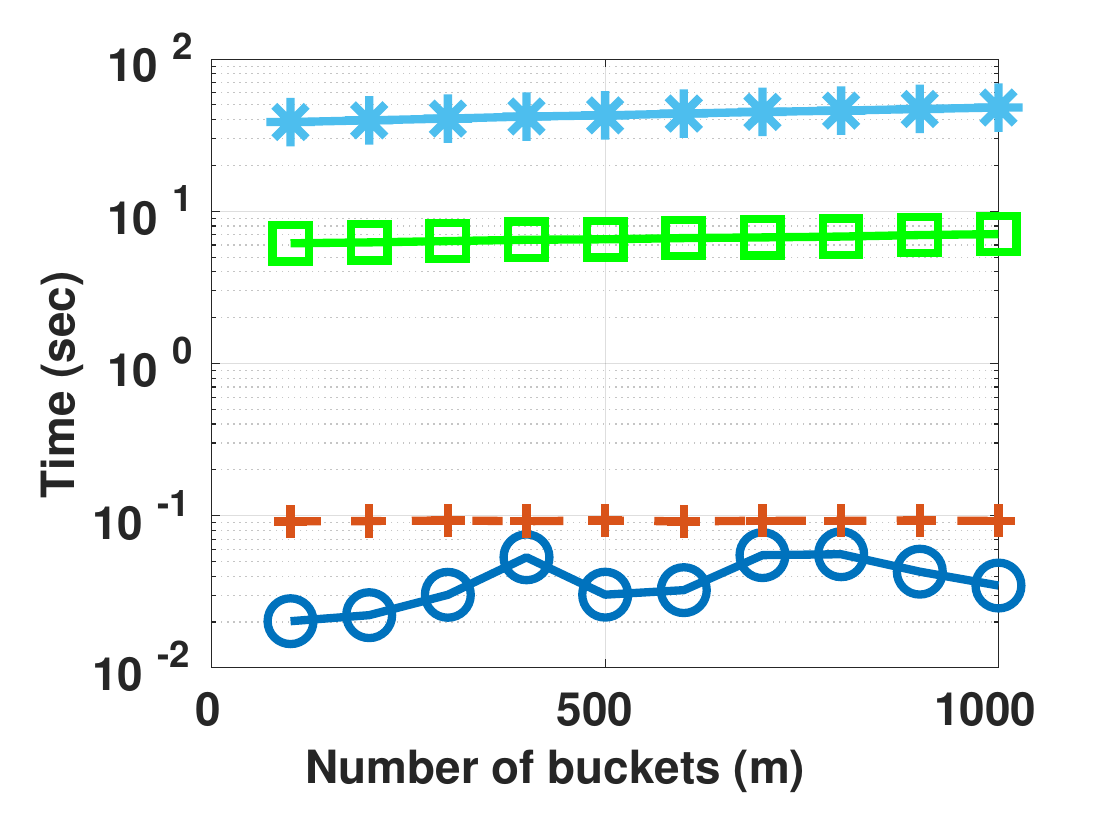}
        \caption{\compas}
    \end{subfigure}
    \hfill
    \begin{subfigure}[t]{0.24\textwidth}
        \centering
        \includegraphics[width=\textwidth]{plots/diabetes/diabetes_prep_time_varying_number_of_buckets.pdf}
        \caption{\diabetes}
    \end{subfigure}
    \hfill
    \begin{subfigure}[t]{0.24\textwidth}
        \centering
        \includegraphics[width=\textwidth]{plots/popsim/popsim_prep_time_varying_number_of_buckets.pdf}
        \caption{\popsim}
    \end{subfigure}
    \caption{Effect of varying number of buckets $m$ on preprocessing time}
\end{minipage}
\end{figure*}

\begin{figure*}[!tb]
\begin{minipage}[t]{\linewidth}
    \begin{subfigure}[t]{0.24\textwidth}
        \centering
        \includegraphics[width=\textwidth]{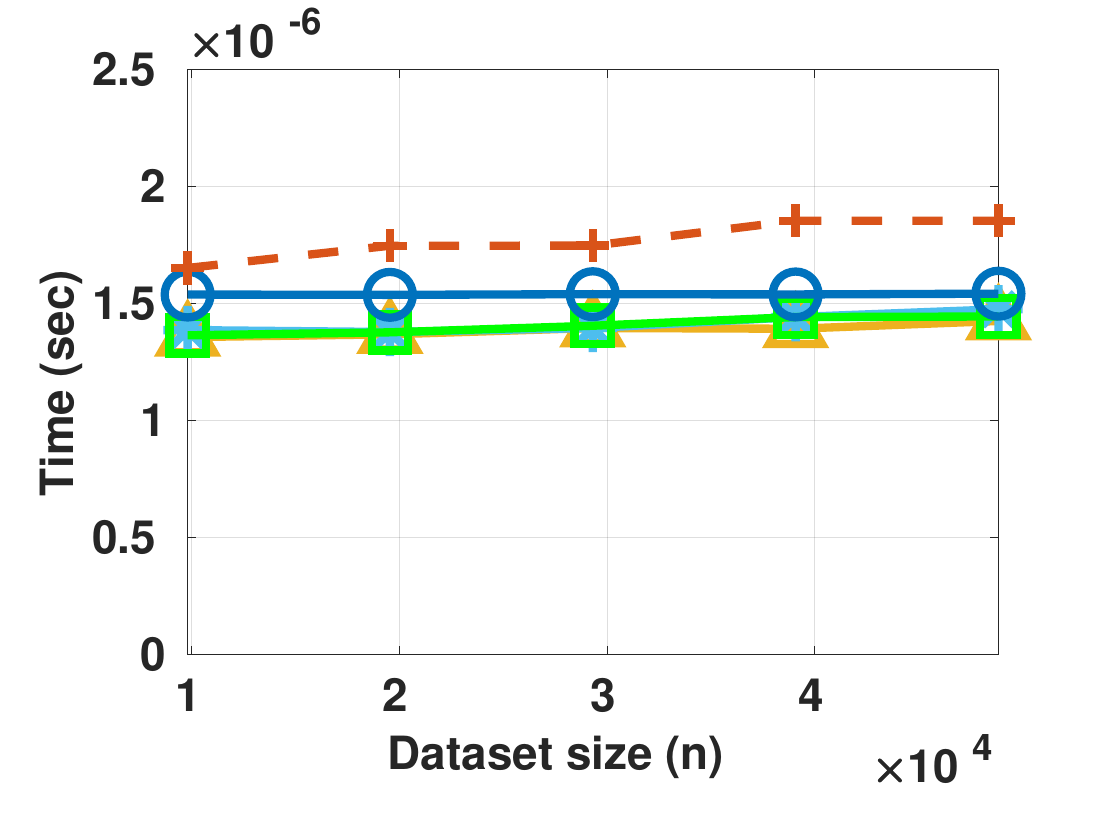}
        \caption{\adult}
    \end{subfigure}
    \hfill
    \begin{subfigure}[t]{0.24\textwidth}
        \centering
        \includegraphics[width=\textwidth]{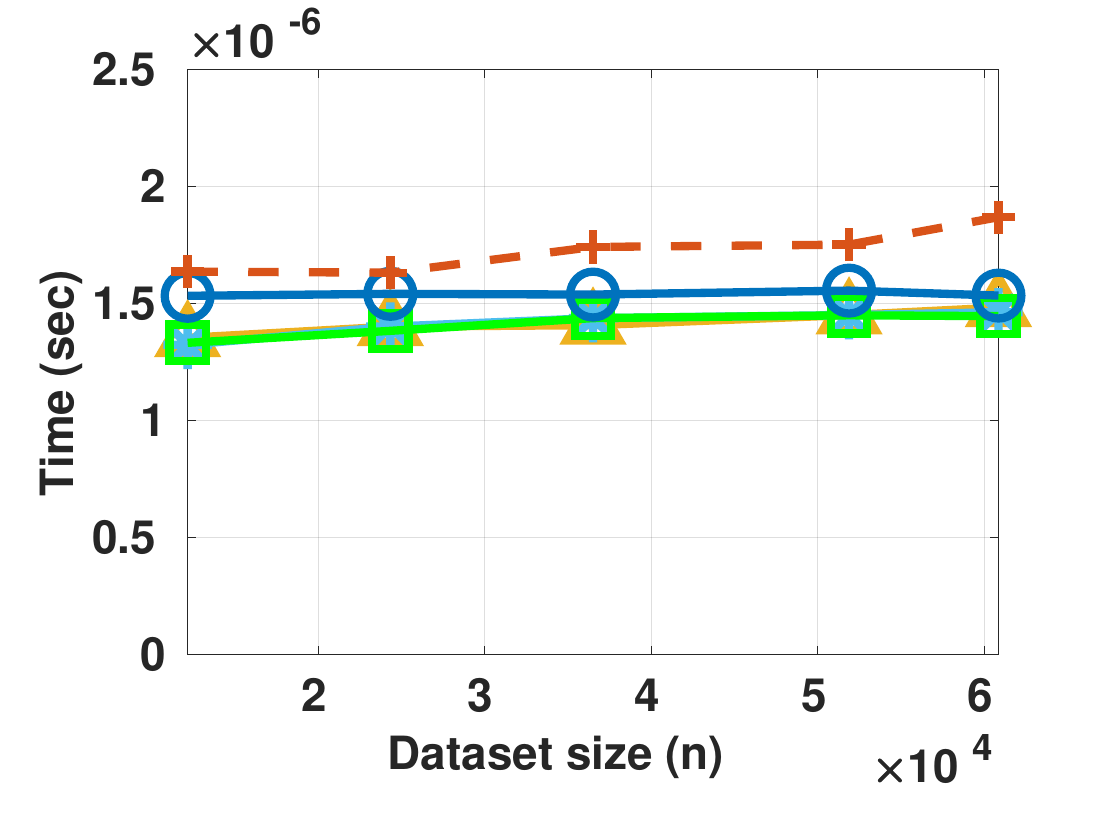}
        \caption{\compas}
    \end{subfigure}
    \hfill
    \begin{subfigure}[t]{0.24\textwidth}
        \centering
        \includegraphics[width=\textwidth]{plots/diabetes/diabetes_query_time_varying_size.pdf}
        \caption{\diabetes}
    \end{subfigure}
    \hfill
    \begin{subfigure}[t]{0.24\textwidth}
        \centering
        \includegraphics[width=\textwidth]{plots/popsim/popsim_query_time_varying_size.pdf}
        \caption{\popsim}
    \end{subfigure}
\caption{Effect of varying dataset size $n$ on query time}
\end{minipage}
\end{figure*}

\begin{figure*}[!tb]
\begin{minipage}[t]{\linewidth}
    \begin{subfigure}[t]{0.24\textwidth}
        \centering
        \includegraphics[width=\textwidth]{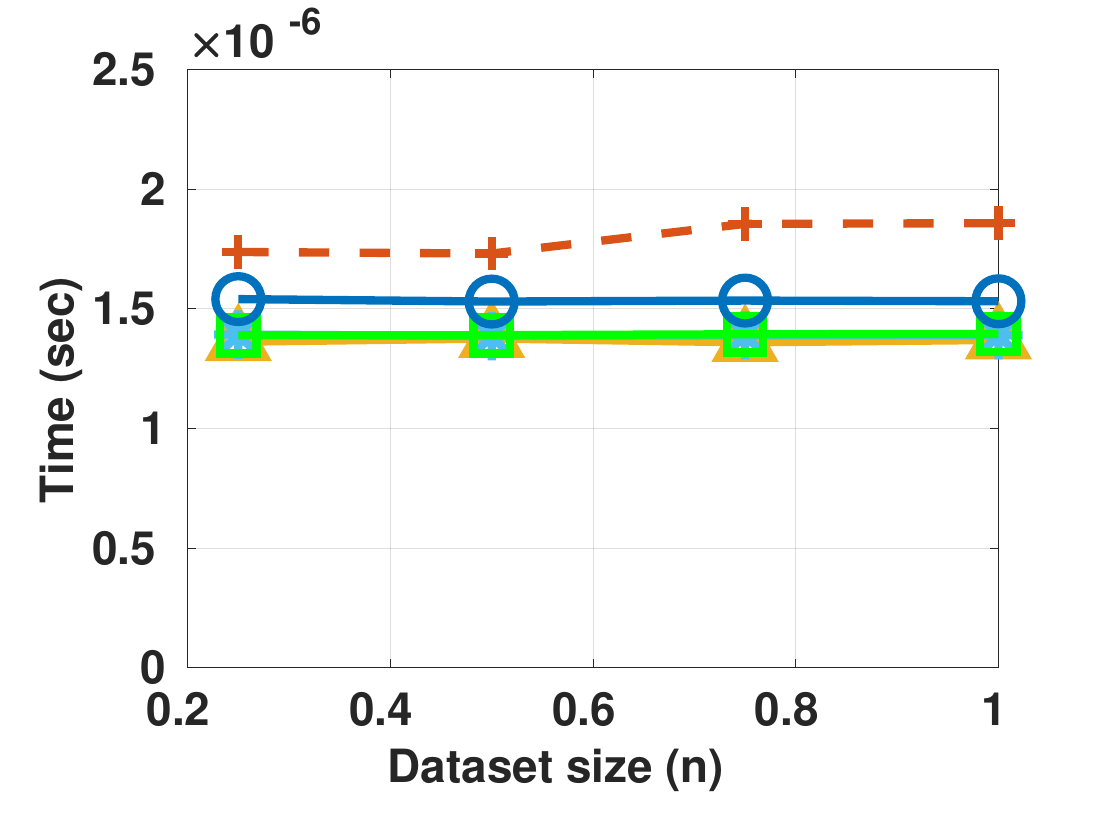}
        \caption{\adult}
    \end{subfigure}
    \hfill
    \begin{subfigure}[t]{0.24\textwidth}
        \centering
        \includegraphics[width=\textwidth]{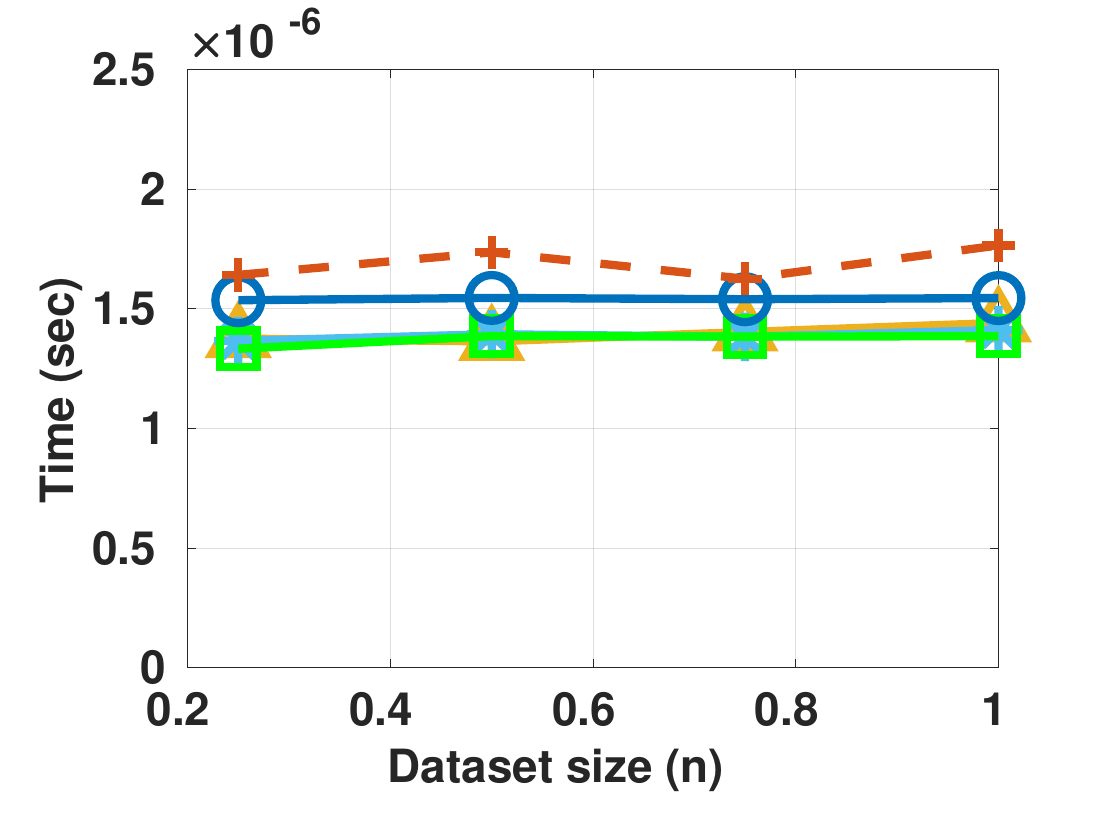}
        \caption{\compas}
    \end{subfigure}
    \hfill
    \begin{subfigure}[t]{0.24\textwidth}
        \centering
        \includegraphics[width=\textwidth]{plots/diabetes/diabetes_query_time_varying_ratio.pdf}
        \caption{\diabetes}
    \end{subfigure}
    \hfill
    \begin{subfigure}[t]{0.24\textwidth}
        \centering
        \includegraphics[width=\textwidth]{plots/popsim/popsim_query_time_varying_ratio.pdf}
        \caption{\popsim}
    \end{subfigure}
    \caption{Effect of varying majority to minority ratio on query time}
\end{minipage}
\end{figure*}

\begin{figure*}[!tb]
\begin{minipage}[t]{\linewidth}
    \begin{subfigure}[t]{0.24\textwidth}
        \centering
        \includegraphics[width=\textwidth]{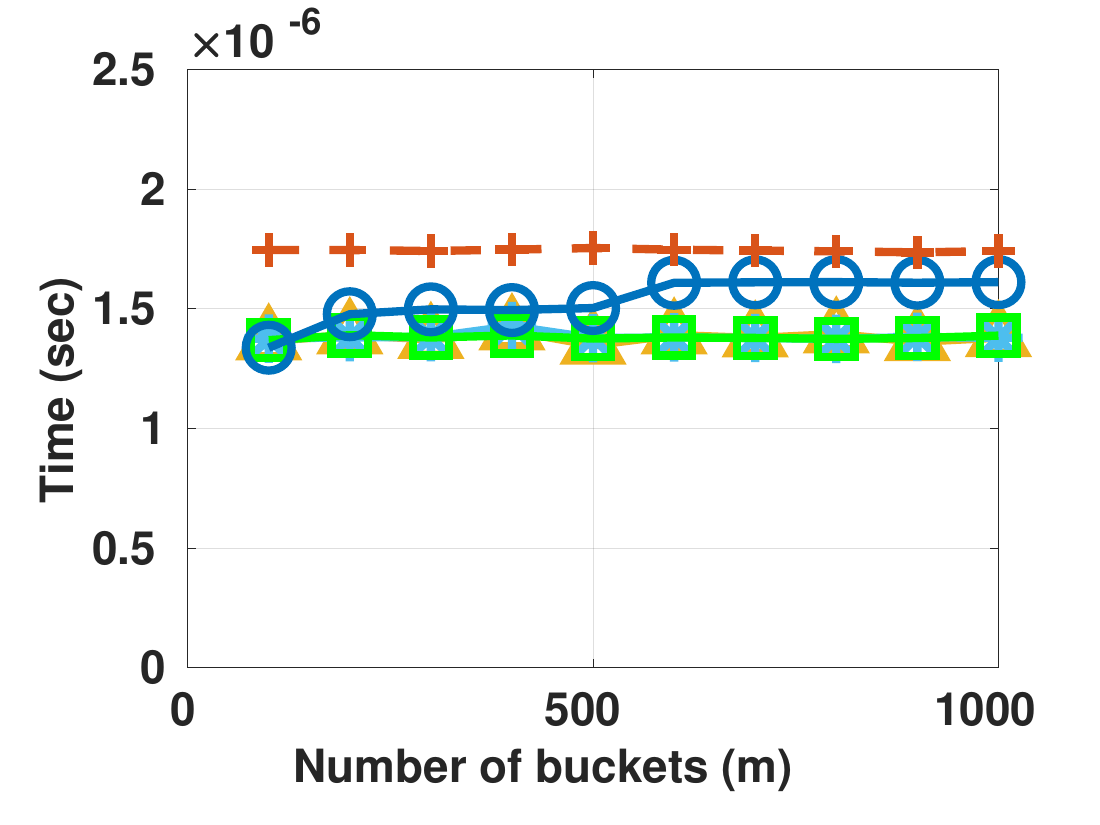}
        \caption{\adult}
    \end{subfigure}
    \hfill
    \begin{subfigure}[t]{0.24\textwidth}
        \centering
        \includegraphics[width=\textwidth]{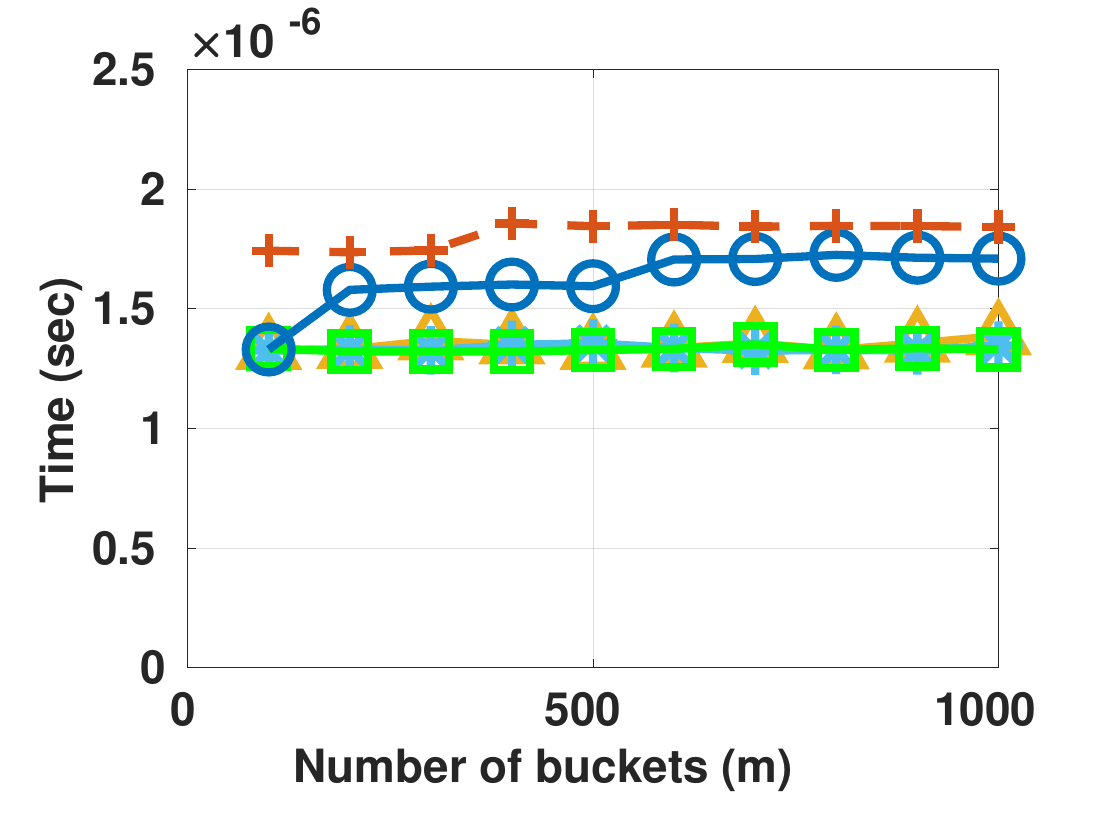}
        \caption{\compas}
    \end{subfigure}
    \hfill
    \begin{subfigure}[t]{0.24\textwidth}
        \centering
        \includegraphics[width=\textwidth]{plots/diabetes/diabetes_query_time_varying_number_of_buckets.pdf}
        \caption{\diabetes}
    \end{subfigure}
    \hfill
    \begin{subfigure}[t]{0.24\textwidth}
        \centering
        \includegraphics[width=\textwidth]{plots/popsim/popsim_query_time_varying_number_of_buckets.pdf}
        \caption{\popsim}
    \end{subfigure}
    \caption{Effect of varying number of buckets $m$ on query time}
\end{minipage}
\end{figure*}

\begin{figure*}[!tb]
    \begin{minipage}[t]{0.24\linewidth}
        \centering
        \includegraphics[width=\textwidth]{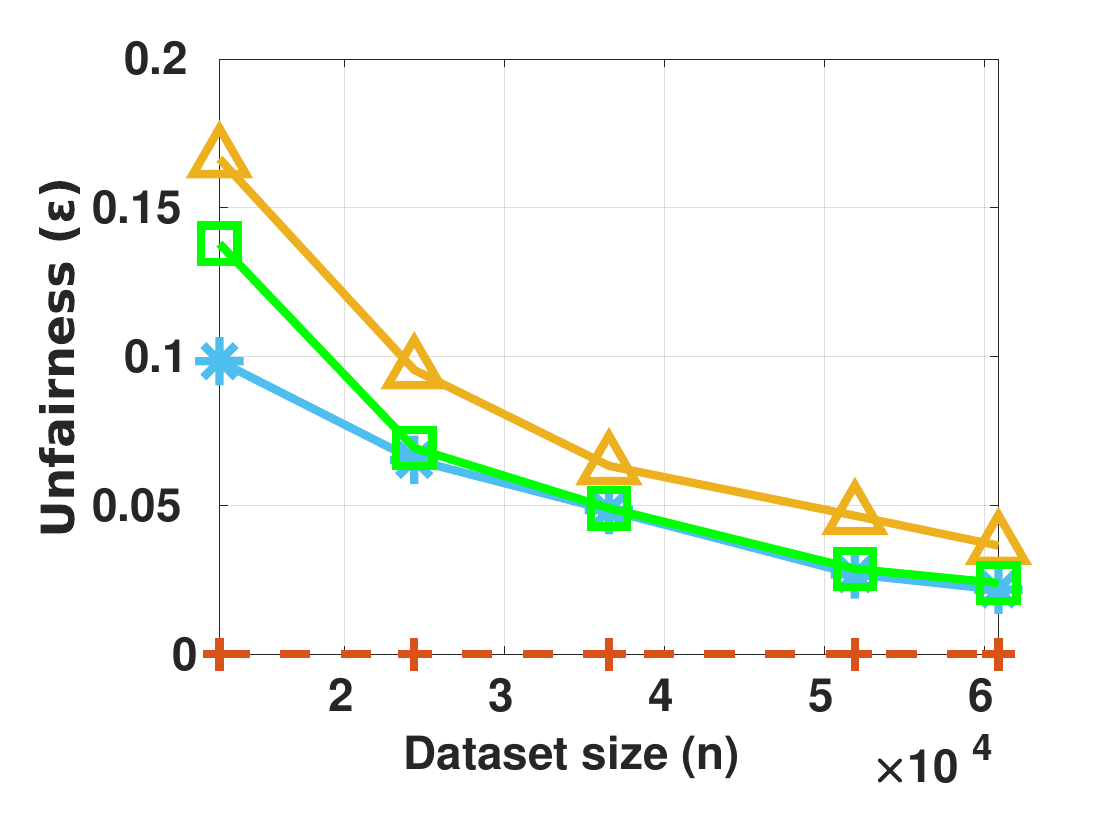}
        \vspace{-2.5em}
        \caption[]{Effect of varying dataset size $n$ on unfairness, \compas, {\tt race}}
        \vspace{-1em}
        \label{fig:compas_non_binary_n_vs_unfairness}
    \end{minipage}
    \hfill
    \begin{minipage}[t]{0.24\linewidth}
        \centering
        \includegraphics[width=\textwidth]{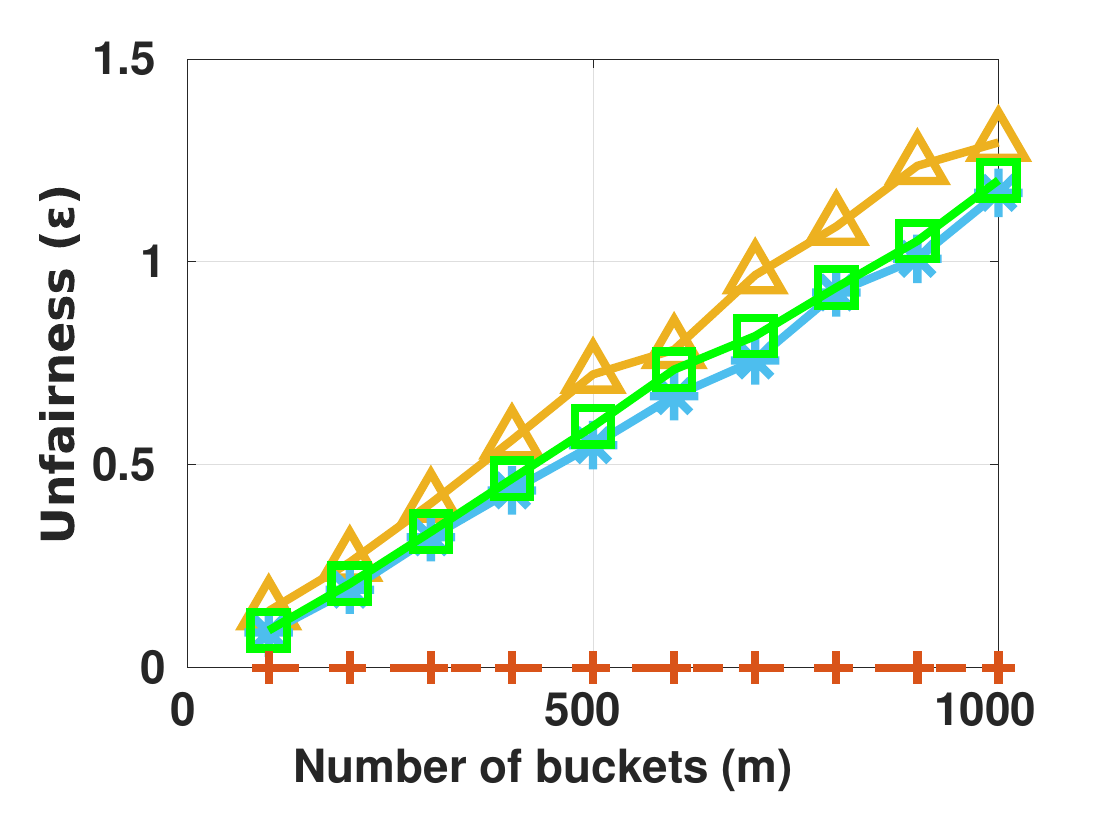}
        \vspace{-2.5em}
        \caption[]{Effect of varying number of buckets $m$ on unfairness, \compas, {\tt race}}
        \vspace{-1em}
        \label{fig:compas_non_binary_m_vs_unfairness}
    \end{minipage}
    \hfill
    \begin{minipage}[t]{0.24\linewidth}
        \centering
        \includegraphics[width=\textwidth]{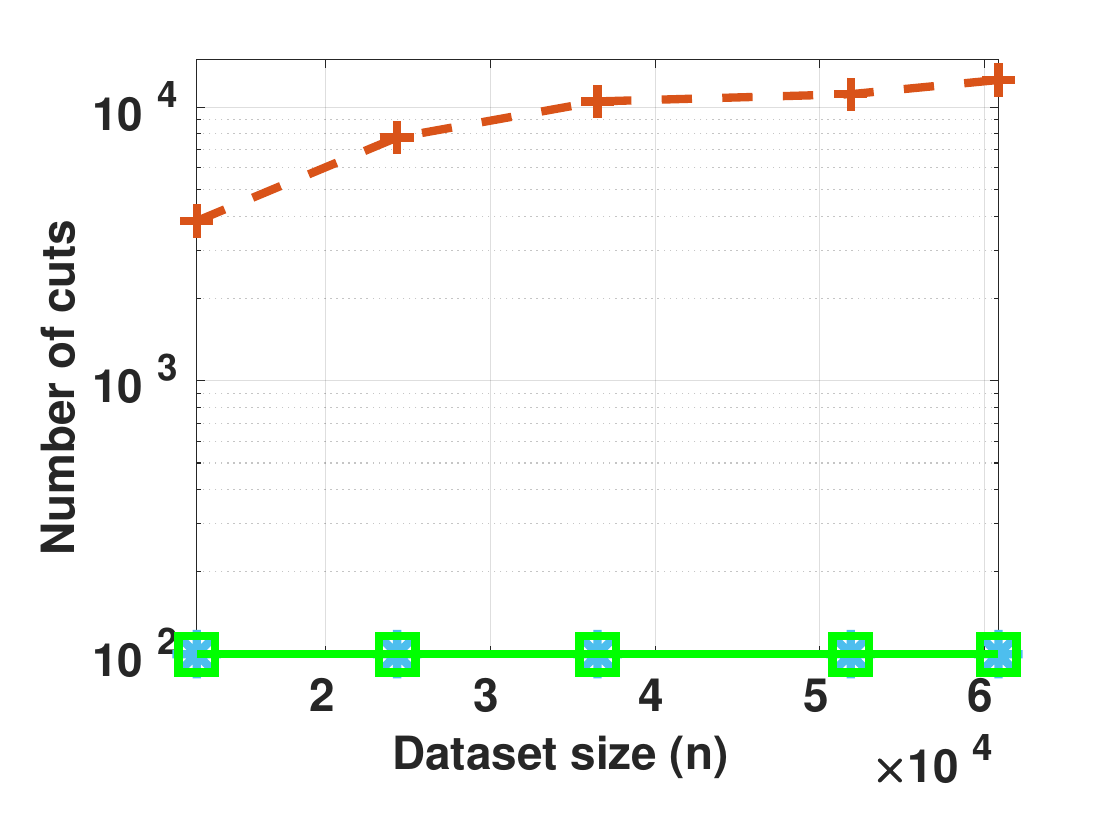}
        \vspace{-2.5em}
        \caption{Effect of varying dataset size $n$ on space, \compas, {\tt race}}
        \vspace{-1em}
        \label{fig:compas_non_binary_n_vs_space}
    \end{minipage}
    \hfill
    \begin{minipage}[t]{0.24\linewidth}
        \centering
        \includegraphics[width=\textwidth]{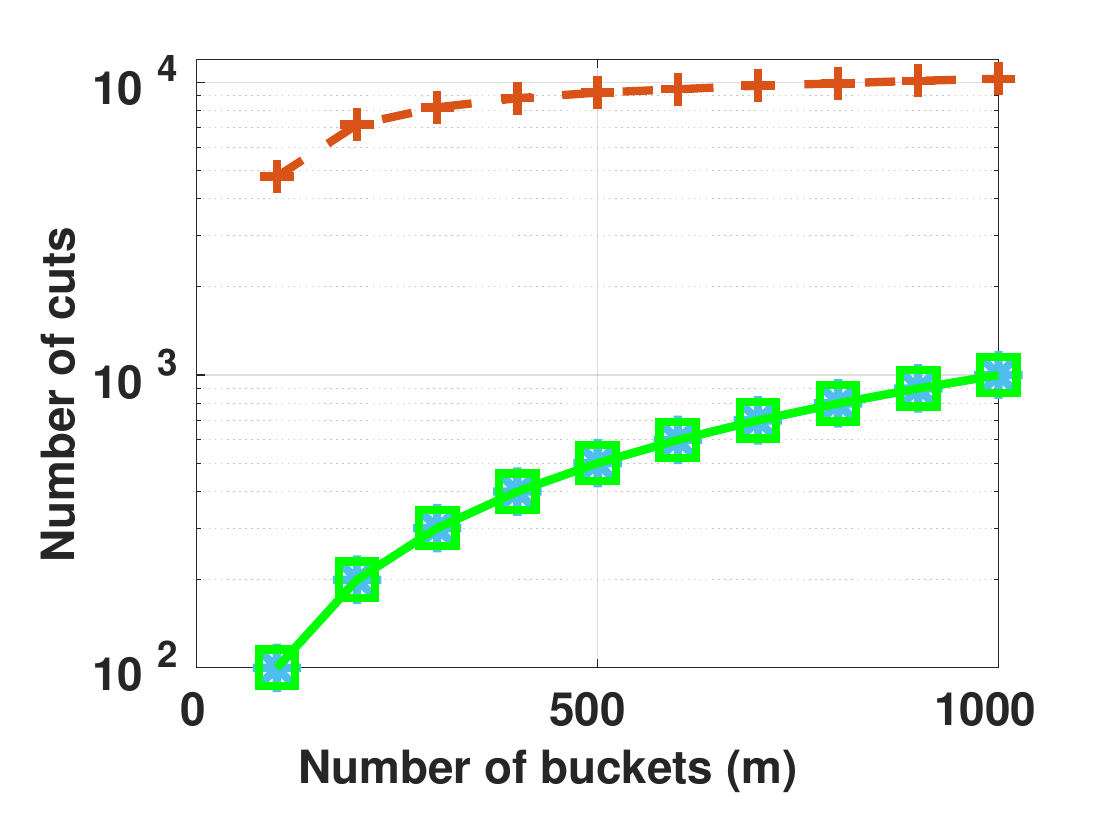}
        \vspace{-2.5em}
        \caption[]{Effect of varying number of buckets $m$ on space, \compas, {\tt race}}
        \vspace{-1em}
        \label{fig:compas_non_binary_m_vs_space}
    \end{minipage}
\end{figure*}

\begin{figure*}[!tb]
    \begin{minipage}[t]{0.24\linewidth}
        \centering
        \includegraphics[width=\textwidth]{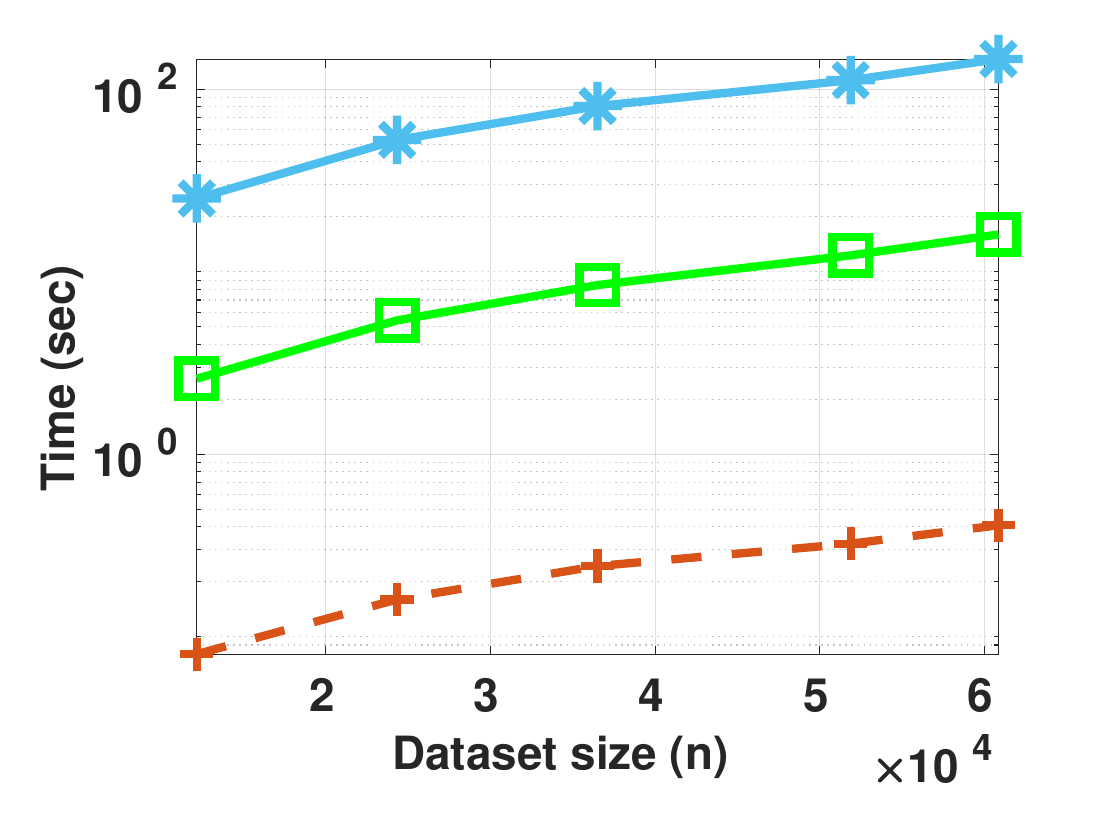}
        \vspace{-2.5em}
        \caption{Effect of varying dataset size $n$ on preprocessing time, \compas, {\tt race}}
        \vspace{-1em}
        \label{fig:compas_non_binary_n_vs_prep_time}
    \end{minipage}
    \hfill
    \begin{minipage}[t]{0.24\linewidth}
        \centering
        \includegraphics[width=\textwidth]{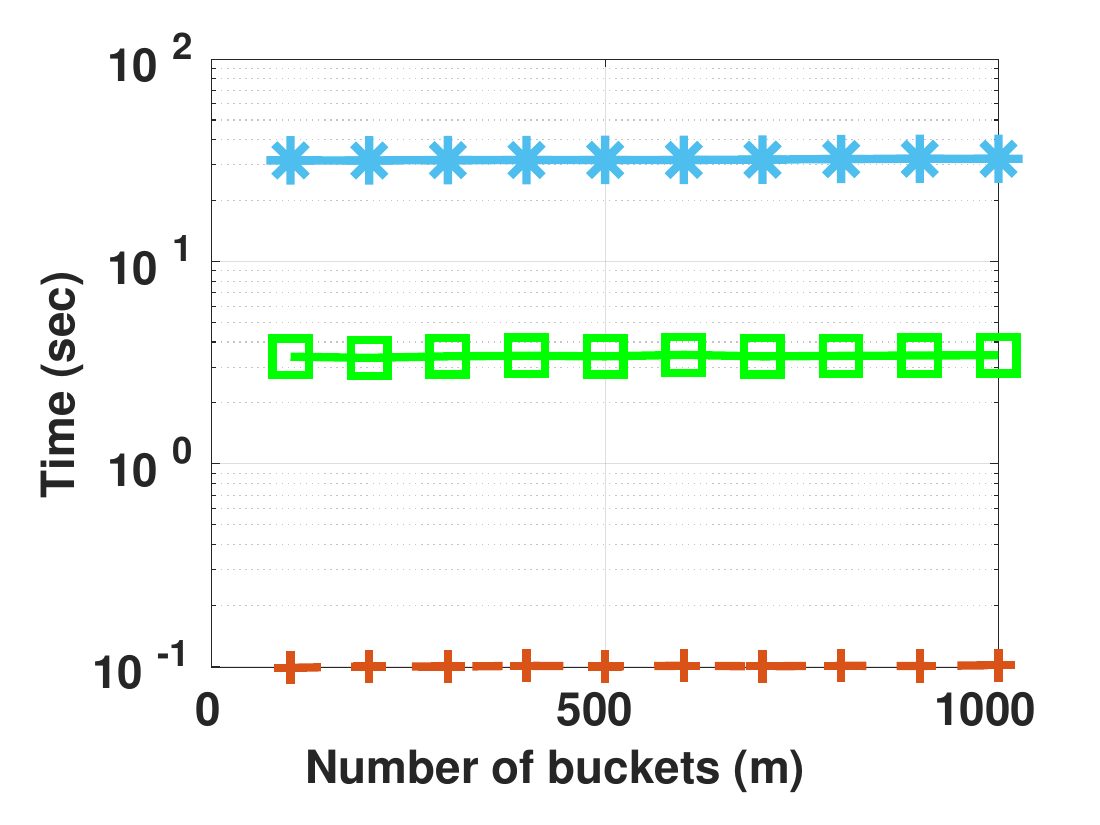}
        \vspace{-2.5em}
        \caption{Effect of varying number of buckets $m$ on preprocessing time, \compas, {\tt race}} 
        \vspace{-1em}
        \label{fig:compas_non_binary_m_vs_prep_time}
    \end{minipage}
    \hfill
    \begin{minipage}[t]{0.24\linewidth}
        \centering
        \includegraphics[width=\textwidth]{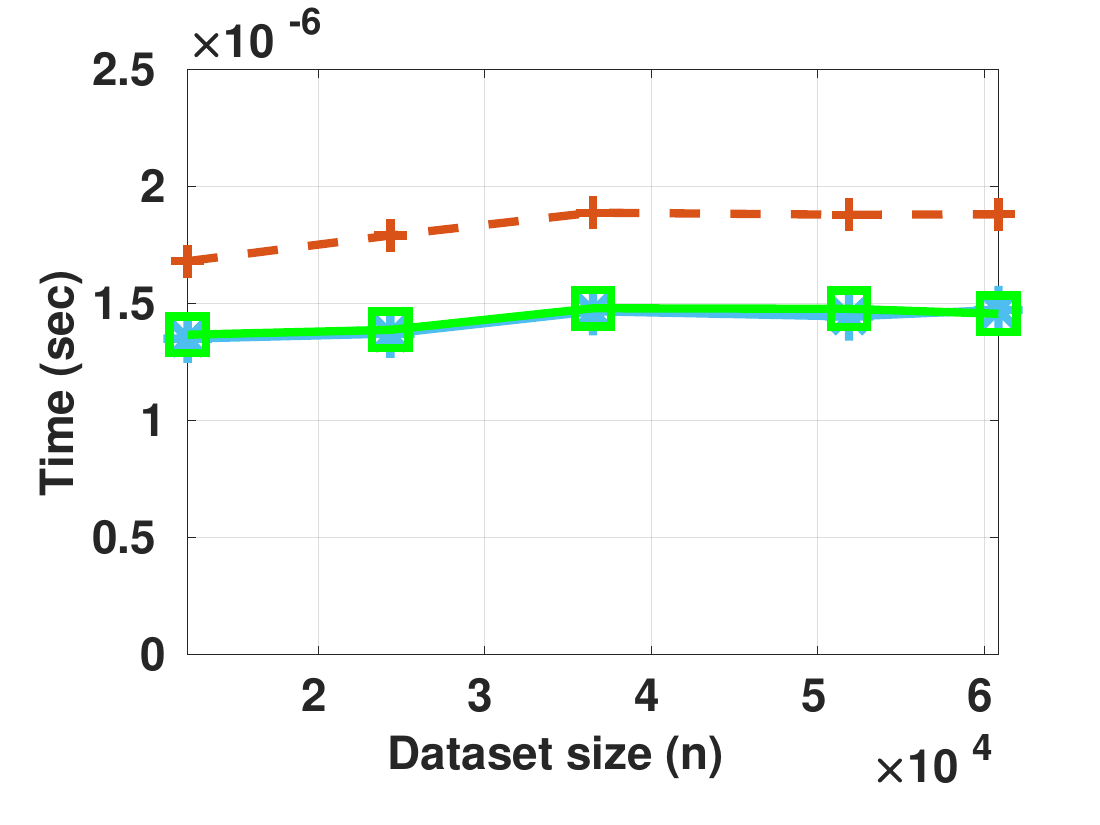}
        \vspace{-2.5em}
        \caption{Effect of varying dataset size $n$ on query time, \compas, {\tt race}}
        \vspace{-1em}
        \label{fig:compas_non_binary_n_vs_query_time}
    \end{minipage}
    \hfill
    \begin{minipage}[t]{0.24\linewidth}
        \centering
        \includegraphics[width=\textwidth]{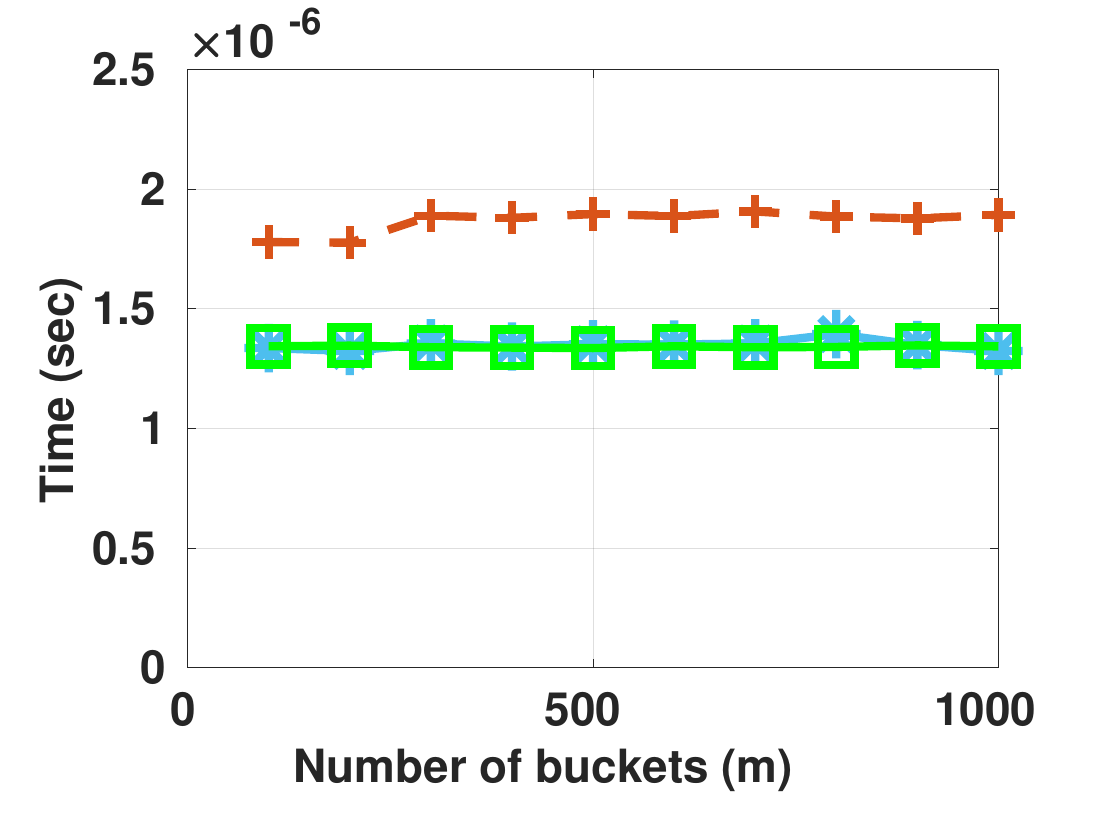}
        \vspace{-2.5em}
        \caption[]{Effect of varying number of buckets $m$ on query time, \compas, {\tt race}}
        \vspace{-1em}
        \label{fig:compas_non_binary_m_vs_query_time}
    \end{minipage}
\end{figure*}

\begin{figure*}[!tb]
    \begin{minipage}[t]{0.24\linewidth}
        \centering
        \includegraphics[width=\textwidth]{plots/adult_learned.pdf}
        \vspace{-2.5em}
        \caption{Learning setting: Unfairness evaluation over held out data, \adult} 
        \vspace{-1em}
    \end{minipage}
    \hfill
    \begin{minipage}[t]{0.24\linewidth}
        \centering
        \includegraphics[width=\textwidth]{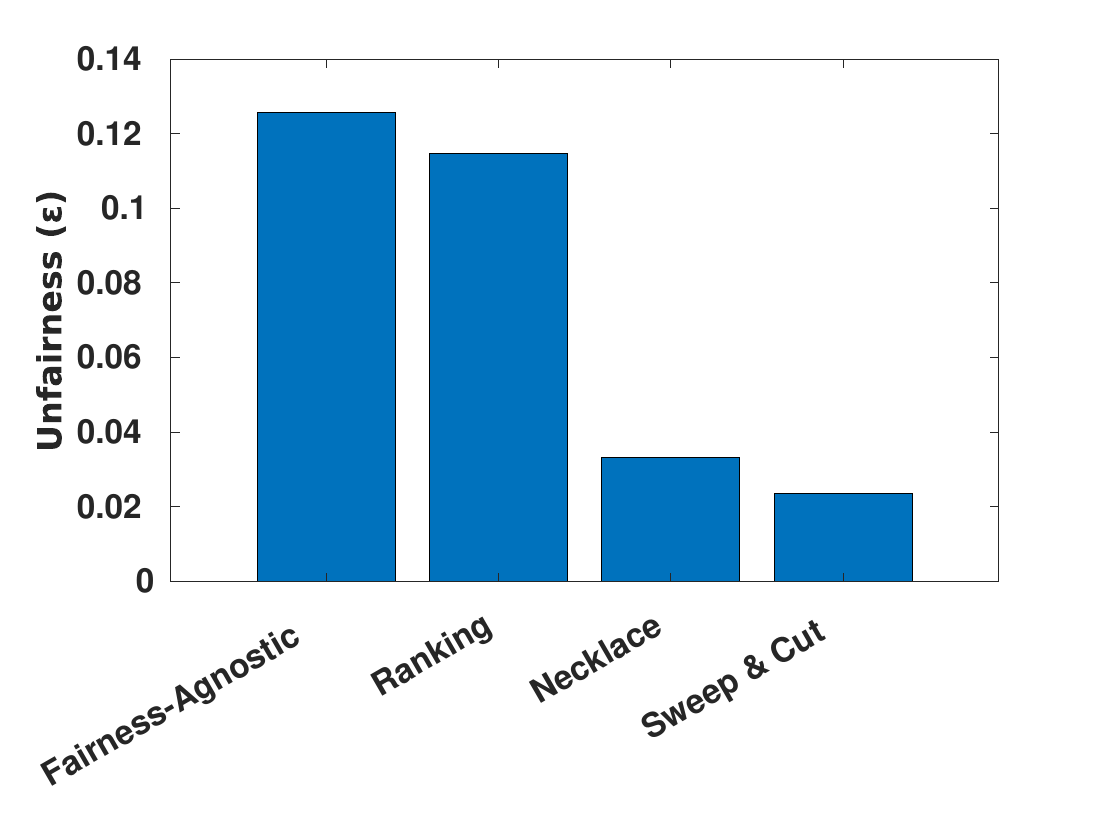}
        \vspace{-2.5em}
        \caption{Learning setting: Unfairness evaluation over held out data, \compas}
        \vspace{-1em}
    \end{minipage}
    \hfill
    \begin{minipage}[t]{0.24\linewidth}
        \centering
        \includegraphics[width=\textwidth]{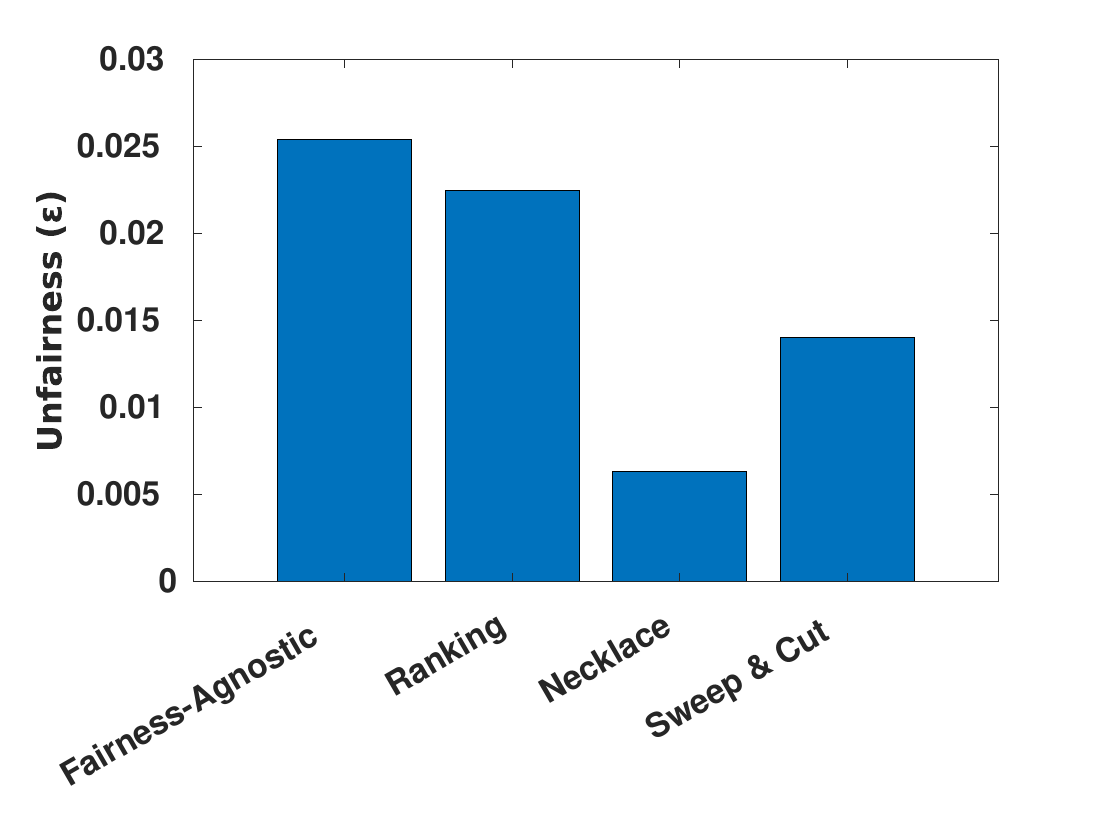}
        \vspace{-2.5em}
        \caption{Learning setting: Unfairness evaluation over held out data, \diabetes}
        \vspace{-1em}
    \end{minipage}
    \hfill
    \begin{minipage}[t]{0.24\linewidth}
        \centering
        \includegraphics[width=\textwidth]{plots/popsim_learned.pdf}
        \vspace{-2.5em}
        \caption{Learning setting: Unfairness evaluation over held out data, \popsim}
        \vspace{-1em}
    \end{minipage}
\end{figure*}

\begin{figure*}[!tb]
    \begin{minipage}[t]{0.24\linewidth}
        \centering
        \includegraphics[width=\textwidth]{plots/adult/local_search_adult.pdf}
        \vspace{-2.5em}
        \caption{Effect of local search based algorithm on unfairness, \adult}
        \vspace{-1em}
    \end{minipage}
    \hfill
    \begin{minipage}[t]{0.24\linewidth}
        \centering
        \includegraphics[width=\textwidth]{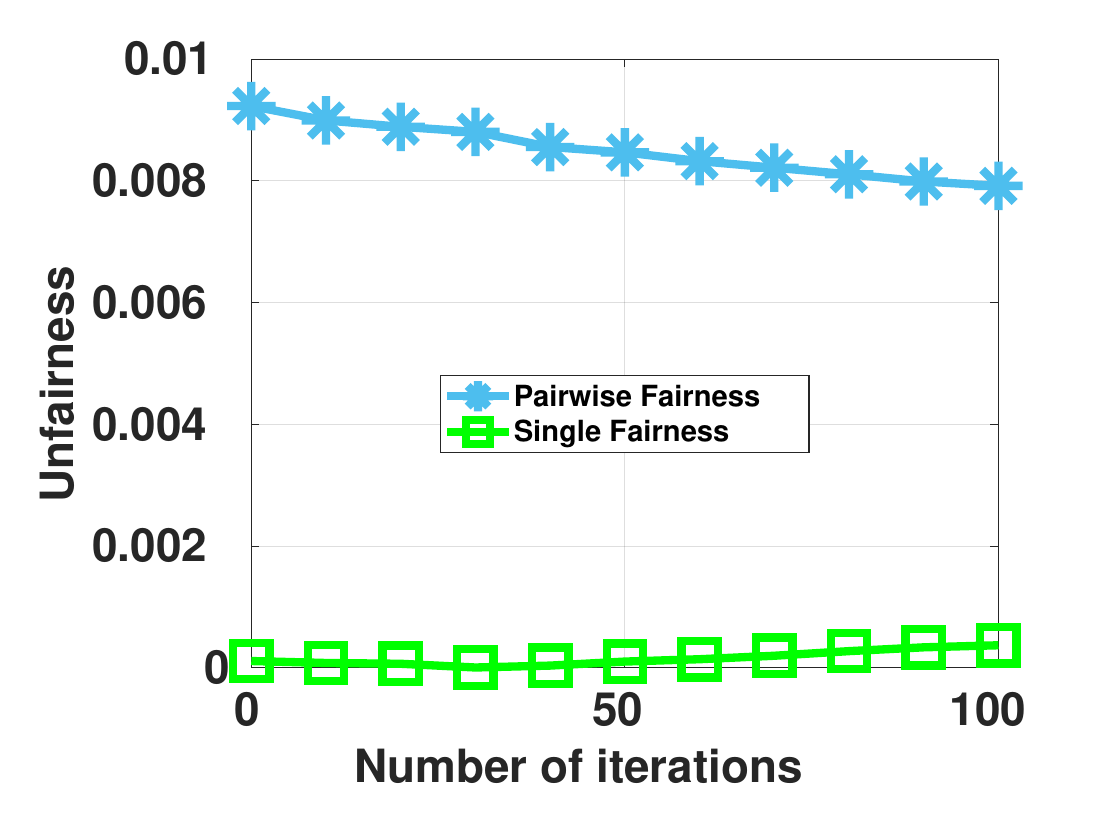}
        \vspace{-2.5em}
        \caption{Effect of local search based algorithm on unfairness, \compas}
        \vspace{-1em}
    \end{minipage}
    \hfill
    \begin{minipage}[t]{0.24\linewidth}
        \centering
        \includegraphics[width=\textwidth]{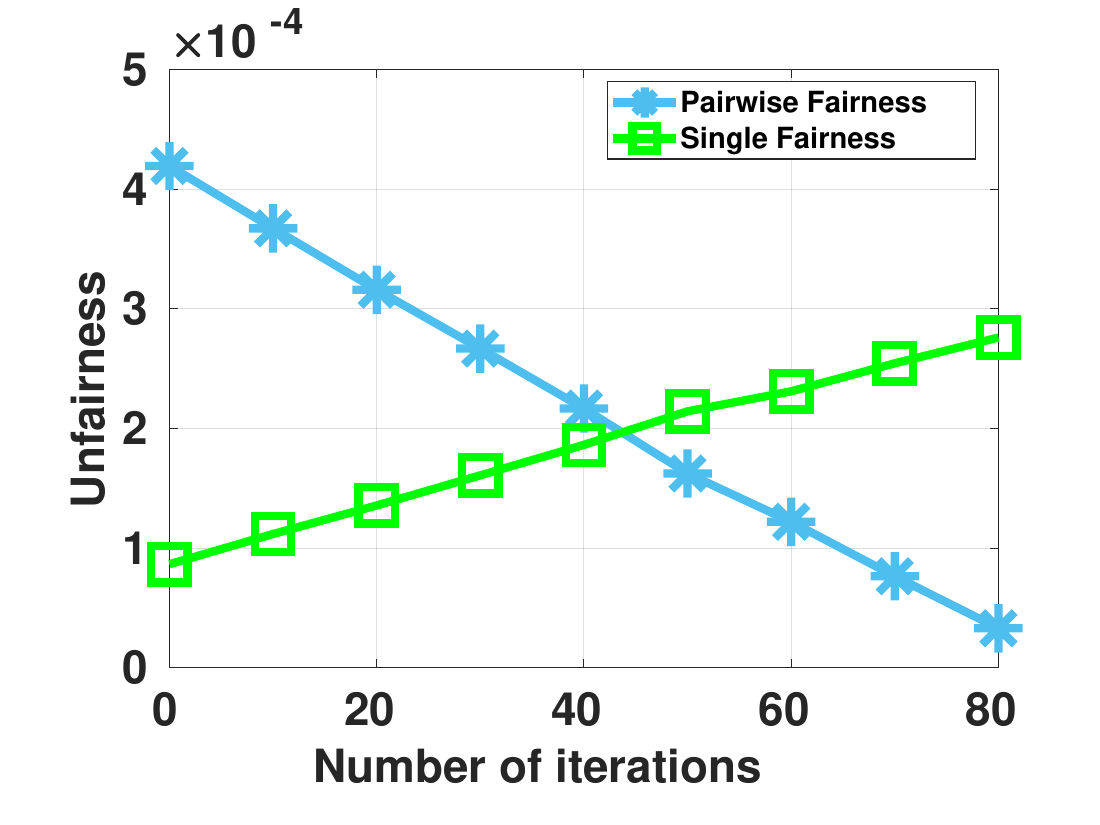}
        \vspace{-2.5em}
        \caption{Effect of local search based algorithm on unfairness, \diabetes}
        \vspace{-1em}
    \end{minipage}
    \hfill
    \begin{minipage}[t]{0.24\linewidth}
        \centering
        \includegraphics[width=\textwidth]{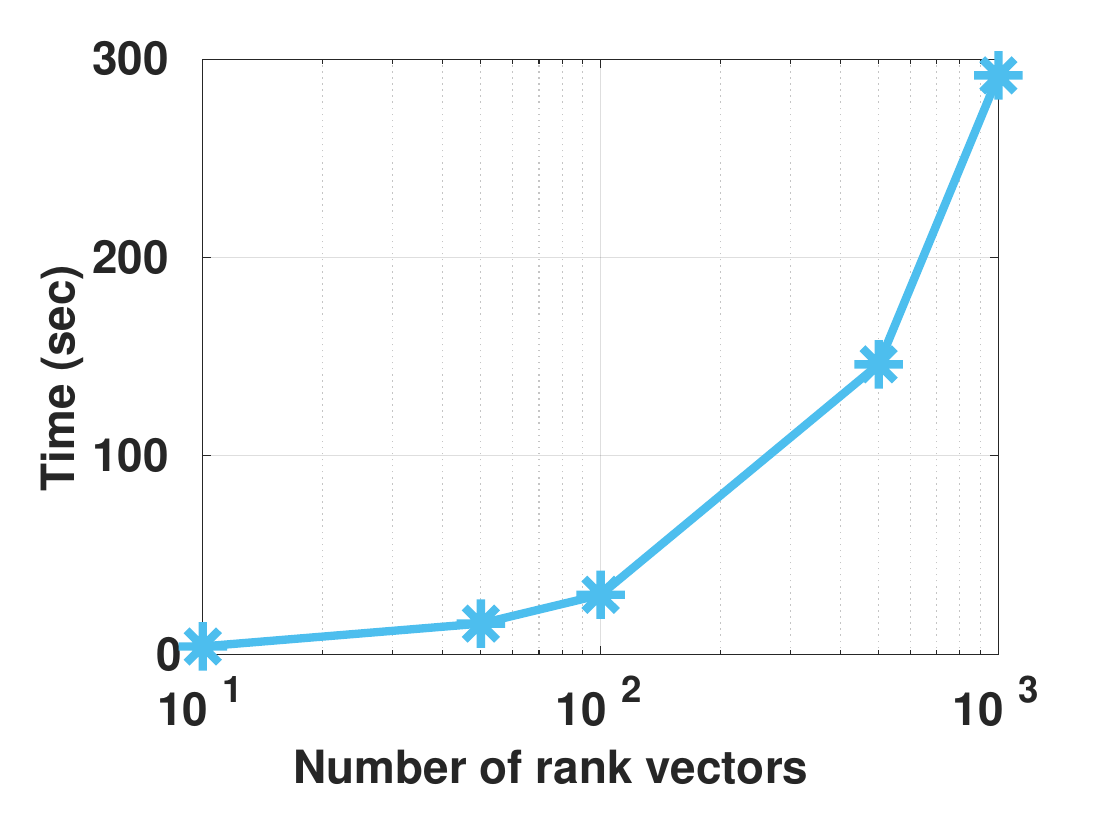}
        \vspace{-2.5em}
        \caption{Effect of varying number of sampled vectors on preprocessing time, \diabetes}
        \vspace{-1em}
    \end{minipage}
\end{figure*}

\begin{figure*}[!tb]
    \begin{minipage}[t]{0.24\linewidth}
        \centering
        \includegraphics[width=\textwidth]{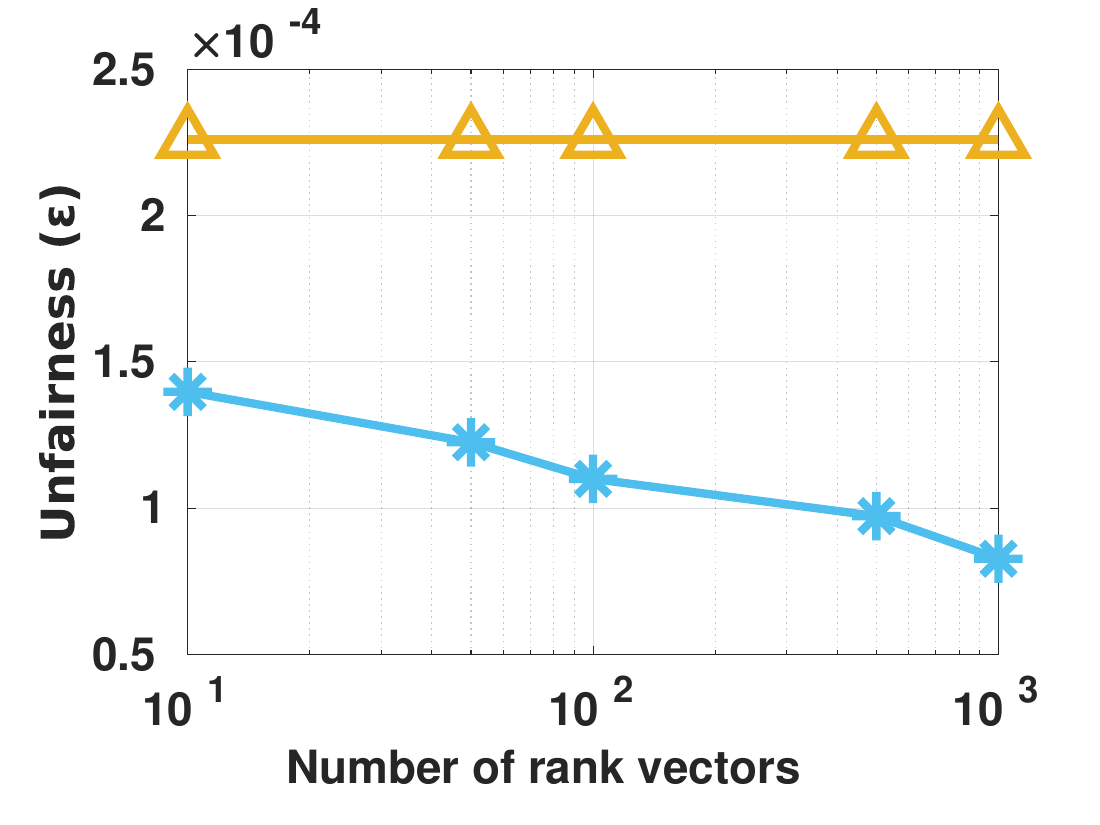}
        \vspace{-2.5em}
        \caption{Effect of varying number of sampled vectors on unfairness, \diabetes}
        \vspace{-1em}
    \end{minipage}
    \hfill
    \begin{minipage}[t]{0.24\linewidth}
        \centering
        \includegraphics[width=\textwidth]{plots/adult/adult_unfairness_varying_number_of_vectors.pdf}
        \vspace{-2.5em}
        \caption{Effect of varying number of sampled vectors on unfairness, \adult}
        \vspace{-1em}
    \end{minipage}
    \hfill
    \begin{minipage}[t]{0.24\linewidth}
        \centering
        \includegraphics[width=\textwidth]{plots/compas/compas_unfairness_varying_number_of_vectors.pdf}
        \vspace{-2.5em}
        \caption{Effect of varying number of sampled vectors on unfairness, \compas, {\tt sex}}
        \vspace{-1em}
    \end{minipage}
    \hfill
    \begin{minipage}[t]{0.24\linewidth}
        \centering
        \includegraphics[width=\textwidth]{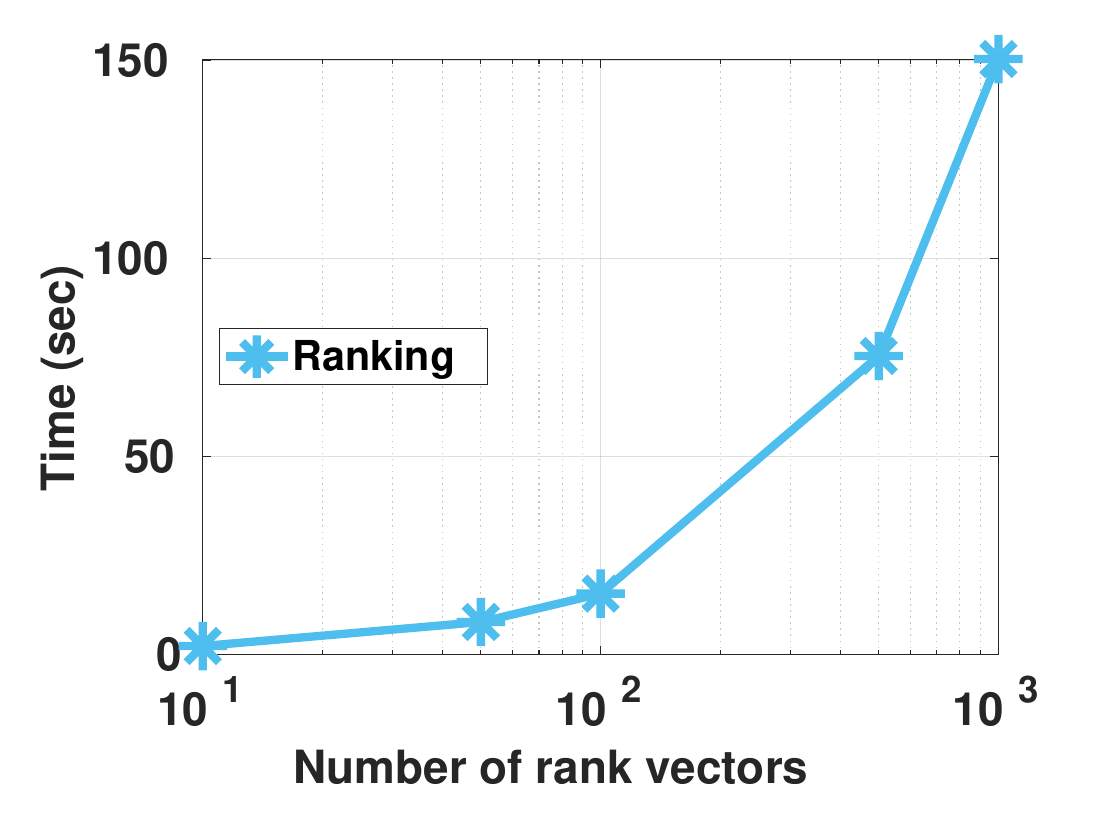}
        \vspace{-2.5em}
        \caption{Effect of varying number of sampled vectors on preprocessing time, \compas} 
        \vspace{-1em}
    \end{minipage}
\end{figure*}

\end{document}